\newif\iftodo\todofalse

\newif\ifecoop\ecooptrue
\newif\iftechreport\techreporttrue

\PassOptionsToPackage{usenames,dvipsnames}{xcolor}
\documentclass[a4paper,numberwithinsect,autoref,thm-restate,british]{lipics-v2021}
\newcommand{\toolName}{\texttt{NEST}\xspace}%

\usepackage{tikz}
\usetikzlibrary{shapes, arrows, positioning, automata}
\usetikzlibrary{decorations.pathmorphing}
\usetikzlibrary{shadows, calc}
\usepackage{wrapfig}
\usepackage{float}
\usepackage{array}
\usepackage{relsize}
\usepackage{colortbl}
\usepackage{enumitem}
\usepackage{listings}
\usepackage{thmtools}
\usepackage{svg}
\usepackage{arydshln}
\usepackage{cleveref}

\ifecoop\else
\newtheorem*{remark}{Remark}
\fi

\newcommand{\tx}[1]{\text{#1}}
\newcommand{\ttt}[1]{\texttt{#1}}
\newcommand{\tsc}[1]{\textsc{#1}}

\ifecoop
\else
\begin{abstract}
This paper introduces \toolName (Network-Enforced Session Types), a runtime verification framework that moves application-level protocol monitoring into the network fabric. Unlike prior work that instruments or wraps application code, we synthesize packet-level monitors that enforce protocols directly in the data plane. We develop algorithms to generate network-level monitors from session types and extend them to handle packet loss and reordering. We implement \toolName in P4 and evaluate it on applications including microservice and network-function models, showing that network-level monitors can enforce realistic non-trivial protocols.%
\end{abstract}
 \fi

\title{\toolName: Network Enforced Session Types\iftechreport\\ (Technical Report)\fi}

\titlerunning{\toolName: Network Enforced Session Types\iftechreport~(Technical Report)\fi}

\author{Jens Kanstrup Larsen}{Technical University of Denmark, Copenhagen, Denmark}{jekla@dtu.dk}{0009-0006-4039-3808}{}
\author{Alceste Scalas}{Technical University of Denmark, Copenhagen, Denmark}{alcsc@dtu.dk}{0000-0002-1153-6164}{}
\author{Guy Amir}{Cornell University, Ithaca, NY, USA}{gda42@cornell.edu}{0000-0002-7951-7795}{}
\author{Jules Jacobs}{Cornell University, Ithaca, NY, USA}{jj758@cornell.edu}{0000-0003-1976-3182}{}
\author{Jana Wagemaker}{Radboud University, Nijmegen, Netherlands}{jana.wagemaker@ru.nl}{0000-0002-8616-3905}{}
\author{Nate Foster}{Cornell University, Ithaca, NY, USA}{jnfoster@cs.cornell.edu}{0000-0002-6557-684X}{}

\authorrunning{J.~K.~Larsen, A.~Scalas, G.~Amir, J.~Jacobs, J.~Wagemaker, and N.~Foster}

\Copyright{Jens Kanstrup Larsen, Alceste Scalas, Guy Amir, Jules Jacobs, Jana Wagemaker, Nate Foster}


\usepackage{booktabs}   %
\usepackage{multirow}
\usepackage{subcaption} %

\usepackage[nomargin,inline,index,%
  status=draft %
]{fixme} %
\fxusetheme{colorsig}
\FXRegisterAuthor{fxAS}{anfxAS}{AS}%
\FXRegisterAuthor{fxNY}{anfxNY}{NY}%
\FXRegisterAuthor{fxPH}{anfxPH}{PH}

\usepackage{pifont}%

\usepackage{stmaryrd}%

\usepackage{xifthen}%
\usepackage{xspace}%

\usepackage{mathtools}%
\newif\ifdraft%
  \draftfalse%
  \drafttrue%

\newcommand{\ifempty}[3]{%
  \ifthenelse{\isempty{#1}}{#2}{#3}%
}%

\newcommand{\unfold}[1]{%
  {\color{black}\operatorname{unf}\!\left({#1}\right)}}%
\newcommand{\notImplies}{\mathrel{{\kern 0.5em}{\not{\kern -0.5em}\implies}}}%
\newcommand{\notImpliedBy}{\mathrel{{\kern 1em}{\not{\kern -1em}\impliedby}}}%
\newcommand{\subst}[3]{{#1}\!\left\{{#2} \mapsto {#3}\right\}}

\newcommand{\relR}{\mathrel{\mathcal{R}}}
\newcommand{\tbisim}{\mathrel{\overset{\tau}{\sim}}}

\newcommand{\coloncolonequals}{\Coloneqq}%
\newcommand{\bnfsep}{\mathbin{\;\big|\;}}%

\definecolor{ruleColor}{rgb}{0.1, 0.3, 0.1}%
\definecolor{roleColor}{rgb}{0.5, 0.0, 0.0}%
\newcommand{\roleCol}[1]{{\color{roleColor}#1}}%
\newcommand{\roleFmt}[1]{\boldsymbol{\roleCol{\mathtt{#1}}}}%
\newcommand{\roleP}[1][]{%
  \ifempty{#1}{{\color{roleColor}\roleFmt{p}}}{{\color{roleColor}\roleFmt{p}_{#1}}}%
}%
\newcommand{\roleQ}[1][]{%
  \ifempty{#1}{{\color{roleColor}\roleFmt{q}}}{{\color{roleColor}\roleFmt{q}_{#1}}}%
}%
\newcommand{\roleQi}[1][]{%
  \ifempty{#1}{{\color{roleColor}\roleFmt{q}'}}{{\color{roleColor}\roleFmt{q}'_{#1}}}%
}%
\newcommand{\roleR}[1][]{%
  \ifempty{#1}{{\color{roleColor}\roleFmt{r}}}{{\color{roleColor}\roleFmt{r}_{\!#1}}}%
}%
\newcommand{\labFmt}[2][]{\ifempty{#1}{\mathtt{#2}}{\mathtt{#2}\textsubscript{#1}}}%

\newcommand{\tyFmt}[1]{{\color{black}#1}}%
\newcommand{\tyBool}{\tyFmt{\operatorname{Bool}}}%
\newcommand{\tyInt}{\tyFmt{\operatorname{Int}}}%
\newcommand{\tyString}{\tyFmt{\operatorname{Str}}}%

\definecolor{gtColor}{rgb}{0.43, 0.21, 0.1}%
\definecolor{stColor}{rgb}{0, 0, 0.9}%
\newcommand{\stFmt}[1]{{\color{stColor}#1}}%

\newcommand{\inTag}{\stFmt{\&}}%
\newcommand{\outTag}{\stFmt{\oplus}}%
\newcommand{\enqueueTag}{\stFmt{?}}%
\newcommand{\stIn}[3]{\ifempty{#1}{}{\roleFmt{#1}}\stFmt{\&{#2}\ifempty{#3}{}{({#3})}}}%
\newcommand{\stOut}[3]{\ifempty{#1}{}{\roleFmt{#1}}\stFmt{\oplus{#2}\ifempty{#3}{}{({#3})}}}%

\newcommand{\stChoice}[2]{\stFmt{#1}\ifempty{#2}{}{\stFmt{({#2})}}}%

\newcommand{\stSeq}{\mathbin{\!\stFmt{.}\!}}%
\newcommand{\stIntC}{\mathbin{\stFmt{\oplus}}}%
\newcommand{\stIntSum}[3]{\stFmt{{\sum}_{#1}{#2}{\oplus}{#3}}}%
\newcommand{\stExtC}{\mathbin{\stFmt{\&}}}%
\newcommand{\stExtSum}[3]{\stFmt{\roleFmt{#2}{\mathlarger{\mathlarger{\mathlarger{\&}}}}_{#1}{#3}}}%
\newcommand{\stRec}[2]{\stFmt{\mu{#1}.{#2}}}%
\newcommand{\stEnd}{\stFmt{\mathbf{end}}}%

\newcommand{\stLab}[1][]{\stFmt{\ifempty{#1}{\labFmt{m}}{\labFmt{m}_{#1}}}}%
\newcommand{\stLabi}[1][]{\stFmt{\ifempty{#1}{\labFmt{m}'}{\labFmt{m}'_{#1}}}}%
\newcommand{\stLabFmt}[1]{\stFmt{\labFmt{#1}}}%

\newcommand{\stS}[1][]{\stFmt{\ifempty{#1}{S}{S_{#1}}}}%
\newcommand{\stSi}[1][]{\stFmt{\ifempty{#1}{S'}{S'_{#1}}}}%

\newcommand{\stT}[1][]{\stFmt{\ifempty{#1}{T}{T_{#1}}}}%
\newcommand{\stTi}[1][]{\stFmt{\ifempty{#1}{T'}{T'_{#1}}}}%
\newcommand{\stTii}[1][]{\stFmt{\ifempty{#1}{T''}{T''_{#1}}}}%
\newcommand{\stTiii}[1][]{\stFmt{\ifempty{#1}{T'''}{T'''_{#1}}}}%
\newcommand{\stTiiii}[1][]{\stFmt{\ifempty{#1}{T''''}{T''''_{#1}}}}%
\newcommand{\stTiiiii}[1][]{\stFmt{\ifempty{#1}{T'''''}{T'''''_{#1}}}}%

\newcommand{\stRecVarBase}{\stFmt{\mathbf{t}}}%
\newcommand{\stRecVar}[1][]{\stFmt{\ifempty{#1}{\stRecVarBase}{\stRecVarBase_{#1}}}}%
\newcommand{\stRecVari}[1][]{\stFmt{\ifempty{#1}{\stRecVar'}{\stRecVar'_{#1}}}}%
\newcommand{\stQ}[1][]{\stFmt{\ifempty{#1}{\sigma}{\sigma_{#1}}}}%
\newcommand{\stQi}[1][]{\stFmt{\ifempty{#1}{\sigma'}{\sigma'_{#1}}}}%
\newcommand{\stQii}[1][]{\stFmt{\ifempty{#1}{\sigma''}{\sigma''_{#1}}}}%
\newcommand{\stQiii}[1][]{\stFmt{\ifempty{#1}{\sigma'''}{\sigma'''_{#1}}}}%
\newcommand{\stQEmpty}{\stFmt{\epsilon}}%
\newcommand{\stQCons}[2]{\stFmt{{#1}\mathbin{\!\cdot\!}{#2}}}%
\newcommand{\stQConsWide}[2]{\stFmt{{#1}\mathbin{\cdot}{#2}}}%
\newcommand{\stQMsg}[3]{\stFmt{{#1}\mathbin{\!\mathbf{\triangleright}\!}{#2}\ifempty{#3}{}{({#3})}}}%
\newcommand{\stQMsgWide}[3]{\stFmt{{#1}\mathbin{\mathbf{\triangleright}}{#2}\ifempty{#3}{}{({#3})}}}%
\newcommand{\stqT}[1][]{\stFmt{\ifempty{#1}{\ddot{T}}{\ddot{T}_{#1}}}}%
\newcommand{\stqTi}[1][]{\stFmt{\ifempty{#1}{\ddot{T}'}{\ddot{T}'_{#1}}}}%
\newcommand{\stqPair}[2]{\stFmt{{#1}\mathbin{{\langle}\!{|}}{#2}}}%

\newcommand{\transition}[3]{{#1}\mathrel{\xrightarrow{#2}}{#3}}

\newcommand{\transitionS}[3]{{#1}\mathrel{\xrightarrow{#2}{}^{\!\!*}}{#3}}
\newcommand{\noTransition}[2]{{#1}\xrightarrow{#2}\hspace{-10pt}\not}
\newcommand{\stTagNoQueue}{\stFmt{\alpha}}

\newcommand{\stIntTag}[1]{\stFmt{\tau({#1})}}
\newcommand{\stExtTag}{\stFmt{\gamma}}

\newcommand{\stOutTag}[3]{\stFmt{{#1} \outTag {#2} ({#3})}}
\newcommand{\stInTag}[3]{\stFmt{{#1} \inTag {#2} ({#3})}}
\newcommand{\stEnqueueTag}[3]{\stFmt{{#1} ? #2 \ifempty{#3}{}{({#3})}}}
\newcommand{\stSendTag}[3]{\stFmt{{#1} ! #2 \ifempty{#3}{}{({#3})}}}
\newcommand{\dequeueTag}[3]{\stFmt{{#1} \inTag {#2} ({#3})}}

\newcommand{\xmark}{\text{\ding{56}}}%
\definecolor{monColor}{rgb}{0.21, 0.5 , 0.1}%
\newcommand{\monFmt}[1]{{\color{monColor}#1}}%
\newcommand{\monM}[1][]{\monFmt{\ifempty{#1}{M}{M_{#1}}}}
\newcommand{\monMi}[1][]{\monFmt{\ifempty{#1}{M'}{M_{#1}'}}}

\newcommand{\monLabYes}[1]{\monFmt{#1\,\checkmark}}
\newcommand{\monLabNo}[1]{\monFmt{#1\,{\color{red}\xmark}}}
\newcommand{\monLabSilent}[1]{\monFmt{\tau({#1})}}

\newcommand{\monT}[1][]{\monFmt{\ifempty{#1}{\widehat{\stFmt{T}}}{\widehat{\stFmt{T}}_{\stFmt{#1}}}}}
\newcommand{\monTi}[1][]{\monFmt{\ifempty{#1}{\widehat{\stFmt{T'}}}{\widehat{\stFmt{T'}}_{\stFmt{#1}}}}}

\newcommand{\mstPair}[2]{\monFmt{\left\lceil \stFmt{#1} \right\rceil} \monFmt{#2}}

\newcommand{\enc}[1]{\monFmt{\left\llbracket{\stFmt{#1}}\right\rrbracket}}

\newcommand{\cpmPar}[2]{\monFmt{#1 \mathbin{\parallel} #2}}
\newcommand{\cpmRst}[3]{\monFmt{(\nu\,#1,\hspace{-1.5pt}#2)#3}}

\definecolor{netColor}{rgb}{0.6, 0.4, 0.8}%
\newcommand{\netFmt}[1]{{\color{netColor}#1}}%
\newcommand{\netN}[1][]{\netFmt{\ifempty{#1}{N}{N_{#1}}}}
\newcommand{\netNi}[1][]{\netFmt{\ifempty{#1}{N'}{N'_{#1}}}}
\newcommand{\netNii}[1][]{\netFmt{\ifempty{#1}{N''}{N''_{#1}}}}

\newcommand{\netPar}{\mathbin{\netFmt{\parallel}}}
\newcommand{\netBigPar}[2]{\netFmt{\mathop{\parallel}\limits_{#1}{#2}}}
\newcommand{\netPair}[2]{\netFmt{{#1}:{#2}}}
\newcommand{\netLab}{\netFmt{\delta}}
\newcommand{\netIOLab}[2]{\netFmt{{#1}:{#2}}}
\newcommand{\netDequeueLab}[4]{\netFmt{\tau({#1}:\dequeueTag{#2}{#3}{#4})}}
\newcommand{\netCommLab}[4]{\netFmt{\tau({#1} \rightarrow {#2}:\stFmt{{#3}(#4)})}}
\newcommand{\netTau}[1][]{\netFmt{\ifempty{#1}{\tau}{\tau({#1})}}}

\definecolor{mnetColor}{rgb}{0.29, 0.33, 0.13}%
\newcommand{\mnetFmt}[1]{{\color{mnetColor}#1}}%
\newcommand{\mnetN}[1][]{\mnetFmt{\ifempty{#1}{\widehat{N}}{\widehat{N}_{#1}}}}
\newcommand{\mnetNi}[1][]{\mnetFmt{\ifempty{#1}{\widehat{N}'}{\widehat{N}'_{#1}}}}
\newcommand{\mnetNii}[1][]{\mnetFmt{\ifempty{#1}{\widehat{N}''}{\widehat{N}''_{#1}}}}
\newcommand{\mnetNiii}[1][]{\mnetFmt{\ifempty{#1}{\widehat{N}'''}{\widehat{N}'''_{#1}}}}

\newcommand{\mnetPar}{\mathbin{\mnetFmt{\parallel}}}
\newcommand{\mnetBigPar}[2]{\mnetFmt{\mathop{\parallel}\limits_{#1}{#2}}}
\newcommand{\mnetPair}[2]{\mnetFmt{{#1}:{#2}}}
\newcommand{\mnetIOLabYes}[2]{\mnetFmt{{#1}:{#2}}{\color{monColor}\checkmark}}
\newcommand{\mnetIOLabNo}[2]{\mnetFmt{{#1}:{#2}}{\color{red}\xmark}}
\newcommand{\mnetDequeueLab}[4]{\mnetFmt{\tau({#1}:\dequeueTag{#2}{#3}{#4})}}
\newcommand{\mnetCommLab}[4]{\mnetFmt{\tau({#1} \rightarrow {#2}:\stFmt{{#3}(#4)})}}

\newcommand{\mnetLab}{\mnetFmt{\eta}}

\newcommand{\mnetInstr}[1]{\mnetFmt{\operatorname{mon}\!\left({#1}\right)}}

\newcommand{\internalLab}[1]{\tau({#1})} %
\lstdefinelanguage{py}{
    keywords={
        and, as, assert, break, class, continue, def, del, elif, else,
        except, False, finally, for, from, global, if, import, in, is,
        lambda, None, nonlocal, not, or, pass, raise, return, True, try,
        while, with, yield
    },
    morekeywords=[2]{BookReviewSessionHandler, sendMsg, recvMsg},
    morekeywords=[3]{Client, Info, Review, Ratings, Details},
    morekeywords=[4]{request, review_request, ratings_request,
    ratings_response, review_response, detail_request, detail_response, response},
    sensitive=true,
    morecomment=[l]{\#},
    morecomment=[s]{'''}{'''},
    morecomment=[s]{"""}{"""},
    morestring=[b]',
    morestring=[b]",
    stringstyle=\color{olive},
    commentstyle=\color{gray},
    keywordstyle=\color{blue},
    keywordstyle=[2]\color{green!50!black},
    showstringspaces=false,
    basicstyle=\small\ttfamily,
} 

\lstdefinelanguage{scala}{
    keywords={
        abstract, case, catch, class, def, do, else, extends, false, final,
        finally, for, if, implicit, import, lazy, match, new, null, object,
        override, package, private, protected, return, sealed, super, this,
        throw, trait, true, try, type, val, var, while, with, yield
    },
    morekeywords=[2]{=>, Map, Rec, ExtChoice, IntChoice, Var},
    morekeywords=[3]{a_BookClient, a_BookReview, a_BookDetails},
    morekeywords=[4]{l_Request, l_ReviewRequest, l_DetailRequest,
    l_ReviewResponse, l_DetailResponse, l_Response},
    sensitive=true,
    morecomment=[l]{//},
    morecomment=[s]{/*}{*/},
    morestring=[b]",
    morestring=[b]',
    stringstyle=\color{olive},
    commentstyle=\color{gray}, 
    keywordstyle=\color{blue},
    keywordstyle=[2]\color{green!50!black},
    keywordstyle=[3]\color{purple},
    keywordstyle=[4]\color{teal},
    showstringspaces=false,
    basicstyle=\small\ttfamily
}

\lstdefinelanguage{Scribble}{
    morekeywords={
        protocol, role, to, from, function, val
    },
    morekeywords=[2]{
        ->,
    },
    morecomment=[l]{//},
    morecomment=[s]{/*}{*/},
    morestring=[b]",
    morekeywords=[3]{Client, Info, Review, Ratings, Details},
    morekeywords=[4]{request, review_request, ratings_request,
    ratings_response, review_response, detail_request, detail_response,
    response}
}
\lstdefinelanguage{Scribble2}{
    morekeywords={
        protocol, role, to, from, function, val
    },
    morekeywords=[2]{
        ->,
    },
    morecomment=[l]{//},
    morecomment=[s]{/*}{*/},
    morestring=[b]",
    morekeywords=[3]{Client, Info, Review, Ratings, Details}
}

\lstdefinelanguage{P4}{%
    keywords = [1]{control, extern, typedef, action, struct, table},
    keywordstyle = [1]\color{blue},
    keywords = [2]{in, out, inout, bit, key, actions, size, default_action},
    keywordstyle = [2]\color{OliveGreen},
    keywords = [3]{test, ig, simple_header, nested_header, Headers, exact},
    keywordstyle = [3]\color{violet},
    literate=*
    {0}{{{\color{purple}0}}}{1}
    {1}{{{\color{purple}1}}}{1}
    {2}{{{\color{purple}2}}}{1}
    {3}{{{\color{purple}3}}}{1}
    {4}{{{\color{purple}4}}}{1}
    {7}{{{\color{purple}7}}}{1}
    {8}{{{\color{purple}8}}}{1}
    {16}{{{\color{purple}16}}}{1},
    morekeywords={%
      action, algorithm, apply, attributes, bytes,
      calculated_field, control, counter, current, default, direct,
      drop, else, false, field_list, field_list_calculation, fields,
      header, header_type, hit, if, input, instance_count, last, latest,
      layout, length, mask, max_length, metadata, meter, min_width, miss
      output_width, packets, parse_error, parser, parser_exception, payload,
      primitive_action, register, result, return, saturating, select,
      signed, static, switch, true, type, update, valid, verify, width},
    sensitive,
}[keywords,comments]%

\lstdefinestyle{P4Style}{ %
  basicstyle=\small\ttfamily,
  breaklines=true,                 %
  commentstyle=\color{commentsColor}\textit,    %
  deletekeywords={...},            %
  escapeinside={\%*}{*)},          %
  extendedchars=true,              %
  frame=bottomline,	                   	   %
  keepspaces=true,                 %
  keywordstyle=\color{keywordsColor}\bfseries,       %
  language=P4,                 %
  otherkeywords={*,...},           %
  numbers=left,                    %
  numberstyle=\tiny\color{gray},
  stepnumber=1,
  frame=single,
  xleftmargin=2.5em,
  framexleftmargin=2em,
  escapechar=\%
  stringstyle=\color{stringColor}, %
}
\usepackage[textsize=scriptsize,colorinlistoftodos,\iftodo{}\else{disable}\fi]{todonotes}

\iftodo
  \newcommand{\change}[2]{\todo[color=blue!30,prepend,size=\tiny]{\textsf{#1}}{{\color{blue}#2}}}
\else
  \newcommand{\change}[2]{#2}
\fi

\newcommand{\derive}[2]{\genfrac{}{}{0.5pt}{0}{\begin{gathered}#1\end{gathered}}{#2}}
\newcommand{\Rule}[3]{\derive{\begin{gathered}#1\end{gathered}}{#2}\ \textsc{#3}}

\newcommand{\thesisNewline}{}

\crefformat{section}{\S#2#1#3}
\Crefformat{section}{\S#2#1#3}

\crefformat{subsection}{\S#2#1#3}
\Crefformat{subsection}{\S#2#1#3}
\crefformat{subsubsection}{\S#2#1#3}
\Crefformat{subsubsection}{\S#2#1#3}

\Crefname{figure}{Fig.{}}{Figures}
\crefname{figure}{Fig.{}}{Figures}
\Crefname{definition}{Def.{}}{Definitions}
\crefname{definition}{Def.{}}{Definitions}

\AtEndEnvironment{example}{\null\hfill\ensuremath{\blacklozenge}}
\AtEndEnvironment{definition}{\null\hfill\ensuremath{\lozenge}}

\iftechreport
	\relatedversiondetails{ECOOP'26 paper}{https://doi.org/10.4230/LIPIcs.ECOOP.2026.31}
\fi

\funding{%
  Research partially supported by: %
  the DTU Nordic Five Tech Alliance grant ``Safe and secure software-defined networks in P4'';
  the Horizon Europe grant no.~101093006 ``TaRDIS''; %
  the Independent Research Fund Denmark project ``Hyben''; %
  and the Dutch research council (NWO) under grant no.\ VI.Veni.242.134 (VerHyp).
  The work of Amir was
  partially supported by a Rothschild Fellowship from Yad Hanadiv (The Rothschild Foundation).
}%

\acknowledgements{}%

\iftodo{}\else\nolinenumbers\fi

\iftechreport
  \hideLIPIcs%
\fi

\EventEditors{Robbert Krebbers and Alexandra Silva}
\EventNoEds{2}
\EventLongTitle{40th European Conference on Object-Oriented Programming (ECOOP 2026)}
\EventShortTitle{ECOOP 2026}
\EventAcronym{ECOOP}
\EventYear{2026}
\EventDate{June 29--July 3, 2026}
\EventLocation{Brussels, Belgium}
\EventLogo{}
\SeriesVolume{372}
\ArticleNo{31}

\makeatletter%
\begin{document}
\maketitle

\ifecoop
\begin{abstract}
This paper introduces \toolName (Network-Enforced Session Types), a runtime verification framework that moves application-level protocol monitoring into the network fabric. Unlike prior work that instruments or wraps application code, we synthesize packet-level monitors that enforce protocols directly in the data plane. We develop algorithms to generate network-level monitors from session types and extend them to handle packet loss and reordering. We implement \toolName in P4 and evaluate it on applications including microservice and network-function models, showing that network-level monitors can enforce realistic non-trivial protocols.%
\end{abstract}
 \fi
\section{Introduction}
\label{sec:introduction}

Session types are a well-established formalism for specifying and
verifying message-passing
programs~\cite{Ho93,HoYoCa08,GaVa10,Wa12,CaPf10}.  Whereas
conventional type systems model the types of data used by each process
(i.e., integers, strings, objects, etc.), session types also model
\textit{how} processes interact by sending and receiving messages.
For example, a process might receive a string from $A$, send an
integer to $B$, and then receive a boolean from $C$.

Session types have been used to verify implementations of complex
multiparty protocols, ensuring that each node only sends and receives
well-typed messages and that the system does not fail unexpectedly.
Although they originated in process algebras, session types have been
incorporated into mainstream programming languages including Rust, Go,
Java, OCaml, and Scala~\cite{ChBaTo22,NgYo16,CaHuJoNgYo19,ImYoYu19,JeMuLa15,Ko19,CuYo21,CuYoVa22,
  HuKoPeYoHo10,HuYoHo08,ScYo16,Pa17,OrYo16,PuTo08,ImYuAg10,Yoshida2024}.

\ifecoop
\subparagraph{Network-level monitoring: opportunities and challenges.}
\else
\paragraph{Network-level monitoring: opportunities and challenges}
\fi
A \emph{monitor} observes a system at runtime and checks conformance to a specification. For a session type $\stT$, it observes sent and received messages; on a violation, it can raise an alert or drop the message.
Runtime monitors are useful when the programs running on certain nodes cannot be statically type-checked, or for providing defense in depth. Most prior work has focused on runtime monitoring at the application
level~\cite{BuFrSc21,BuFrScTrTu21,BoChDeHoYo17,DBLP:journals/fac/NeykovaBY17,DBLP:journals/fmsd/DemangeonHHNY15,DBLP:conf/rv/NeykovaYH13}. This
paper asks a different question: \textit{can we synthesize
  monitors that enforce session types at the network level?}
We have two primary motivations.

First, deploying monitors deeper in the network stack
places them beyond end-host control.
This gives stronger assurance in mixed-trust settings: in a public cloud, provider-managed
monitors can enforce session types even when tenants do not trust one another.
This motivation is illustrated in \Cref{fig:network-microservice-example},
which shows a network based on the ``BookInfo''
application described by Istio~\cite{IstioBookinfo}. The network
includes four end hosts, each implementing a different microservice
(\lstinline|Info|, \lstinline|Review|, \lstinline|Details|, and
\lstinline|Ratings|) connected to each other and to an external
\lstinline|Client|. \change{\#123B: outline BookInfo protocol}{
first have the \lstinline|Client| query the \lstinline|Info| end host
for information on a book. \lstinline|Info| then queries \lstinline|Review|
and \lstinline|Details| for information, where the former also itself
queries \lstinline|Ratings|. Finally, \lstinline|Info| replies to the
\lstinline|Client| with the information obtained.}

Without network-level session monitoring, faulty
or malicious code running on one of the end hosts (e.g.,
\lstinline|Info|) may generate invalid packets that reach other end
hosts, consuming network resources and potentially crashing
applications when they receive unexpected messages.

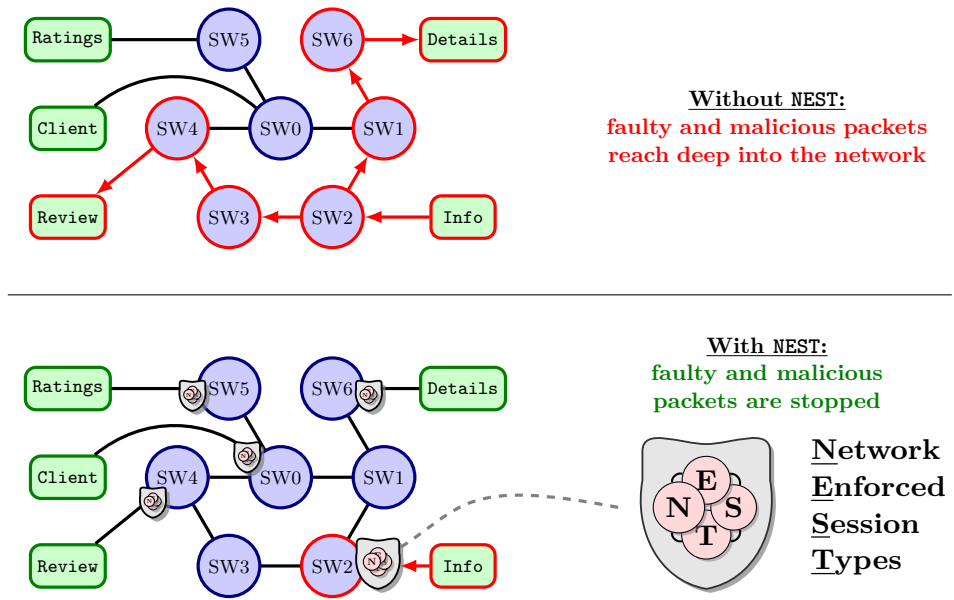
\begin{figure}
    \centering
\newcommand{\nestlogosmall}[2]{
  \begin{scope}[shift={(#1)}, scale=#2]
    \draw[
      fill=gray!20,
      draw=black!100,
      line width=0.7pt,
      drop shadow={shadow xshift=1.5pt, shadow yshift=-1.5pt, opacity=0.7}
    ]
    (-0.7,0.62)
      to[out=90, in=180] (0,0.7)
      to[out=0,  in=90] (0.7,0.62)
      -- (0.7,0.3)
      to[out=-80, in= 10] (0,-0.9)
      to[out=170, in=-100] (-0.7,0.3)
      -- cycle;

      \node[circle, draw, fill=red!15, inner sep=2pt, drop shadow={shadow xshift=1pt, shadow yshift=-1pt}] (T) at (270:0.2){\textbf{T}};
      \node[circle, draw, fill=red!15, inner sep=2pt, drop shadow={shadow xshift=1pt, shadow yshift=-1pt}] (S) at (0:0.2)  {\textbf{S}};
      \node[circle, draw, fill=red!15, inner sep=2pt, drop shadow={shadow xshift=1pt, shadow yshift=-1pt}] (E) at (90:0.2) {\textbf{E}};
      \node[circle, draw, fill=red!15, inner sep=2pt, drop shadow={shadow xshift=1pt, shadow yshift=-1pt}] (N) at (180:0.2){\textbf{N}};
  
      \draw[thick] (N) to[bend left=60] (E);
      \draw[thick] (E) to[bend left=60] (S);
      \draw[thick] (S) to[bend left=60] (T);
      \draw[thick] (T) to[bend left=60] (N);
  \end{scope}
}

\begin{tabular}{m{12cm}m{4cm}}
    \begin{tikzpicture}[
        scale=0.8, transform shape,
        font=\small,
        >=latex,
        role/.style={very thick, rectangle, draw=green!50!black, fill=green!20,
                     text centered, minimum height=2em, minimum width=3em,
                     rounded corners},
        sw/.style={very thick, circle, draw=blue!50!black, fill=blue!20,
                   text centered, minimum size=2.5em},
        highlight/.style={draw=red, very thick},
        link/.style={very thick},
        highlightLink/.style={-latex, very thick, color=red}
      ]

      \node[sw]               (SW0) at (0,0)        {SW0};
      \node[sw, highlight]    (SW1) at (1.7, 0)     {SW1};
      \node[sw, highlight]    (SW2) at (0.85,-1.47) {SW2};
      \node[sw, highlight]    (SW3) at (-0.85,-1.47){SW3};
      \node[sw, highlight]    (SW4) at (-1.7,0)     {SW4};
      \node[sw]               (SW5) at (-0.85,1.47) {SW5};
      \node[sw, highlight]    (SW6) at (0.85,1.47)  {SW6};

      \node[role]             (Ratings) at (-3.5,  1.47)  {\texttt{Ratings}};
      \node[role]             (Client)  at (-3.5,  0)  {\texttt{Client}};
      \node[role, highlight]  (Details) at (3, 1.47)  {\texttt{Details}};
      \node[role, highlight]  (Review)  at (-3.5, -1.47)    {\texttt{Review}};
      \node[role, highlight]  (Info)    at (3, -1.47)   {\texttt{Info}};

      \draw[link]                  (Ratings) -- (SW5);
      \draw[link]                  (Client) to[bend left=40] (SW0);
      \draw[highlightLink, <-]     (Details) -- (SW6);
      \draw[highlightLink, <-]     (Review)  -- (SW4);
      \draw[highlightLink, ->]     (Info)    -- (SW2);

      \draw[link]                  (SW0) -- (SW1);
      \draw[link]                  (SW0) -- (SW4);
      \draw[link]                  (SW0) -- (SW5);
      \draw[highlightLink, <-]     (SW1) -- (SW2);
      \draw[highlightLink, ->]     (SW1) -- (SW6);
      \draw[highlightLink, ->]     (SW2) -- (SW3);
      \draw[highlightLink, ->]     (SW3) -- (SW4);

      \node[text width=5.5cm, align=center, font=\large] at (8,0) {\textbf{\underline{Without \toolName:} \\ \color{red}{faulty and malicious packets reach deep into the network}}};
      
      \end{tikzpicture}
    \end{tabular}    
    \begin{tabular}{m{12cm}m{4cm}}
      \\[-1pt] \hline \\
      \centering
      \begin{tikzpicture}[
          scale=0.8, transform shape,
          font=\small,
          >=latex,
          role/.style={very thick, rectangle, draw=green!50!black, fill=green!20,
                       text centered, minimum height=2em, minimum width=3em,
                       rounded corners},
          sw/.style={very thick, circle, draw=blue!50!black, fill=blue!20,
                     text centered, minimum size=2.5em},
          highlight/.style={draw=red, very thick},
          link/.style={very thick},
          highlightLink/.style={-latex, very thick, color=red}
        ]

        \node[sw]               (SW0) at (0,0)        {SW0};
        \node[sw]    (SW1) at (1.7, 0)     {SW1};
        \node[sw, highlight]    (SW2) at (0.85,-1.47) {SW2};
        \node[sw]    (SW3) at (-0.85,-1.47){SW3};
        \node[sw]    (SW4) at (-1.7,0)     {SW4};
        \node[sw]               (SW5) at (-0.85,1.47) {SW5};
        \node[sw]    (SW6) at (0.85,1.47)  {SW6};

        \node[role]             (Ratings) at (-3.5,  1.47)  {\texttt{Ratings}};
        \node[role]             (Client)  at (-3.5,  0)  {\texttt{Client}};
        \node[role]  (Details) at (3, 1.47)  {\texttt{Details}};
        \node[role]  (Review)  at (-3.5, -1.47)    {\texttt{Review}};
        \node[role, highlight]  (Info)    at (3, -1.47)   {\texttt{Info}};

        \draw[link]                  (Ratings) -- (SW5);
        \draw[link]                  (Client) to[bend left=40] (SW0);
        \draw[link]     (Details) -- (SW6);
        \draw[link]     (Review)  -- (SW4);
        \node (Reject) at (1.6,-1.4) {\huge\color{red}\xmark};

        \draw[link, highlight, ->]     (Info)    -- ($(SW2)+(1.1,0)$);

        \draw[link]                  (SW0) -- (SW1);
        \draw[link]                  (SW0) -- (SW4);
        \draw[link]                  (SW0) -- (SW5);
        \draw[link]     (SW1) -- (SW2);
        \draw[link]     (SW1) -- (SW6);
        \draw[link]     (SW2) -- (SW3);
        \draw[link]     (SW3) -- (SW4);
        \nestlogosmall{1.6,-1.4}{0.5}

        \nestlogosmall{-0.55,0.35}{0.3}
        \nestlogosmall{-2.10,-0.4}{0.3}
        \nestlogosmall{-1.45,1.37}{0.3}
        \nestlogosmall{1.45,1.37}{0.3}
  
        \node[text width=4cm, align=center, font=\large\bfseries] 
          at (8,1.7) {\underline{With \toolName:} \\ \color{green!50!black}{faulty and malicious packets are stopped}};
        \begin{scope}[shift={(7,-0.5)}, scale=1.5]
          \draw[
            fill=gray!20,
            draw=black!80,
            line width=0.7pt,
            drop shadow={shadow xshift=1.5pt, shadow yshift=-1.5pt, opacity=0.7}
          ]
          (-0.7,0.62)
            to[out=90, in=180] (0,0.7)
            to[out=0,  in=90] (0.7,0.62)
            -- (0.7,0.3)
            to[out=-80, in= 10] (0,-0.9)
            to[out=170, in=-100] (-0.7,0.3)
            -- cycle;

          \node[circle, draw, fill=red!15, inner sep=2pt, drop shadow={shadow xshift=1pt, shadow yshift=-1pt}] (T) at (270:0.3){\textbf{T}};
          \node[circle, draw, fill=red!15, inner sep=2pt, drop shadow={shadow xshift=1pt, shadow yshift=-1pt}] (S) at (0:0.3)  {\textbf{S}};
          \node[circle, draw, fill=red!15, inner sep=2pt, drop shadow={shadow xshift=1pt, shadow yshift=-1pt}] (E) at (90:0.3) {\textbf{E}};
          \node[circle, draw, fill=red!15, inner sep=2pt, drop shadow={shadow xshift=1pt, shadow yshift=-1pt}] (N) at (180:0.3){\textbf{N}};
      
          \draw[thick] (N) to[bend left=60] (E);
          \draw[thick] (E) to[bend left=60] (S);
          \draw[thick] (S) to[bend left=60] (T);
          \draw[thick] (T) to[bend left=60] (N);

          \node[text width=1.5cm, align=left, font=\small\bfseries]
            at (1.9,0) {\underline{N}etwork \underline{E}nforced \underline{S}ession \underline{T}ypes};
        \end{scope}

        \draw[very thick, dashed, decorate, gray] ($(N)+(-1,0)$) .. controls (3,0) and (3,-0.5) .. (Reject);
        
        \end{tikzpicture}
      \end{tabular}
     \caption{A network of microservices based on the Istio ``BookInfo''
      application \cite{IstioBookinfo}, highlighting a flow of bad packets
      generated by a faulty implementation of the \lstinline|Info|
      microservice, and how \toolName can stop them.}
    \label{fig:network-microservice-example}
\end{figure}

Second, network-level monitors can run efficiently on suitable hardware, such as programmable switches and NICs. %
\change{\#123A: why no performance evaluation}{%
  For instance, the P4 language for programming network switches \cite{p4}
  is designed for high performance, with line-rate processing of %
  packets~\cite{P4FPGA19,NetCache2017}.
}%

\change{\#123B: rationale for application-layer monitoring at network layer}{%
  Notably, since session types describe application-level protocols, their monitoring
  at the network level breaks the classical network layering. %
  This is not uncommon in modern networks, where classical layering abstraction
  are sometimes broken to implement various functionalities. Middleboxes such as
  NAT boxes, load balancers, proxies, and content caches realise functionality
  at the transport layer (i.e., TCP/UDP or Layer 4) or above, by manipulating
  packets at the network layer (i.e., IP or Layer 3): see the APLOMB paper
  \cite{DBLP:conf/sigcomm/SherryHSKRS12} for a more detailed explanation,
  including a survey of cloud operators. %
  Moreover, existing network-layer devices already enforce simple
  application-level patterns (e.g., firewalls, NAT, DPI) -- but these policies
  are limited: for example, NAT lets internal hosts receive packets only after
  outbound traffic. %
}%
Compared to these approaches, session types can express richer protocols and
come equipped with formal guarantees.

However, realizing network-level session type monitors requires addressing several challenges:
\begin{itemize}
\item Session types can model rich behaviors that go well beyond static policies and simple firewalls. Hence, monitor synthesis must be \textit{automatic}. 
\item Monitors for session types must also be \textit{stateful}, to track the current protocol state and accept or
  reject packets accordingly. While basic connection tracking exists in devices such as stateful firewalls and NAT boxes,
  session types bring significant complexity.
\item Network-level monitors must handle reordered packets and retransmissions after \textit{loss}, often with limited hardware buffering. By contrast, existing application-level session monitors assume reliable transport (e.g., TCP) and enough buffering to reorder packets.
\end{itemize}

\ifecoop
\subparagraph{Contributions and outline.}
\else
\paragraph{Contributions and outline}
\fi%
\change{\#123B: improve contributions description}{%
   To address these challenges, we design and implement \toolName, a tool for generating and deploying network-level session monitors. \toolName takes as input a set of local session types (represented with a Scala 3 embedded DSL) and their associated roles, generates corresponding P4 monitor representations, and uploads them on P4-enabled network devices. We develop the formal foundations of \toolName, and evaluate it on a set of representative multiparty protocols using Mininet \cite{DBLP:conf/hotnets/LantzHM10}, a realistic software-defined networking emulation platform. %
}%
By enforcing session types at the network perimeter, \toolName monitors can discard protocol-violating packets early, saving network resources and preventing invalid traffic from reaching downstream hosts and devices.
More broadly, our approach realises an ``off-by-default'' network~\cite{BaChRaRoSh05}, where only authorized packets can traverse the network.

The main contributions of this paper are as follows:
\begin{description}
  \item[\Cref{sec:design}] introduces \toolName. Given
    \change{\#123B: clarify \toolName input}{%
      a set of session types written in a Scala 3 embedded DSL, %
    }%
    \toolName synthesizes network-layer monitors based on P4.
    \toolName also generates an API for
    writing end-host programs whose communication patterns are tracked
    by our monitors (\Cref{sec:implementation-st-api}).
  \item[\Cref{sec:formal-model}] presents the formal model at the basis of \toolName monitors.
    We introduce a novel monitor synthesis technique (\Cref{def:monitor-state-st}) tailored
    for network-level session monitors that reject bad packets, while ensuring
    soundness -- i.e., messages from well-behaved end hosts
    are not rejected (\Cref{thm:net-mon-bisim}).
  \item[\Cref{sec:evaluation}] presents a qualitative and quantitative
    evaluation of \toolName, showing that the monitors it generates
    accept valid packets and reject invalid ones across a variety of
    representative multiparty protocols.
\end{description}
Finally, \Cref{sec:related-work} discusses related work, and \Cref{sec:conclusion} concludes with future directions.
\change{Cite artifact}{%
  \toolName is available in the companion artifact of this work %
  \iftechreport{\cite{larsen_2026_19661497} }\else{}\fi%
  with instructions for reproducing the evaluation in \Cref{sec:evaluation}.%
}%

\section{Background}
\label{sec:background}

We briefly review session types using the Istio
BookInfo application~\cite{IstioBookinfo} from
\Cref{sec:introduction}, a microservice application for an online bookseller.
First, a \lstinline|Client| requests information on a given book. The
outward-facing \lstinline|Info| service then gathers data from two
internal microservices: \lstinline|Details| and
\lstinline|Review|. The \lstinline|Details| microservice replies
immediately with data (e.g., author and ISBN).
\begin{wrapfigure}{r}{0.48\textwidth}
\vspace{-5mm}
\begin{lstlisting}[escapechar=\%,basicstyle=\scriptsize\ttfamily]
protocol ReviewProtocol {
  role Client, Info, Review, Ratings, Details;
  Client -> Info: request(int);
  Info -> Review: review_request(int);
  Info -> Details: detail_request(int);
  Review -> Ratings: ratings_request(int);
  Ratings -> Review: ratings_response(float);
  Review -> Info: review_response(str);
  Details -> Info: detail_response(str);
  Info -> Client: response(str);
}
\end{lstlisting}
\vspace{-4mm}
\caption{Global session type for BookInfo.}
\label{fig:InfoGlobalType}
\vspace{-4mm}
\end{wrapfigure}
Meanwhile, \lstinline|Review| first queries the internal \lstinline|Ratings|
service and then replies with reviews and ratings. The behavior of
BookInfo is captured by the global session type in
\Cref{fig:InfoGlobalType}: \lstinline|Client| sends \lstinline|Info| a
\lstinline|request| message with \lstinline|int| payload;
\lstinline|Info| then sends \lstinline|Review| a
\lstinline|review_request(int)| message, and so on.\change{\#123B,\#123C: clarify \toolName's inputs with a footnote}{}\footnote{%
  \label{footnote:global-types-proj}%
  The syntax of the global type in \Cref{fig:InfoGlobalType} is inspired by tools like Scribble (\url{https://github.com/scribble/scribble-java})
  and $\nu$Scr (\url{https://github.com/nuScr/nuScr})
  and only serves to illustrate the standard session types framework and the BookInfo protocol. %
  As explained later in \Cref{sec:design} and
  \Cref{remark:determining-liveness-half-duplex}, \toolName takes as input one or more
  \emph{local} session types that are part of a multiparty protocol that is output-live (\Cref{def:live-net}) and half-duplex. If needed, such local session types can be
  obtained from a global type using the standard \emph{projection} techniques mentioned in this section,
  and implemented in tools like Scribble, $\nu$Scr, and mpstk (\url{https://github.com/alcestes/mpstk-crash-stop}).%
}%

By \textit{projecting} a global session type to one role, we obtain
a \emph{local session type}. The local session type in \Cref{fig:InfoLocalSessionType} describes the communication protocol
enacted by role \lstinline|Info|. Although this simple example is linear, session types can also express \emph{branching} and \emph{recursion}:
\change{\#123B: explain branching and recursion}{%
   branching session types are choice points where a role may send or receive one among several different messages
   and possibly continue the session in a different way (see \Cref{ex:bad-network-message} later on), and recursive session types
   allow for repeating part of a session (see \Cref{ex:infinite-state-monitor}).%
}%

Local session types are typically used for compile-time type checking. In practice, many components cannot be session-typed (e.g., unsupported languages/frameworks or inaccessible participants such as the external \lstinline|Client| in BookInfo). Even then, local session types remain a precise and expressive protocol specification language, so we use them as the basis for network monitoring.

\begin{figure}
  \centering
\begin{minipage}[t]{0.49\textwidth}
\centering
\begin{lstlisting}[escapechar=\%,basicstyle=\scriptsize\ttfamily]
ReviewProtocol[Info] = {
    Client?request(int);
    Review!review_request(int);
    Details!detail_request(int);
    Review?review_response(str);
    Details?detail_response(str);
    Client!response(str);
}
\end{lstlisting}
\end{minipage}\;\begin{minipage}[t]{0.49\textwidth}
\centering
\begin{lstlisting}[language=py,basicstyle=\scriptsize\ttfamily]
sh = BookReviewSessionHandler(42, Info)
while True:
  book_id = sh.recvMsg(Client, request)
  sh.sendMsg(Review, review_request(book_id))
  sh.sendMsg(Details, detail_request(book_id))
  review = sh.recvMsg(Review, review_response)
  detail = sh.recvMsg(Details, detail_response)
  sh.sendMsg(Client, response(review + detail))
\end{lstlisting}
\end{minipage}
\caption{Local session type (left) and end host code (right) for the \lstinline|Info| role in the BookInfo application~~\cite{IstioBookinfo}. The \texttt{?} and \texttt{!} symbols represent receiving and sending a message, respectively.}
\label{fig:our-approach}%
\label{fig:InfoLocalSessionType}
\end{figure}
\section{\toolName: overview and end-host monitor design}
\label{sec:design}

This section overviews \toolName, our toolkit for synthesizing and deploying network-level monitors from session type specifications.
\Cref{fig:tool-pipeline} summarizes the workflow.

Our design for \toolName relies on two key assumptions. First, we assume we
are giving a local session type $\stT$ that captures the intended behavior of
each node. For now, an intuitive understanding of session types will suffice.
The formal definition will be given in \Cref{sec:formal-model}. Second, we
assume that the behavior of the devices at the edge of the network can be
specified in P4, a domain-specific language for programming network
switches~\cite{p4}. To understand \toolName, a deep understanding of P4 will not
be necessary. For now, there are two things to know: (i) P4 provides match-action
tables (MATs), which can be populated with entries at runtime to control how
packets are processed; (ii) P4 provides mutable registers, with associated read
and write operations, which can be used to implement stateful packet processing.
We will use both of these features in our design for \toolName.

\change{\#123B: intuition on MATs}{%
A MAT is essentially a table
where each row expresses a rule for recognising and handling packets. %
Intuitively, a MAT has two columns called \textit{match} and \textit{action}. The ``match'' column determines what is matched
(e.g., packet IP address, port, or P4 state register) and the expected values. The
``action'' column determines which action is applied to a matching packet (e.g., forward,
drop, etc.).
The MATs defined by \toolName encode a state machine: in each row, the ``match'' column specifies
how to match a packet based on the current state of the switch, and on the message label,
sender and receiver carried by the packet; then, the ``action'' column
specifies whether to accept the packet (transitioning to another state) or reject it.
}

Given an application's (local) session type, \toolName generates the following:

\begin{itemize}[leftmargin=*]
	\item An \textbf{API} for sending and receiving messages in
          the format expected by our network monitors. We describe the
          message format in \Cref{sec:implementation-st-headers} and
          the session API generation in
          \Cref{sec:implementation-st-api}.
	\item A set of \textbf{P4 routing table entries} (i.e., MAT entries) for
          enforcing the session type on incoming traffic; these encode the
          session type's state progression as messages are sent and received by
          participants.
\end{itemize}

\begin{figure}[t]
	\centering
	\includegraphics[width=.7\textwidth]{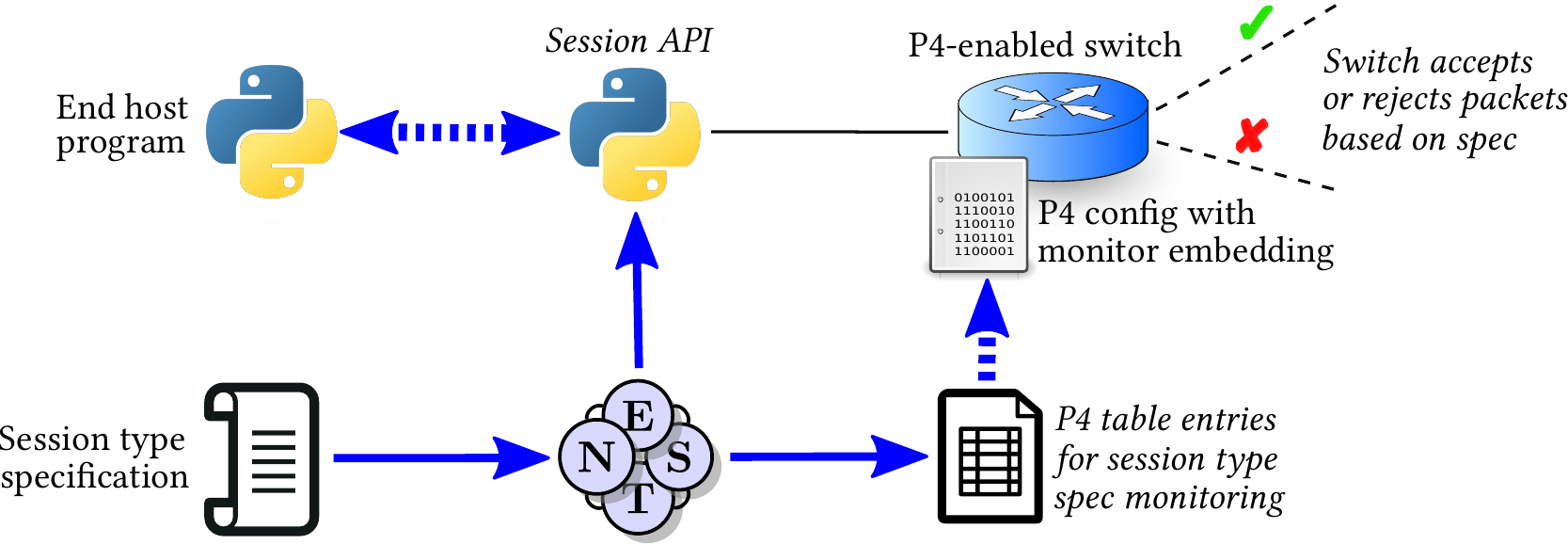}
\caption{Overview of \toolName: the programmer writes a session type specification (left);
\toolName generates a Python session API, which end-host programs use for network communication (right).
\toolName also generates P4 switch configurations, which drop non-conformant packets.
}
	\label{fig:tool-pipeline}
\end{figure}

Our monitors reject illegal packets at the network perimeter.
For example, in \Cref{fig:network-microservice-example}, traffic from the faulty node is dropped at switch \texttt{SW2}.

Given a session type such as \Cref{code:st-info-impl} and a role, \toolName synthesizes a monitor in four steps
\change{\#123B: ``bringing down'' session types to the network layer}{%
  that effectively ``bring down'' an application-level session type specification
  into the lower-level network layer:%
}%

\begin{enumerate}
\item \toolName's synthesis module constructs a state machine for the session-type monitor.
\item \toolName converts this state machine into MAT entries, mapping each transition to one entry. For \Cref{code:st-info-impl}, the generated entries are shown in \Cref{tab:st-monitor-transitions}.
\item %
\toolName then translates these MAT entries into entries for the P4 table \texttt{monitor\_table} (\Cref{sec:implementation-st-headers}), encoding sender/receiver roles and message labels as enumerated IDs.
At this stage, \toolName can also generate monitoring logic for packet loss, duplication, and reordering in TCP connections (\Cref{sec:implementation-practical-protocols}).
\item Finally, \toolName deploys the generated entries on P4-enabled devices acting as monitors.\footnote{%
	The deployment phase uses the P4R-Type library~\cite{LaGuHaSc23}
	to statically ensure that the deployed entries conform to the \texttt{monitor\_table}
	table definition.%
}
\end{enumerate}

\begin{figure}
  \centering
  \noindent%
\begin{minipage}{0.46\textwidth}
\begin{lstlisting}[language=scala,numbers=none,numbersep=0pt,xleftmargin=0pt,frame=none,basicstyle=\fontsize{7}{8}\selectfont\ttfamily]
val lst_BookInfo =
  ExtCh(a_BookClient, Map(l_Request ->
  IntCh(Map((a_BookReview, l_ReviewRequest) ->
  IntCh(Map((a_BookDetails, l_DetailRequest) ->
  ExtCh(a_BookReview, Map(l_ReviewResponse ->
  ExtCh(a_BookDetails, Map(l_DetailResponse ->
  IntCh(Map((a_BookClient, l_Response) ->
  End() ))))))))))))
\end{lstlisting}
\vspace{-3mm}%
		\caption{The \lstinline|Info| session type in
		\Cref{fig:InfoLocalSessionType} (left), represented as a Scala 3 data type
		instance which \toolName takes as input. %
    \lstinline|ExtCh| means ``await an incoming message chosen and sent by another role,'' whereas
    \lstinline|IntCh| means ``choose a recipient role and a message, and perform the send operation.''
    (In this example, all choices have only one option.)}
		\label{code:st-info-impl}
	\end{minipage}
  \hfill
  \begin{minipage}{0.53\textwidth}
    \centering
    \scriptsize
    \renewcommand{\arraystretch}{1.3} %
    \begin{tabular}{|c|c|c|c|c|}
        \hline
        \rowcolor{lightgray} \multicolumn{4}{|c|}{\textbf{Match}} & \cellcolor{lightgray}\\
        \cline{1-4}
        \rowcolor{lightgray}\textbf{State} & \textbf{Sender} & \textbf{Receiver} & \textbf{Label} & \multirow{-2}{*}{\textbf{Action}} \\
        \hline\hline
        m0 & \lstinline|Client|  & \lstinline|Info|    & $\labFmt{req}$    & \ttt{accept}(m1) \\ \hline
        m1 & \lstinline|Info|    & \lstinline|Review|  & $\labFmt{r\_req}$ & \ttt{accept}(m2) \\ \hline
        m2 & \lstinline|Review|  & \lstinline|Info|    & $\labFmt{r\_rsp}$ & \ttt{accept}(m3) \\ \hline
        m2 & \lstinline|Info|    & \lstinline|Details| & $\labFmt{d\_req}$ & \ttt{accept}(m4) \\ \hline
        m3 & \lstinline|Info|    & \lstinline|Details| & $\labFmt{d\_req}$ & \ttt{accept}(m5) \\ \hline
        m4 & \lstinline|Review|  & \lstinline|Info|    & $\labFmt{r\_rsp}$ & \ttt{accept}(m5) \\ \hline
        m4 & \lstinline|Details| & \lstinline|Info|    & $\labFmt{d\_rsp}$ & \ttt{accept}(m6) \\ \hline
        m5 & \lstinline|Details| & \lstinline|Info|    & $\labFmt{d\_rsp}$ & \ttt{accept}(m7) \\ \hline
        m6 & \lstinline|Review|  & \lstinline|Info|    & $\labFmt{r\_rsp}$ & \ttt{accept}(m7) \\ \hline
        m7 & \lstinline|Info|    & \lstinline|Client|  & $\labFmt{rsp}$    & \ttt{accept}(m8) \\ \hline
        \multicolumn{4}{|c|}{Otherwise} & \ttt{reject} \\ \hline
    \end{tabular}
    \vspace{-3mm}
    \caption{A match-action table (MAT) generated by \toolName from the session type in \Cref{code:st-info-impl}. %
    The message labels are shortened for brevity. %
    In the ``action'' column, \texttt{accept}(m$X$) indicates that the monitor transitions to state m$X$ upon accepting the message.%
    }
    \label{tab:st-monitor-transitions}
  \end{minipage}
\end{figure}

After deployment, the P4 device enforces session-type monitoring: upon
receiving a packet, \texttt{monitor\_table} inspects the packet header and
accepts or rejects it. %
\toolName also generates a session API for sending and
receiving packets with the headers expected by the monitors (\Cref{sec:implementation-st-api}).

\ifecoop
  \subparagraph*{Challenges.}
\else
  \paragraph{Challenges.}
\fi
The rest of this section addresses three practical challenges.
\begin{itemize}[leftmargin=*]
  \item \textbf{Accept/reject decisions and session state tracking (\Cref{sec:implementation-st-headers}):}
		how should a P4 device correctly accept/reject packets while tracking multiple session types concurrently?
		\item \textbf{Shared entry points (\Cref{sec:implementation-practical-sharedentry}):} how should one monitor handle multiple end hosts sharing an ingress point?
  \item \textbf{Packet loss, duplication, and reordering (\Cref{sec:implementation-practical-protocols}):}
    how should \toolName support transport protocols (e.g., TCP) that affect packet sequencing?
\end{itemize}

\ifecoop
  \subparagraph*{Assumptions and limitations.}
\else
  \paragraph{Assumptions and limitations.}
\fi

We assume the devices at the edge of the network can be programmed in P4, so that all communication between protocol roles passes through monitored devices. The current version of \toolName also assumes each session message fits into a single packet; ideas for lifting this restriction to handle fragmentation and additional transport protocols are discussed in \Cref{sec:conclusion-futurework}.

\subsection{Accepting/Rejecting Packets and Tracking Session-Type State}
\label{sec:implementation-st-headers}

\toolName-generated monitors decide whether to accept or reject a packet by processing a dedicated
\emph{session header} (\Cref{fig:session-header}). %
Each monitored-session packet must carry this header; others are rejected by default.
The P4 monitoring logic is implemented by the table \texttt{monitor\_table}
(\Cref{fig:p4-monitor-table}), which extracts from the session header the
\emph{message size} (required for deparsing the packet),
\emph{sender role}, \emph{receiver role}, and
\emph{message label ID}, to decide whether to accept or reject
the packet.\footnote{%
	The session header currently used by \toolName (depicted in
	\Cref{fig:session-header}) supports protocols with up to 15 distinct roles
	and up to 63 distinct message labels. This bound can be increased if
	needed, although doing so may be constrained by packet size and the memory
	available on the P4 switch.%
} %
\texttt{monitor\_table} matches extracted session-header fields
against entries generated by \toolName from a session-type MAT such as \Cref{tab:st-monitor-transitions}. %
Packets matching at least one entry are accepted; otherwise, the default \texttt{reject} action drops them (line~13 in \Cref{fig:p4-monitor-table}).

\begin{figure}
\begin{minipage}{0.5\textwidth}
  \centering%
  \begin{tabular}{|l|l|l|l|l|l|l|l|} \hline
    \rowcolor{lightgray}  \textbf{0} & \textbf{1} & \textbf{2} & \textbf{3} & \textbf{4} & \textbf{5} & \textbf{6} & \textbf{7} \\ \hline
    \multicolumn{8}{|c|}{Message Size} \\ \hline
    \multicolumn{8}{|c|}{Session ID} \\ \cline{3-8}
    \multicolumn{2}{|c|}{} &
     \multicolumn{6}{c|}{Message Label ID} \\ \hline
    \multicolumn{4}{|c|}{Sender Role} & \multicolumn{4}{c|}{Receiver Role}\\ \hline
    \multicolumn{8}{|c|}{\multirow{2}{*}{Session Sequence Number}} \\
    \multicolumn{8}{|c|}{} \\ \hline
  \end{tabular}
    \caption{Specification of the session header.}
    \label{fig:session-header}
\end{minipage}
  \hfill%
  \begin{minipage}{0.45\textwidth}
    \begin{lstlisting}[style=P4Style,basicstyle=\scriptsize\ttfamily]
table monitor_table {
  key = {
        hdr.session.sender: exact;
        hdr.session.label: exact;
        hdr.session.receiver: exact;
        currentState: exact;
    }
    actions = {
        accept;
        reject;
    }
    size = 1024;
    default_action = reject();
}
\end{lstlisting}
    \vspace{-3mm}%
    \caption{Specification of the P4 table for monitor transitions.}
    \label{fig:p4-monitor-table}
  \end{minipage}
\end{figure}

\texttt{monitor\_table}'s \texttt{accept} action also inspects the
\textit{session ID}, allowing one device to distinguish and track many sessions in parallel.
When a packet with session ID $s$ is accepted, \texttt{accept} uses $s$ to index a P4 stateful register and update that session's state.

The session header also carries a \textit{session sequence number}, used to handle retransmissions (discussed in \Cref{sec:implementation-practical-protocols}).

\ifecoop
  \subparagraph*{Computing the session monitor MAT.}
\else
  \paragraph{Computing the session monitor MAT}
\fi
\toolName's monitoring strategy hinges on correctly synthesizing
the session-type monitor MAT that determines which packets are accepted or rejected. %
For instance, %
the state machine in \Cref{fig:st-info-visual-graph} corresponds to the
standard semantics of the session type in \Cref{code:st-info-impl} -- whereas
the state machine in \Cref{fig:st-monitor-visual-graph} is synthesized by \toolName
from the same session type, and then converted into the MAT in
\Cref{tab:st-monitor-transitions}. %
Observe that state m2 of the MAT allows
\lstinline|Info| to either send $\stLabFmt{d\_req}$ to \lstinline|Details| or
receive $\stLabFmt{r\_rsp}$ from \lstinline|Review| -- whereas these actions are
sequential in \Cref{code:st-info-impl,fig:st-info-visual-graph} (first send, then receive). %
\change{\#123B: ``bringing down'' session types to the network layer}{%
  The difference between \Cref{fig:st-info-visual-graph} and
  \Cref{fig:st-monitor-visual-graph} is a consequence of ``bringing down''
  a session type specification to the network layer for monitoring purposes. %
}%
This is because, while enforcing a session type $\stT$,
\toolName's network-level monitor may see packets that diverge from $\stT$'s expected order.
The monitor must distinguish ``bad'' packets that violate $\stT$ from ``good'' packets delivered in a different order. %
We formalize this in \Cref{sec:formal-model} and prove correctness.

\newcommand{\stInfo}{\stT[\text{\lstinline|Info|}]}%
\begin{figure}
        \centering
        \begin{subfigure}{.4\textwidth}
            \small
            \centering
            \scalebox{0.95}{%
            \begin{tikzpicture}
                \node[circle, draw=black, double=white, minimum size=7mm, very thick] (s0)                     {s0};
                \node[circle, draw=black, minimum size=7mm, very thick] (s1) [below = .5cm of s0] {s1};
                \node[circle, draw=black, minimum size=7mm, very thick] (s2) [below = .5cm of s1] {s2};
                \node[circle, draw=black, minimum size=7mm, very thick] (s3) [below = .5cm of s2] {s3};
                \node[circle, draw=black, minimum size=7mm, very thick] (s4) [below = .5cm of s3] {s4};
                \node[circle, draw=black, minimum size=7mm, very thick] (s5) [below = .5cm of s4] {s5};
                \node[circle, draw=black, minimum size=7mm, very thick] (s6) [below = .5cm of s5] {s6};
                \draw[->, thick] (s0.south) -- (s1.north) node[midway, right = .1cm]          {$\stEnqueueTag{\roleFmt{c}}{\labFmt{req}}{\tyInt}$};
                \draw[->, thick] (s1.south) -- (s2.north) node[midway, right = .1cm]          {$\stSendTag{\roleFmt{r}}{\labFmt{r\_req}}{\tyInt}$};
                \draw[->, thick] (s2.south) -- (s3.north) node[midway, right = .1cm]          {$\stSendTag{\roleFmt{d}}{\labFmt{d\_req}}{\tyInt}$};
                \draw[->, thick] (s3.south) -- (s4.north) node[midway, right = .1cm]          {$\stEnqueueTag{\roleFmt{r}}{\labFmt{r\_rsp}}{\tyString}$};
                \draw[->, thick] (s4.south) -- (s5.north) node[midway, right = .1cm]          {$\stEnqueueTag{\roleFmt{d}}{\labFmt{d\_rsp}}{\tyString}$};
                \draw[->, thick] (s5.south) -- (s6.north) node[midway, right = .1cm]          {$\stSendTag{\roleFmt{c}}{\labFmt{rsp}}{\tyString}$};
            \end{tikzpicture}
            }%
            \caption{The state machine of the session type %
              in \Cref{fig:InfoLocalSessionType,code:st-info-impl} %
              (following the formal semantics of $\stInfo$ by \Cref{def:local-semantics-no-queues}).}
            \label{fig:st-info-visual-graph}
        \end{subfigure}%
        \hspace{1cm}
        \begin{subfigure}{.52\textwidth}
            \tiny
            \centering
            \scalebox{0.9}{%
            \begin{tikzpicture}
                \node[circle, draw=black, double=white, minimum size=7mm, very thick, font=\scriptsize] (m0)                          {m0};
                \node[circle, draw=black, minimum size=7mm, very thick, font=\scriptsize] (m1) [below = .5cm of m0, font=\scriptsize]      {m1};
                \node[circle, draw=black, minimum size=7mm, very thick, font=\scriptsize] (m2) [below = .5cm of m1, font=\scriptsize]       {m2};
                \node[circle, draw=black, minimum size=7mm, very thick, font=\scriptsize] (m3) [below left = 1cm of m2, font=\scriptsize]   {m3};
                \node[circle, draw=black, minimum size=7mm, very thick, font=\scriptsize] (m4) [below right = 1cm of m2, font=\scriptsize]  {m4};
                \node[circle, draw=black, minimum size=7mm, very thick, font=\scriptsize] (m5) [below right = 1cm of m3, font=\scriptsize]  {m5};
                \node[circle, draw=black, minimum size=7mm, very thick, font=\scriptsize] (m6) [below right = 1cm of m4, font=\scriptsize]  {m6};
                \node[circle, draw=black, minimum size=7mm, very thick, font=\scriptsize] (m7) [below right = 1cm of m5, font=\scriptsize]  {m7};
                \node[circle, draw=black, minimum size=7mm, very thick, font=\scriptsize] (m8) [below = .5cm of m7, font=\scriptsize]       {m8};
                \draw[->, thick, font=\scriptsize] (m0.south) -- (m1.north) node[midway, right = .1cm, font=\scriptsize]                         {$\monLabYes{\stEnqueueTag{\roleFmt{c}}{\labFmt{req}}{\tyInt}}$};
                \draw[->, thick, font=\scriptsize] (m1.south) -- (m2.north) node[midway, right = .1cm]                         {$\monLabYes{\stSendTag{\roleFmt{r}}{\labFmt{r\_req}}{\tyInt}}$};
                \draw[->, thick, font=\scriptsize] (m2.west) .. controls +(left:5mm) and +(up:5mm) .. (m3.north) node[midway, left = .1cm, font=\scriptsize] {$\monLabYes{\stEnqueueTag{\roleFmt{r}}{\labFmt{r\_rsp}}{\tyString}}$};
                \draw[->, thick, font=\scriptsize] (m2.east) .. controls +(right:5mm) and +(up:5mm) .. (m4.north) node[midway, right = .1cm, font=\scriptsize]                {$\monLabYes{\stSendTag{\roleFmt{d}}{\labFmt{d\_req}}{\tyInt}}$};
                \draw[->, thick, font=\scriptsize] (m3.south) .. controls +(down:5mm) and +(left:5mm) .. (m5.west) node[midway, left = .1cm, font=\scriptsize]         {$\monLabYes{\stSendTag{\roleFmt{d}}{\labFmt{d\_req}}{\tyInt}}$};
                \draw[->, thick, font=\scriptsize] (m4.south) .. controls +(down:5mm) and +(right:5mm) .. (m5.east) node[midway, above left = .2cm and -.2cm, font=\scriptsize]  {$\monLabYes{\stEnqueueTag{\roleFmt{r}}{\labFmt{r\_rsp}}{\tyString}}$};
                \draw[->, thick, font=\scriptsize] (m4.east) .. controls +(right:5mm) and +(up:5mm) .. (m6.north) node[midway, right = .1cm, font=\scriptsize]        {$\monLabYes{\stEnqueueTag{\roleFmt{d}}{\labFmt{d\_rsp}}{\tyString}}$};
                \draw[->, thick, font=\scriptsize] (m5.south) .. controls +(down:5mm) and +(left:5mm) .. (m7.west) node[midway, left = .1cm, font=\scriptsize]         {$\monLabYes{\stEnqueueTag{\roleFmt{d}}{\labFmt{d\_rsp}}{\tyString}}$};
                \draw[->, thick, font=\scriptsize] (m6.south) .. controls +(down:5mm) and +(right:5mm) .. (m7.east) node[midway, right = .1cm, font=\scriptsize]        {$\monLabYes{\stEnqueueTag{\roleFmt{r}}{\labFmt{r\_rsp}}{\tyString}}$};
                \draw[->, thick, font=\scriptsize] (m7.south) -- (m8.north) node[midway, left = .1cm, font=\scriptsize]                          {$\monLabYes{\stSendTag{\roleFmt{c}}{\labFmt{rsp}}{\tyString}}$};
            \end{tikzpicture}
            }%
            \caption{The state machine of the monitor for the session type in \Cref{fig:InfoLocalSessionType,code:st-info-impl} %
              (by the formal semantics of $\enc{\stInfo}$ by \Cref{def:monitor-state-st}). %
              Note that this diagram only shows the monitor's accepting transitions; %
              it omits the transitions that reject any send/receive action which is not explicitly accepted %
              (by rule~\textsc{STMon-Bad} in \Cref{def:monitor-state-st}).
              }
            \label{fig:st-monitor-visual-graph}
        \end{subfigure}
        \caption{The state machines of the session type in \Cref{fig:InfoLocalSessionType,code:st-info-impl} and its monitor. %
          For brevity, we shorten the message labels and write $\roleFmt{c}$, $\roleFmt{r}$, and $\roleFmt{d}$ for \lstinline|Client|, \lstinline|Review|, and \lstinline|Details|, respectively.}%
        \label{fig:st-monitor-visual}
\end{figure}

\subsection{Shared Entry Points}
\label{sec:implementation-practical-sharedentry}

Although \Cref{fig:network-microservice-example} shows one switch per end host, \toolName also supports multiple end hosts per switch, even with different roles.
This raises a practical challenge: \textit{how can one device monitor multiple roles concurrently?}
Suppose a device must monitor roles $\roleP$ and $\roleQ$ in session types $\stT[\roleP]$ and $\stT[\roleQ]$.
Let $\enc{\stT[\roleP]}$ and $\enc{\stT[\roleQ]}$ be the monitors synthesized when these roles are monitored separately.
We build one monitor from these two using standard process-calculus techniques:
\begin{enumerate}
		\item Compute the labelled transition system (LTS) of the \emph{parallel composition} $\cpmPar{\enc{\stT[\roleP]}}{\enc{\stT[\roleQ]}}$ in the style of CCS~\cite{DBLP:books/daglib/0067019}, allowing synchronization when the monitor $\enc{\stT[\roleP]}$ accepts an outgoing message to $\enc{\stT[\roleQ]}$ (or \emph{vice versa}).
		\item Prune non-synchronising transitions where the monitor $\enc{\stT[\roleP]}$ accepts a message to/from $\enc{\stT[\roleQ]}$ (or \emph{vice versa}); in CCS terms, apply restriction $\cpmRst{\roleP}{\roleQ}{\left(\cpmPar{\enc{\stT[\roleP]}}{\enc{\stT[\roleQ]}}\right)}$.
		\item Deploy on the P4 device a monitor state machine matching this restricted composition, so the device monitors $\stT[\roleP]$ and $\stT[\roleQ]$ together.\footnote{%
	  With a minor extension of the formal model introduced later
	  in~\Cref{sec:formal-model}, the behavior of a network using this joint
	  monitor for $\stT[\roleP]$ and $\stT[\roleQ]$ can be proven bisimilar to a network where
	  the end hosts of $\roleP$ and $\roleQ$ are connected to different devices,
	  matching the premises of our monitor correctness result
	  (\Cref{thm:net-mon-bisim}).%
  }%
\end{enumerate}

\subsection{Handling Packet Loss, Duplication, and Reordering}
\label{sec:implementation-practical-protocols}

So far, \toolName-generated monitor state machines and MATs
(\Cref{fig:st-monitor-visual-graph,tab:st-monitor-transitions})
cover the ``core logic'' from \Cref{sec:implementation-st-headers}: for a session type $\stT$, they distinguish violating packets from out-of-order but valid packets.
This is also the focus of the formal synthesis in \Cref{sec:formal-model}.
The result is a P4 monitor for reliable networks.

However, such monitors would often behave incorrectly in real-world
networks, where packets may experience %
loss, duplication, and reordering (even between packets with the same sender).
To address these issues, transport protocols such as TCP add transmission logic
to packets and may adjust their sending order to guarantee reliable
delivery.
Therefore, \toolName is designed to be flexible with respect to transport
protocols between end hosts. In particular, \toolName supports generation
of \emph{TCP-oriented P4 monitors for session types} by augmenting a ``core''
monitor state machine and MAT (like the ones in
\Cref{fig:st-monitor-visual-graph,tab:st-monitor-transitions}, which follow
\Cref{def:monitor-state-st}) with transitions and checks for TCP socket setup
and teardown, packet acknowledgements, and retransmissions within a TCP
connection.
For example, \Cref{fig:implementation-tcp-monitor-visual} depicts the
TCP-oriented monitor obtained from the state machine in~\Cref{fig:st-monitor-visual-graph}
by adding TCP-specific transitions (dashed arrows) and checks.

\begin{itemize}[leftmargin=*]
		\item The TCP-oriented monitor always accepts packets related to TCP
		connection setup and teardown (i.e., TCP headers with the \texttt{SYN} or
		\texttt{FIN} flag set) without changing its state. %
		It also accepts pure TCP acknowledgements (i.e., TCP headers with \texttt{ACK}) unless the \texttt{ACK} is piggybacked with a session header, in which case normal monitoring applies.
	\item As mentioned in \Cref{sec:implementation-st-headers}, the session header
  carries a \textit{session sequence number (SSN)} field, which records the total number
  of messages sent by the sender so far. These sequence numbers are
  tracked by our TCP-oriented monitors. %
  The TCP-oriented monitor uses a detailed decision procedure based on SSNs; the details of this
  decision procedure are simplified away for the transitions shown in~\Cref{fig:implementation-tcp-monitor-visual},
  but are explained in~\Cref{fig:tcp-monitor-decision} (essentially, the end states in~\Cref{fig:tcp-monitor-decision}
  correspond to transitions that each state in~\Cref{fig:implementation-tcp-monitor-visual} can take).
  We explain the decision procedure:
    \begin{enumerate}
      \item If an incoming packet's SSN is less than or equal to the monitor's stored SSN for that sender, the packet is accepted as a retransmission and does not update monitor state.
	      \item Otherwise (SSN greater), the packet is first matched on sender and receiver \textit{only} (ignoring label).%
            \footnote{To match packets only on sender and receiver, the TCP-oriented monitor uses an additional table,
            \texttt{receiver\_table}, a partial version of \texttt{monitor\_table} that omits \texttt{hdr.session.label}.}
            \change{\#123C: why match on sender/receiver only?}{This is done to distinguish truly invalid packets
            from valid packets that arrive out of order. If the monitor was to match simultaneously on sender, recipient
            and label, a failed match would be ambigious. The monitor must distinguish these two cases: in the first
            case it must drop the session (see step 4), while in the second case it must just reject the packet (that will be
            retransmitted).}%
      \item If both sender and recipient match, the monitor then matches the packet label.
	          If that matches, there are two cases. %
            Suppose that the monitor is guarding the network border for end host $\roleP$:
          \begin{itemize}
            \item For each outgoing packet sent by $\roleP$, if the packet SSN matches the next expected value (i.e., the stored SSN plus one), the monitor accepts the packet and updates its state. Instead, if the SSN is too high, the monitor drops the packet: this is because the SSN being too high means that some packets were lost, hence the monitor awaits their retransmission with the correct SSN.
            \item For each incoming packet towards $\roleP$, the monitor accepts the SSN and updates its state without further checks: the monitor assumes that the SSN was already checked and accepted by the border monitor on the sender side.
          \end{itemize}
	      \item If recipient matches but label does not, the packet is rejected.
	          Moreover, if the monitor is guarding end host $\roleP$ and the packet is sent by $\roleP$ with an SSN that is exactly one higher than stored SSN, the monitor rejects the packet \textit{and} permanently closes the session -- because this indicates a session violation (not mere packet retransmission or reordering).\footnote{%
              Closing the session, rather than only rejecting the packet, ensures the bad message is not later accepted as a retransmission after SSN advances.%
            }%

	  \end{enumerate}
\end{itemize}

\begin{figure}
      \begin{minipage}{0.35\textwidth}
        \tiny
        \centering
        \begin{tikzpicture}
            \node[circle, draw=black, double=white, minimum size=7mm, very thick] (m0)                            {m0};
            \node[circle, draw=black, minimum size=7mm, very thick]  (fail)  [right = 1cm of m0] {fail};
            \node[circle, draw=black, minimum size=7mm, very thick]  (m1)  [below = 1cm of m0] {m1};
            \node[circle, draw=black, minimum size=7mm, very thick]  (m2)  [below = 1cm of m1] {m2};
            \node (m3)  [below left = 5mm of m2] {};
            \node (m4)  [below right = 5mm of m2] {};
            \draw[->, thick] (m0.south) -- (m1.north) node[midway, left] {\textbf{Accept} ($\monLabYes{\stEnqueueTag{\roleFmt{c}}{\labFmt{req}}{\tyInt}}$)};
            \path[->, dashed] (m0) edge [loop left] node[midway, left] {\shortstack{\textbf{Accept retransmit}\\or \textbf{Reject}}} (m0);
            \draw[->, dashed] (m0.east) -- (fail.west) node[midway, above] {\shortstack{\textbf{Reject and}\\\textbf{close}}};
            \draw[->, thick] (m1.south) -- (m2.north) node[midway, left] {\textbf{Accept} ($\monLabYes{\stSendTag{\roleFmt{r}}{\labFmt{r\_req}}{\tyInt}}$)};
            \path[->, dashed] (m1) edge [loop left] node[midway, left] {\shortstack{\textbf{Accept retransmit}\\or \textbf{Reject}}} (m0);
            \draw[->, dashed] (m1.east) .. controls +(right:10mm) .. (fail.south) node[midway, above left=1mm and -2mm] {\shortstack{\textbf{Reject and}\\\textbf{close}}};
            \draw[->, thick] (m2.south west) -- (m3.north east) {};
            \draw[->, thick] (m2.south east) -- (m4.north west) {};
            \path[->, dashed] (m2) edge [loop left] node[midway, left] {\shortstack{\textbf{Accept retransmit}\\or \textbf{Reject}}} (m0);
            \draw[->, dashed] (m2.east) .. controls +(right:13mm) .. (fail.south) node[midway, above left=-2mm and 0mm] {\shortstack{\textbf{Reject and}\\\textbf{close}}};
        \end{tikzpicture}
        \captionof{figure}{Initial states and transitions of the TCP-oriented
          monitor generated from \Cref{fig:st-monitor-visual-graph}.}
        \label{fig:implementation-tcp-monitor-visual}
      \end{minipage}
      \hfill%
      \begin{minipage}{0.625\linewidth}
    \centering
    \includegraphics[width=0.9\linewidth]{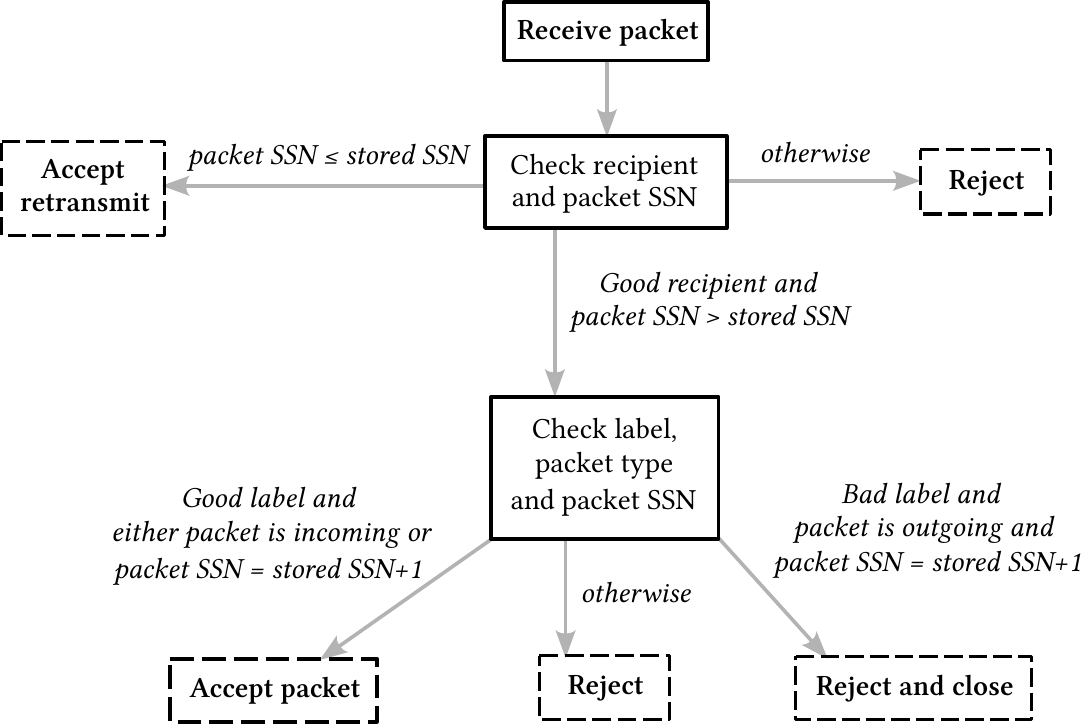}
    \captionof{figure}{A flowchart showing how the TCP-oriented monitor processes incoming packets.}
    \label{fig:tcp-monitor-decision}
      \end{minipage}
\end{figure}

With this approach, our TCP-enabled monitors can handle packet loss (since TCP
eventually retransmits lost and unacknowledged packets), duplication (treated as
a special case of retransmission), and reordering of packets from the same
sender (by simply ignoring out-of-order packets and waiting for their
retransmission).

\ifecoop
\subparagraph{Limitations.}%
\else
\paragraph{Limitations.}%
\fi
Our approach to monitoring TCP connections has limitations:

\begin{itemize}
	\item It requires TCP-oriented monitors to let all \texttt{ACK} packets pass
	through the network. A malicious end host might abuse this to flood the
	network with spoofed ACK messages, which the monitors would not reject. %
	This is a common risk in networks with TCP services, and it may require
	mitigations such as rate limiting to defend against TCP-based DoS
	(denial-of-service) attacks.%
	\item
	Since the TCP-oriented monitors do not reject packets with low session
	sequence numbers, a malicious or faulty sender could in principle send
	messages with low session sequence numbers through the network without having
	them rejected at the border.
		To mitigate this, monitors could rate-limit forwarding of such packets, since they can be treated as retransmissions. The rate-limiting design is orthogonal to our monitoring logic, so we leave it to future work. Packets with lower-than-expected session sequence numbers are ignored by recipient-side programs using the \toolName-generated APIs (\Cref{sec:implementation-st-api}).%
		\item The current version of \toolName does not support packet fragmentation; we discuss remedies in \Cref{sec:conclusion-futurework}.
\end{itemize}

\subsection{API Generation for Session-Monitored Applications}
\label{sec:implementation}
\label{sec:implementation-st-api}

Because our monitors require packets to carry the session header described in
\Cref{sec:implementation-st-headers,sec:implementation-practical-protocols},
\toolName generates an API for writing monitor-compliant applications.
The API hides session-header details and exposes human-readable message-label constants derived from the input session type, rather than numeric label IDs.
Our current prototype targets Python and supports the send/receive style in \Cref{fig:our-approach}.
\change{\#123B: clarify that the API is not ``Scribble-style''}{%
  By design, the API generated by \toolName is minimalistic and does not
  enforce the ordering of send/receive operations specified in the input session type.
  We take advantage of this in our
  evaluation (\Cref{sec:evaluation-qualitative-reject}) to write programs that do \emph{not} follow a
  session type and show that \toolName monitors correctly reject their messages.%
}%

The generated API is based on a \ttt{SessionManager} class that instantiates
\ttt{Session} objects, each representing one session.
Each \ttt{Session} is created with a protocol and a session ID.
The API also maintains a queue of incoming messages, from which
\texttt{Session} objects dequeue via \texttt{recvMsg()}. This prevents the host
program from incorrectly dropping messages (e.g., when a message is delivered
earlier than expected).

\ifecoop
\subparagraph{Automatic Session ID Propagation.} %
\else
\paragraph{Automatic Session ID Propagation.} %
\fi
Beyond basic support for sends and receives, the \ttt{SessionManager} also handles automatic
\textit{propagation} of session IDs.
A program can create a \ttt{Session} without an ID; it then adopts the ID from the first
incoming message that carries an ID not already in use on that \ttt{SessionManager}.
Subsequent messages sent by that \ttt{Session} propagate that ID, and the \ttt{Session}
dequeues only messages with that ID.%

To support automatic ID propagation, the protocol must include an \emph{initiator}
that creates a new session ID and sends the first message(s) to one or more participants, who then
learn and propagate that ID.
All of the multiparty protocols we use in our evaluation (\Cref{sec:evaluation})
follow this pattern.
\section{Proving the Correctness of \toolName Monitors} %
\label{sec:formal-model}

In this section, we establish the correctness of \toolName's monitor synthesis,
i.e., the monitor state machine and MAT outlined at the end of \Cref{sec:implementation-st-headers}. %
We focus on two key challenges for defining the ``core logic'' of our network-layer monitors:
\begin{enumerate}
	\item\label{item:formal:buffer} Monitors must correctly accept or reject
	packets \emph{immediately}, without resorting to buffering -- which is assumed in previous work on session monitoring %
    \change{\#123B: add citations}{%
        \cite{BoChDeHoYo13,BoChDeHoYo17,DBLP:journals/fmsd/DemangeonHHNY15} %
    }%
    but is often infeasible in network devices due to memory limitations.
    \item\label{item:formal:async} Monitors must make correct decisions even
    when packets from different senders are delivered in an order
	that does not match the expectations of the end hosts.
\end{enumerate}

To isolate these challenges, we study monitor synthesis and
correctness in an idealised network where messages are delivered
instantaneously, without fragmentation, loss, duplication, or same-sender reordering; %
as explained in \Cref{sec:implementation-practical-protocols}, \toolName handles
these aspects\footnote{%
    Except packet fragmentation, which is less common in modern networks configured
    with a consistent maximum transmission unit (MTU). Support for fragmentation
    is future work discussed in \Cref{sec:conclusion-futurework}.%
} %
by augmenting the ``core logic'' of the monitors with additional checks and
transitions tailored to TCP as a transport protocol.
Challenges~\ref{item:formal:buffer} and \ref{item:formal:async}, instead,
fundamentally affect monitoring logic, independently of the transport
protocol in use. %

In \Cref{sec:formal-model:session-types-networks} we formalize networks where end-host behaviors are modeled as session types and protocol-violating packets may still reach hosts. In \Cref{sec:formal-model:monitored-st-networks}, monitors block those packets.
\Cref{sec:formal-model:mon-synthesis} then formalizes our monitor synthesis (\Cref{def:monitor-state-st}).
Finally, \Cref{sec:formal-model:correctness} proves soundness: synthesized monitors do not reject traffic when all end hosts follow the enforced protocol (\Cref{thm:net-mon-bisim}).

\subsection{Session Types and Networks}
\label{sec:formal-model:session-types-networks}

In \Cref{def:local-types,def:local-semantics} we model a network end host as a
(local) session type with a \emph{multi-input queue} that stores incoming
messages from multiple senders, while preserving the order of messages from each
sender \cite{DBLP:conf/fsttcs/DemangeonY15}. The idea is that the session type
models the behavior of a message-passing program, while the queue models the
end host's ability to buffer incoming messages (e.g.~in its network stack). This
modeling is standard in the session-types literature, except that
multi-output queues are often used instead of input queues.

\begin{definition}[End host model]%
    \label{def:local-types}%
    The syntax of \textbf{session types with multi-input queues} %
    is:

    \smallskip%
  \centerline{$%
  \begin{array}{r@{\hskip 2mm}c@{\hskip 2mm}l@{\hskip 2mm}l}
    \textstyle%
    \mbox{{Payload type, a.k.a.~sort}} &
    \stS &\coloncolonequals& \tyBool \bnfsep \tyInt \bnfsep \tyString \bnfsep \cdots
    \\[1mm]
    \mbox{{Role}} & \roleP, \roleQ, \roleR &\coloncolonequals& \roleFmt{Info} \bnfsep \roleFmt{Client} \bnfsep \roleFmt{Review} \bnfsep \roleFmt{Details} \bnfsep \cdots
    \\[1mm]
    \mbox{{Session type}} &
    \stT%
    &\coloncolonequals&%
     \stIntSum{i \in I}{\roleP[i]}{\stChoice{\stLab[i]}{\stS[i]} \stSeq \stT[i]}%
     \bnfsep%
     \stExtSum{i \in I}{\roleP}{\stChoice{\stLab[i]}{\stS[i]} \stSeq \stT[i]}%
     \bnfsep \stEnd%
      \bnfsep%
      \stRec{\stRecVar}{\stT}%
      \bnfsep%
      \stRecVar%
     \\[1mm]
   \mbox{Input queue type} &
   \stQ
   &\coloncolonequals&%
    \stQCons{\stQMsg{\roleP}{\stLab}{\stS}}{\stQ}%
    \bnfsep \stQEmpty
    \\[1mm]
    \mbox{Session type with input queue} &
    \stqT
    &\coloncolonequals&%
    \stqPair{\stT}{\stQ}%
   \\[1mm]
  \end{array}
  $}
  \smallskip%

  \noindent%
  where $I \!\neq\! \emptyset$ and the \emph{message labels} $\stLab[i]$ are pairwise distinct. %
  We require session types to be closed %
  and recursion variables to be guarded.
\end{definition}

The type $\stIntSum{i \in I}{\roleP[i]}{\stChoice{\stLab[i]}{\stS[i]} \stSeq \stT[i]}$
represents an \emph{internal choice} where the end host selects one
recipient role $\roleP[i]$ and sends a message with label $\stLab[i]$ carrying payload type $\stS[i]$;
then, the interaction continues as specified in $\stT[i]$.
Dually, $\stExtSum{i \in I}{\roleP}{\stChoice{\stLab[i]}{\stS[i]} \stSeq \stT[i]}$
represents an \emph{external choice} where the end host awaits a message $\stLab[i]$ with payload type $\stS[i]$ from sender $\roleP$;
then, the interaction continues as specified in $\stT[i]$. %
\change{\#123B: clarify roles and payload types lists are just examples}{%
  The lists of possible roles $\roleP, \roleQ, \roleR$ and payload types $\stS$ are provided as examples and can be extended as needed. %
}%
The type $\stEnd$ represents terminated
sessions, while $\stRec{\stRecVar}{\stT}$ and $\stRecVar$ represent recursion. We define $\stqPair{\stT}{\stQ}$ as the
pairing of a session type with a multi-input queue $\stQ$, which may contain
elements of the form $\stQMsg{\roleP}{\stLab}{\stS}$, representing a
message $\stLab$ sent by role $\roleP$ with payload type $\stS$.

\begin{example}
    \label{ex:formal-info-st}
    The formal definition of the session type $\stInfo$ for the role \lstinline|Info| from
    \Cref{fig:InfoLocalSessionType} is the following (for brevity, we shorten role names and message labels):

    \smallskip%
    \noindent%
    $%
        \fontsize{8}{12}\selectfont%
        \stInfo \;=\; \stFmt{\roleFmt{c} \stExtC \labFmt{req}(\tyInt) \stSeq\ \roleFmt{r} \stIntC \labFmt{r\_req}(\tyInt) \stSeq\ \roleFmt{d} \stIntC \labFmt{d\_req}(\tyInt) \stSeq\ \roleFmt{r} \stExtC \labFmt{r\_rsp}(\tyString) \stSeq\ \roleFmt{d} \stExtC \labFmt{d\_rsp}(\tyString) \stSeq\ \roleFmt{c} \stIntC \labFmt{rsp}(\tyString)}
    $%
\end{example}%

\begin{definition}
    \label{def:local-semantics-no-queues}%
    \label{def:local-semantics}\label{def:local-semantics-queues}\label{def:queue-congruence}%
    The labeled transition system (LTS) semantics of session types (without queues) is defined as follows,
    using the labels
    $\stTagNoQueue \coloncolonequals \stOutTag{\roleP}{\stLab}{\stS} \bnfsep \stInTag{\roleP}{\stLab}{\stS}$:

    \smallskip%
    \centerline{\(
      \small%
      \begin{array}{c}
        \Rule
            {k \in I}
            {\transition
                {\stIntSum{i \in I}{\roleP[i]}{\stChoice{\stLab[i]}{\stS[i]} \stSeq \stT[i]}}
                {\stOutTag{\roleP[k]}{\stLab[k]}{\stS[k]}}
                {\stT[k]}}
            {S-IntC}
        \quad
        \Rule
            {k \in I}
            {\transition
                {\stExtSum{i \in I}{\roleP}{\stChoice{\stLab[i]}{\stS[i]} \stSeq \stT[i]}}
                {\stInTag{\roleP}{\stLab[k]}{\stS[k]}}
                {\stT[k]}}
            {S-ExtC}
        \\\\[-5pt]
        \Rule
            {\transition{\subst{\stT}{\stRecVar}{\stRec{\stRecVar}{\stT}}}{\stTagNoQueue}{\stTi}}
            {\transition{\stRec{\stRecVar}{\stT}}{\stTagNoQueue}{\stTi}}
            {S-Rec}
      \end{array}
    \)}%
    \medskip%

    The LTS semantics of session types with input queues is defined as follows:

    \smallskip%
    \centerline{\(
      \small%
      \begin{array}{c}
        \Rule
            {\transition{\stT}{\stOutTag{\roleP}{\stLab}{\stS}}{\stTi}}
            {\transition
                {\stqPair{\stT}{\stQ}}
                {\stSendTag{\roleP}{\stLab}{\stS}}
                {\stqPair{\stTi}{\stQ}}}
            {SQ-Send}
        \quad
        \Rule
            {\transition{\stT}{\stInTag{\roleP}{\stLab}{\stS}}{\stTi}
              \quad \stQ \equiv \stQCons{\stQMsg{\roleP}{\stLab}{\stS}}{\stQi}}
            {\transition
                {\stqPair{\stT}{\stQ}}
                {\stIntTag{\dequeueTag{\roleP}{\stLab}{\stS}}}
                {\stqPair{\stTi}{\stQi}}}
            {SQ-Deq}
        \\\\[-10pt]
        \Rule
            {}
            {\transition
                {\stqPair{\stT}{\stQ}}
                {\stEnqueueTag{\roleP}{\stLab}{\stS}}
                {\stqPair{\stT}{\stQCons{\stQ}{\stQMsg{\roleP}{\stLab}{\stS}}}}}
            {SQ-Recv}
      \end{array}
    \)}%
    \medskip%

    \noindent%
    where $\tau$ denotes an internal transition that does not synchronize with others; in \textsc{SQ-Deq}, $\equiv$ is the smallest congruence s.t.
    $\roleP \neq \roleQ$ implies\; %
        $\stQCons{\stQCons{\stQMsg{\roleP}{\stLab}{\stS}}{\stQMsg{\roleQ}{\stLabi}{\stSi}}}{\stQ} \equiv \stQCons{\stQCons{\stQMsg{\roleQ}{\stLabi}{\stSi}}{\stQMsg{\roleP}{\stLab}{\stS}}}{\stQ}$.\\%
    \change{\#123C: missing definition of $\stExtTag$}{%
        We use the symbol $\stExtTag$ to denote either a label
        $\stSendTag{\roleP}{\stLab}{\stS}$ (send message, emitted by rule \textsc{SQ-Send}) or
        $\stEnqueueTag{\roleP}{\stLab}{\stS}$ (enqueue message, emitted by rule
        \textsc{SQ-Recv}).%
    }%
\end{definition}

In \Cref{def:local-semantics-no-queues}, a session type $\stT$ transitions
by emitting labels representing an internal choice
$\stOutTag{\roleP}{\stLab}{\stS}$ or an external choice
$\stInTag{\roleP}{\stLab}{\stS}$, by rules
\textsc{S-IntC}, \textsc{S-ExtC}, and \textsc{S-Rec}. %
For instance, for the session type $\stInfo$ in \Cref{ex:formal-info-st}, these rules yield the
transition system in \Cref{fig:st-info-visual-graph}.

When $\stT$ is composed with a multi-input queue, %
rule \textsc{SQ-Send} says that internal choices of $\stT$ enable a ``send''
action $\stSendTag{\roleP}{\stLab}{\stS}$. %
Rule \textsc{SQ-Deq} uses a standard queue congruence $\equiv$ allowing for
swapping two queued messages with different senders: this enables the selective
dequeuing and consumption of the oldest queued message from each sender.
Then, rule \textsc{SQ-Deq}
says that the session type can consume the oldest message from $\roleP$ (from the queue head, via congruence $\equiv$)
with an internal action (a.k.a.~``$\tau$-action'') $\stIntTag{\dequeueTag{\roleP}{\stLab}{\stS}}$;
this can happen only if $\stT$ is an external choice that awaits a message $\stLab$ from $\roleP$ with payload type $\stS$,
and the oldest queued message from $\roleP$ satisfies these conditions. %
Finally, rule \textsc{SQ-Recv} allows an arbitrary message $\stQMsg{\roleP}{\stLab}{\stS}$
to be received from the ``outside world'' and appended to the queue,
via a ``receive'' action $\stEnqueueTag{\roleP}{\stLab}{\stS}$.

\begin{example}[Semantics of session types with input queues]
    \label{ex:session-type-queue-reduct}
    Consider the session type from \Cref{ex:formal-info-st}, paired with
    an initially empty input queue:

    \smallskip%
    \centerline{\(%
    \begin{array}{c}
        \small%
        \stqPair
            {\roleFmt{c} \stExtC \labFmt{req}(\tyInt) \stSeq\ \roleFmt{r} \stIntC \labFmt{r\_req}(\tyInt) \stSeq\ \roleFmt{d} \stIntC \labFmt{d\_req}(\tyInt) \stSeq\ \roleFmt{r} \stExtC \labFmt{r\_rsp}(\tyString) \stSeq\ \roleFmt{d} \stExtC \labFmt{d\_rsp}(\tyString) \stSeq\ \roleFmt{c} \stIntC \labFmt{rsp}(\tyString)}
            {\stQEmpty}
    \end{array}
    \)}%
    \smallskip%

    Suppose that $\roleFmt{c}$ sends the expected message with label
    $\labFmt{req}$ to the end host modeled by this session type. First,
    the message is moved to the input queue by rule \tsc{SQ-Recv}, via a transition $\transition{}{\stEnqueueTag{\roleFmt{c}}{\labFmt{req}}{\tyInt}}{}$,
    resulting in the following session type with queue:

    \smallskip%
    \centerline{\(%
    \begin{array}{c}
        \fontsize{8}{12}\selectfont%
        \stqPair
            {\roleFmt{c} \stExtC \labFmt{req}(\tyInt) \stSeq\,\roleFmt{r} \stIntC \labFmt{r\_req}(\tyInt) \stSeq\,\roleFmt{d} \stIntC \labFmt{d\_req}(\tyInt) \stSeq\,\roleFmt{r} \stExtC \labFmt{r\_rsp}(\tyString) \stSeq\,\roleFmt{d} \stExtC \labFmt{d\_rsp}(\tyString) \stSeq\,\roleFmt{c} \stIntC \labFmt{rsp}(\tyString)}
            {\stQCons{\stQMsg{\roleFmt{c}}{\labFmt{req}}{\tyInt}}{\stQEmpty}}
    \end{array}
    \)}%
    \smallskip%

    Since the head of the input queue now contains a message that matches one of the cases
    in the topmost external choice, rule \tsc{SQ-Deq} enables a transition $\transition{}{\stIntTag{\dequeueTag{\roleFmt{c}}{\labFmt{req}}{\tyInt}}}{}$
    which consumes the queued message, resulting in:

    \smallskip%
    \centerline{\(%
    \begin{array}{c}
        \small%
        \stqPair
            {\roleFmt{r} \stIntC \labFmt{r\_req}(\tyInt) \stSeq\ \roleFmt{d} \stIntC \labFmt{d\_req}(\tyInt) \stSeq\ \roleFmt{r} \stExtC \labFmt{r\_rsp}(\tyString) \stSeq\ \roleFmt{d} \stExtC \labFmt{d\_rsp}(\tyString) \stSeq\ \roleFmt{c} \stIntC \labFmt{rsp}(\tyString)}
            {\stQEmpty}
    \end{array}
    \)}%
    \smallskip%

    We can then immediately fire the internal choices by applying~\tsc{SQ-Send}
    twice, with transitions $\transition{}{\stSendTag{\roleFmt{r}}{\labFmt{r\_req}}{\tyInt}}{}$ and $\transition{}{\stSendTag{\roleFmt{d}}{\labFmt{d\_req}}{\tyInt}}{}$,
    resulting in:

    \smallskip%
    \centerline{\(%
    \begin{array}{c}
        \small%
        \stqPair
            {\roleFmt{r} \stExtC \labFmt{r\_rsp}(\tyString) \stSeq\ \roleFmt{d} \stExtC \labFmt{d\_rsp}(\tyString) \stSeq\ \roleFmt{c} \stIntC \labFmt{rsp}(\tyString)}
            {\stQEmpty}
    \end{array}
    \)}%
    \smallskip%

    Now, suppose role $\roleFmt{d}$ sends the expected message with label
    $\labFmt{r\_rsp}$ first. We can enqueue the message with rule \tsc{SQ-Recv}
    and transition
    $\transition{}{\stEnqueueTag{\roleFmt{d}}{\labFmt{d\_rsp}}{\tyString}}{}$,
    leading to the following configuration. Note that the message from
    $\roleFmt{d}$ is at the head of the queue, but cannot be consumed yet
    because it does not match any of the cases in the topmost external choice:

    \smallskip%
    \centerline{\(%
    \begin{array}{c}
        \small%
        \stqPair
            {\roleFmt{r} \stExtC \labFmt{r\_rsp}(\tyString) \stSeq\ \roleFmt{d} \stExtC \labFmt{d\_rsp}(\tyString) \stSeq\ \roleFmt{c} \stIntC \labFmt{rsp}(\tyString)}
            {\stQCons{\stQMsg{\roleFmt{d}}{\labFmt{d\_rsp}}{\tyString}}{\stQEmpty}}
    \end{array}
    \)}%
    \smallskip%

    Once we receive the response message from $\roleFmt{r}$ and enqueue it
    with~\tsc{SQ-Recv}, we can consume both messages by applying~\tsc{SQ-Deq}
    twice. The first application of \tsc{SQ-Deq} below uses queue congruence
    $\equiv$ (\Cref{def:queue-congruence}) to swap the two messages in the
    queue (since they have different senders) and bring the one from
    $\roleFmt{r}$ to the head of the queue, enabling its consumption.

    \smallskip%
    \centerline{\(%
    \begin{array}{c}
        \small%
        \stqPair
            {\roleFmt{r} \stExtC \labFmt{r\_rsp}(\tyString) \stSeq\ \roleFmt{d} \stExtC \labFmt{d\_rsp}(\tyString) \stSeq\ \roleFmt{c} \stIntC \labFmt{rsp}(\tyString)}
            {\stQCons{\stQMsg{\roleFmt{d}}{\labFmt{d\_rsp}}{\tyString}}{\stQEmpty}}\\
        \phantom{\quad \tx{(\tsc{SQ-Recv})}} \transition{}{\stEnqueueTag{\roleFmt{r}}{\labFmt{r\_rsp}}{\tyString}}{} \quad \tx{(\tsc{SQ-Recv})}\\
        \stqPair
            {\roleFmt{r} \stExtC \labFmt{r\_rsp}(\tyString) \stSeq\ \roleFmt{d} \stExtC \labFmt{d\_rsp}(\tyString) \stSeq\ \roleFmt{c} \stIntC \labFmt{rsp}(\tyString)}
            {\stQCons{\stQMsg{\roleFmt{d}}{\labFmt{d\_rsp}}{\tyString}}{\stQCons{\stQMsg{\roleFmt{r}}{\labFmt{r\_rsp}}{\tyString}}{\stQEmpty}}}\\
        \phantom{\quad \tx{(\tsc{SQ-Deq}, \tsc{SQ-Deq})}} \transition{}{\stIntTag{\dequeueTag{\roleFmt{r}}{\labFmt{r\_rsp}}{\tyString}}}{} \transition{}{\stIntTag{\dequeueTag{\roleFmt{d}}{\labFmt{d\_rsp}}{\tyString}}}{} \quad \tx{(\tsc{SQ-Deq}, \tsc{SQ-Deq})}\\
        \stqPair
            {\roleFmt{c} \stIntC \labFmt{rsp}(\tyString)}
            {\stQEmpty}
    \end{array}
    \)}%
    \smallskip%

    Finally, we send the response to $\roleFmt{c}$ and end the protocol:

    \smallskip%
    \hfill$\small%
        \stqPair{\roleFmt{c} \stIntC \labFmt{rsp}(\tyString)}{\stQEmpty}
        \transition{}{\stIntTag{\stSendTag{\roleFmt{c}}{\labFmt{rsp}}{\tyString}}}{}
        \stqPair{\stEnd}{\stQEmpty} \quad \tx{(\tsc{SQ-Send})}
    $%
\end{example}

\begin{example}[Stuck session types and queues due to bad messages]
    \label{ex:bad-messages-unmonitored-networks}%
    Rule \textsc{SQ-Recv} in
    \Cref{def:local-semantics-no-queues} allows enqueuing a message that
    the session type can never consume: this models the case
    where an unexpected message is delivered to the end host from the
    surrounding network. %
    E.g., consider this configuration from \Cref{ex:session-type-queue-reduct}:

    \smallskip%
    \centerline{\(%
    \begin{array}{c}
        \small%
        \stqPair
            {\roleFmt{r} \stExtC \labFmt{r\_rsp}(\tyString) \stSeq\ \roleFmt{d} \stExtC \labFmt{d\_rsp}(\tyString) \stSeq\ \roleFmt{c} \stIntC \labFmt{rsp}(\tyString)}
            {\stQEmpty}
    \end{array}
    \)}%
    \smallskip%
    If the surrounding network now delivers a message with label
    $\labFmt{unexpected}$ from role $\roleFmt{r}$, we can enqueue it using
    \textsc{SQ-Recv}, with transition
    $\transition{}{\stEnqueueTag{\roleFmt{r}}{\labFmt{unexpected}}{\tyInt}}{}$
    leading to:

    \smallskip%
    \centerline{\(%
    \begin{array}{c}
        \small%
        \stqPair
            {\roleFmt{r} \stExtC \labFmt{r\_rsp}(\tyString) \stSeq\ \roleFmt{d} \stExtC \labFmt{d\_rsp}(\tyString) \stSeq\ \roleFmt{c} \stIntC \labFmt{rsp}(\tyString)}
            {\stQCons{\stQMsg{\roleFmt{r}}{\labFmt{unexpected}}{\tyInt}}{\stQEmpty}}
    \end{array}
    \)}%
    \smallskip%
    This session type with queue is now stuck: it can only proceed by dequeuing a message
    with label $\labFmt{r\_rsp}$ from $\roleFmt{r}$, but the oldest queued message from $\roleFmt{r}$ has label
    $\labFmt{unexpected}$. If more messages are queued, they will not be consumed either.
\end{example}

We model networks in \Cref{def:network} as parallel compositions of
roles with an end host behavior represented as a session type with an input queue.

\begin{definition}[Network]
    \label{def:network}
    \label{def:network-semantics}
    We define a network as:

    \smallskip\centerline{\(%
        \netN \;\;\coloncolonequals\;\;
        \netN \netPar \netNi
        \;\bnfsep\;
        \netPair{\roleP[]}{\stqT[]}
        \qquad\text{where $\roleP$ does not occur in $\stqT$}
    \)}%

    \noindent%
    with the following LTS semantics, using the labels
    \(\netLab \;\coloncolonequals\; \netIOLab{\roleP}{\stExtTag} \bnfsep \netDequeueLab{\roleP}{\roleQ}{\stLab}{\stS} \bnfsep \netCommLab{\roleP}{\roleQ}{\stLab}{\stS}\)\;
    with $\stExtTag$ from \Cref{def:local-semantics-queues}: (for brevity we omit the symmetric rules of \textsc{Net-Par} and \textsc{Net-Comm}):

    \smallskip%
    \centerline{\(%
    \begin{array}{c}
        \Rule
            {\transition{\stqT}{\stExtTag}{\stqTi}}
            {\transition
                {\netPair{\roleP}{\stqT}}
                {\netIOLab{\roleP}{\stExtTag}}
                {\netPair{\roleP}{\stqTi}}
            }
            {Net-$\stExtTag$}
        \qquad
        \Rule
            {\transition{\stqT}{\stIntTag{\dequeueTag{\roleQ}{\stLab}{\stS}}}{\stqTi}}
            {\transition
                {\netPair{\roleP}{\stqT}}
                {\netDequeueLab{\roleP}{\roleQ}{\stLab}{\stS}}
                {\netPair{\roleP}{\stqTi}}
            }
            {Net-Deq}\\
        \Rule
            {
                \transition{\netN[1]}{\netLab}{\netNi[1]}
            }
            {\transition{\netN[1] \netPar \netN[2]\;}{\netLab}{\;\netNi[1] \netPar \netN[2]}}
            {Net-Par}
        \qquad
        \Rule
            {
                \transition{\netN[1]}{\netIOLab{\roleP}{\stSendTag{\roleQ}{\stLab}{\stS}}}{\netNi[1]}
                \quad
                \transition{\netN[2]}{\netIOLab{\roleQ}{\stEnqueueTag{\roleP}{\stLab}{\stS}}}{\netNi[2]}
            }
            {\transition{\netN[1] \netPar \netN[2]\;}{\netCommLab{\roleP}{\roleQ}{\stLab}{\stS}}{\;\netNi[1] \netPar \netNi[2]}}
            {Net-Comm}
    \end{array}
    \)}%
    \vspace{-6mm}%
\end{definition}

In \Cref{def:network-semantics} above, rules \textsc{Net-$\stExtTag$} and
\textsc{Net-Deq} decorate a transition of a session type with input queue
$\stqT$ by including the role $\roleP$ that emitted the transition; %
specifically, \textsc{Net-$\stExtTag$} is used when $\stqT$
emits or enqueues a message (via rules~\textsc{SQ-Send} or \textsc{SQ-Recv} in \Cref{def:local-semantics-queues}),
while \textsc{Net-Deq} is used when $\stqT$ internally consumes a message from its queue
(via rule~\textsc{SQ-Deq} in \Cref{def:local-semantics-queues}). %
Rule \textsc{Net-Comm} says that if in the sub-network $\netN[1]$ there is $\roleP$
sending a message to $\roleQ$, and in the sub-network $\netN[2]$ there is $\roleQ$
that can receive that message, then the network advances with a communication
transition $\netCommLab{\roleP}{\roleQ}{\stLab}{\stS}$. %
The rightmost premise of rule~\textsc{Net-Comm} implies that the message
from $\roleP$ is added to $\roleQ$'s input queue, by \textsc{Net-$\stExtTag$} and
\textsc{SQ-Recv} in \Cref{def:local-semantics-queues}. %
Also, by the same rules, \emph{any} $\netN[2]$ that includes $\roleQ$ can always receive
\emph{any} message from \emph{any} $\roleP$, i.e., rule
\textsc{Net-Comm} allows invalid messages to be sent/received between end
hosts; see \Cref{ex:bad-network-message}.

\begin{example}[Bad messages in an unmonitored network]
    \label{ex:bad-network-message}%
    Consider the following network with roles $\roleP$, $\roleQ$, and $\roleR$: (for brevity, we omit the payload types and $\stEnd$s)

    \smallskip%
    \centerline{\(%
    \fontsize{8}{12}\selectfont
    \begin{array}{c}
        \netPair{\roleP}{\stqPair{\stT[\roleP]}{\stQEmpty}} \;\netPar\;%
        \netPair{\roleQ}{\stqPair{\stT[\roleQ]}{\stQEmpty}} \;\netPar\;%
        \netPair{\roleR}{\stqPair{\stT[\roleR]}{\stQEmpty}}%
        \qquad\text{where}\quad%
        \small%
        \stT[\roleP] =
           \stFmt{\sum} \left\{\begin{array}{@{}l@{}}
                        \roleQ \stIntC \stLabFmt{a} \\
                        \roleQ \stIntC \stLabFmt{b}
            \end{array}\right.;
        \quad%
        \stT[\roleQ] =
            \roleP \stExtC \left\{\begin{array}{@{}l@{}}
                        \stLabFmt{a} \stSeq \roleR \stExtC \stLabFmt{a'} \\
                        \stLabFmt{b} \stSeq \roleR \stExtC \stLabFmt{b'}
            \end{array}\right.;
        \quad%
        \stT[\roleR] = \roleQ \stIntC \stLabFmt{a'}
    \end{array}
    \)}%
    \smallskip%

    \noindent%
    Here, $\stT[\roleP]$ can send either $\stLabFmt{a}$ or $\stLabFmt{b}$ to $\roleQ$.
    Meanwhile, $\stT[\roleQ]$ expects to receive either $\stLabFmt{a}$ from $\roleP$ and
    then $\stLabFmt{a'}$ from $\roleR$, or $\stLabFmt{b}$ from $\roleP$ and then
    $\stLabFmt{b'}$ from $\roleR$. %
    Instead, $\stT[\roleR]$ just sends $\stLabFmt{a'}$ to $\roleQ$.

    By \Cref{def:network-semantics} the network could reduce as follows,
    with $\stLabFmt{a}$ sent by $\roleP$ and enqueued by $\roleQ$:

    \smallskip%
    \centerline{\(%
    \small%
    \begin{array}{c}
        \Rule
            {\Rule
                {\transition
                    {\netPair{\roleP}{\stqPair{\stT[\roleP]}{\stQEmpty}}}
                    {\netIOLab{\roleP}{\stSendTag{\roleQ}{\stLabFmt{\stLabFmt{a}}}{}}}
                    {\netPair{\roleP}{\stqPair{\stEnd}{\stQEmpty}}}
                \qquad
                \transition
                    {\netPair{\roleQ}{\stqPair{\stT[\roleQ]}{\stQEmpty}}}
                    {\netIOLab{\roleQ}{\stEnqueueTag{\roleP}{\stLabFmt{\stLabFmt{a}}}{}}}
                    {\netPair{\roleQ}{\stqPair{\stT[\roleQ]}{\stQMsg{\roleP}{\stLabFmt{\stLabFmt{a}}}{}}}}}
                {\transition
                    {\netPair{\roleP}{\stqPair{\stT[\roleP]}{\stQEmpty}} \netPar \netPair{\roleQ}{\stqPair{\stT[\roleQ]}{\stQEmpty}}\;}
                    {\netCommLab{\roleP}{\roleQ}{\stLabFmt{\stLabFmt{a}}}{}}
                    {\;\netPair{\roleP}{\stqPair{\stEnd}{\stQEmpty}} \netPar \netPair{\roleQ}{\stqPair{\stT[\roleQ]}{\stQMsg{\roleP}{\stLabFmt{\stLabFmt{a}}}{}}}}}
                {Net-Comm}}
            {\transition
                {\netPair{\roleP}{\stqPair{\stT[\roleP]}{\stQEmpty}} \netPar \netPair{\roleQ}{\stqPair{\stT[\roleQ]}{\stQEmpty}} \netPar \netPair{\roleR}{\stqPair{\stT[\roleR]}{\stQEmpty}}}
                {\netCommLab{\roleP}{\roleQ}{\stLabFmt{\stLabFmt{a}}}{}}
                {\netPair{\roleP}{\stqPair{\stEnd}{\stQEmpty}} \netPar \netPair{\roleQ}{\stqPair{\stT[\roleQ]}{\stQMsg{\roleP}{\stLabFmt{\stLabFmt{a}}}{}}} \netPar \netPair{\roleR}{\stqPair{\stT[\roleR]}{\stQEmpty}}}}
            {Net-Par}
    \end{array}
    \)}%
    \smallskip%

    \noindent%
    Then, $\roleQ$ can consume the message $\stLabFmt{a}$ from $\roleP$, and
    later enqueue and consume $\stLabFmt{a'}$ from $\roleR$: in this case, the
    network reaches a successful final state where every end host is $\stEnd$
    with an empty queue.
    Similarly, if $\roleR$ sends $\stLabFmt{a'}$ to $\roleQ$ first, and $\roleP$
    sends $\stLabFmt{a}$ afterwards, then $\roleQ$ can consume both messages (like the last transitions of \Cref{ex:session-type-queue-reduct}) reaching success. Therefore, in both cases, $\stLabFmt{a'}$ from
    $\roleR$ is a ``good message'' for $\roleQ$.
    However, if $\roleP$ sends message $\stLabFmt{b}$ to $\roleQ$, then
    $\stT[\roleQ]$ consumes it and enters the branch where it expects
    $\stLabFmt{b'}$ from $\roleR$---but $\roleR$ sends $\stLabFmt{a'}$
    instead:

    \newcommand{\startState}{%
        \netPair{\roleP}{\stqPair{\stT[\roleP]}{\stQEmpty}} \netPar \netPair{\roleQ}{\stqPair{\stT[\roleQ]}{\stQEmpty}} \netPar \netPair{\roleR}{\stqPair{\stT[\roleR]}{\stQEmpty}}%
    }%
    \smallskip%
    \centerline{\(%
    \fontsize{8}{12}\selectfont%
    \begin{array}{l@{\;}l}
        \transition
            {\startState\;}
            {\netCommLab{\roleP}{\roleQ}{\stLabFmt{\stLabFmt{b}}}{}}
            {\;\netPair{\roleP}{\stqPair{\stEnd}{\stQEmpty}} \netPar \netPair{\roleQ}{\stqPair{\stT[\roleQ]}{\stQMsg{\roleP}{\stLabFmt{\stLabFmt{b}}}{}}} \netPar \netPair{\roleR}{\stqPair{\stT[\roleR]}{\stQEmpty}}}
            &
            \text{\tiny(by \textsc{Net-Comm}, \textsc{Net-Par})}
            \\
        \transition
            {\phantom{\startState}\;}
            {\netDequeueLab{\roleQ}{\roleP}{\stLabFmt{b}}{}}
            {\;\;\,\netPair{\roleP}{\stqPair{\stEnd}{\stQEmpty}} \netPar \netPair{\roleQ}{\stqPair{\roleR \stExtC \stLabFmt{b'}}{\stQEmpty}} \netPar \netPair{\roleR}{\stqPair{\stT[\roleR]}{\stQEmpty}}}
            &
            \text{\tiny(by \textsc{Net-Deq}, \textsc{Net-Par})}
            \\
        \transition
            {\phantom{\startState}\;}
            {\netCommLab{\roleR}{\roleQ}{\stLabFmt{a'}}{}}
            {\;\netPair{\roleP}{\stqPair{\stEnd}{\stQEmpty}} \netPar \netPair{\roleQ}{\stqPair{\roleR \stExtC \stLabFmt{b'}}{\stQMsg{\roleR}{\stLabFmt{a'}}{}}} \netPar \netPair{\roleR}{\stqPair{\stEnd}{\stQEmpty}}}
            &
            \text{\tiny(by \textsc{Net-Comm}, \textsc{Net-Par})}
    \end{array}
    \)}%
    \smallskip%

    \noindent%
    Therefore, $\roleQ$ cannot consume $\roleR$'s message and gets stuck --
    similarly to \Cref{ex:bad-messages-unmonitored-networks}. %
    Note that in this execution, unlike the cases above, $\stLabFmt{a'}$ from
    $\roleR$ is a ``bad message'' for $\roleQ$.
\end{example}

\subsection{Monitored Networks}
\label{sec:formal-model:monitored-st-networks}

To model the pairing of an end host with a monitor that intercepts all its communications, %
in \Cref{def:st-monitor} below we combine a session type with queue $\stqT$ and a
generic \emph{monitor} $\monM$. For now we only assume that $\monM$ has an LTS
semantics with labels of the form
$\monLabYes{\stSendTag{\roleP}{\stLab}{\stS}}$ /
$\monLabYes{\stEnqueueTag{\roleP}{\stLab}{\stS}}$ to signal that $\monM$ accepts
the corresponding send/receive action by $\stqT$,
and $\monLabNo{\stSendTag{\roleP}{\stLab}{\stS}}$ /
$\monLabNo{\stEnqueueTag{\roleP}{\stLab}{\stS}}$ to signal that $\monM$ rejects
them. %
(We present a concrete instantiation of $\monM$ in \Cref{def:monitor-state-st}
below.)

\begin{definition}
    \label{def:st-monitor}%
    \label{def:st-monitor-semantics}%
    We define a \textbf{monitored session type (with input queue)} as:

    \smallskip%
    \centerline{$%
    \begin{array}{r@{\hskip 2mm}c@{\hskip 2mm}l@{\hskip 2mm}l}
        \textstyle%
        \monT%
            &\coloncolonequals&%
            \mstPair{\stqT}{\monM}
            &\\[1mm]
    \end{array}
    $}%
    \smallskip%
    \noindent%
    where $\monM$ is a monitor. %
    We also define the monitored session type semantics:

    \smallskip%
    \centerline{\(%
    \begin{array}{c}
        \Rule
            {\transition
                {\stqT}
                {\stExtTag}
                {\stqTi}
            \qquad
            \transition
                {\monM}
                {\monLabYes{\stExtTag}}
                {\monMi}}
            {\transition
                {\mstPair{\stqT}{\monM}}
                {\monLabYes{\stExtTag}}
                {\mstPair{\stqTi}{\monMi}}}
            {M-Good}\qquad
       \Rule
           {\transition
               {\stqT}
               {\stOutTag{\roleP}{\stLab}{\stS}}
               {\stqTi}
           \quad
           \transition
               {\monM}
               {\monLabNo{\stOutTag{\roleP}{\stLab}{\stS}}}
               {\monMi}}
           {\transition
               {\mstPair{\stqT}{\monM}}
               {\monLabNo{\stOutTag{\roleP}{\stLab}{\stS}}}
               {\mstPair{\stqTi}{\monMi}}}
           {M-BadOut}\\
        \Rule
            {\transition
                {\monM}
                {\monLabNo{\stInTag{\roleP}{\stLab}{\stS}}}
                {\monMi}}
            {\transition
                {\mstPair{\stqT}{\monM}}
                {\monLabNo{\stInTag{\roleP}{\stLab}{\stS}}}
                {\mstPair{\stqT}{\monMi}}}
            {M-BadIn}\qquad
        \Rule
            {\transition
                {\stqT}{\stIntTag{\dequeueTag{\roleP}{\stLab}{\stS}}}{\stqTi}}
            {\transition
                {\mstPair{\stqT}{\monM}}
                {\monLabSilent{\dequeueTag{\roleP}{\stLab}{\stS}}}
                {\mstPair{\stqTi}{\monM}}}
            {M-Dequeue}
    \end{array}
    \)}%
    \vspace{-5mm}%
\end{definition}

By rule \textsc{M-Good} in \Cref{def:st-monitor-semantics}, the pair
$\mstPair{\stqT}{\monM}$ performs a transition $\monLabYes{\stExtTag}$ if
$\monM$ explicitly accepts the send/receive action $\stExtTag$ emitted by
$\stqT$. %
Rule \textsc{M-BadOut} says that $\monM$ can reject and drop a message sent by
$\stqT$. Rule \textsc{M-BadIn} says that $\monM$ can reject and drop an incoming
message before it lands in $\stqT$'s input queue. %
By rule \textsc{M-Dequeue}, $\monM$ cannot interfere with the
internal action that $\stqT$ performs when consuming a message from its input
queue. %
In \Cref{def:monitored-network} below we monitor networks
(\Cref{def:network}), by adding a monitor to each end host.

\begin{definition}[Monitored Network]
    \label{def:monitored-network}
    \label{def:monitored-network-semantics}
    We define a monitored network as:

    \smallskip\centerline{\(%
        \mnetN \;\;\coloncolonequals\;\;
        \mnetN \mnetPar \mnetNi
        \;\bnfsep\;
        \mnetPair{\roleP[]}{\monT[]}
        \qquad\text{where $\roleP$ does not occur in $\monT$}
    \)}%
    \smallskip

    \noindent%
    with the following semantics, using the labels
    \(\mnetLab \coloncolonequals \mnetIOLabYes{\roleP}{\stExtTag} \bnfsep \mnetIOLabNo{\roleP}{\stExtTag} \bnfsep \mnetDequeueLab{\roleP}{\roleQ}{\stLab}{\stS} \bnfsep \mnetCommLab{\roleP}{\roleQ}{\stLab}{\stS}\)\;
    (for brevity we omit the symmetric rules of \textsc{MNet-Par} and \textsc{MNet-Comm}):

    \smallskip%
    \centerline{\(%
    \small%
    \begin{array}{c}
        \Rule
            {\transition{\monT}{\monLabYes{\stExtTag}}{\monTi}}
            {\transition
                {\mnetPair{\roleP}{\monT}}
                {\mnetIOLabYes{\roleP}{\stExtTag}}
                {\netPair{\roleP}{\monTi}}
            }
            {MNet-$\stExtTag$-Good}
        \quad%
        \Rule
            {\transition{\monT}{\monLabNo{\stExtTag}}{\monTi}}
            {\transition
                {\mnetPair{\roleP}{\monT}}
                {\mnetIOLabNo{\roleP}{\stExtTag}}
                {\netPair{\roleP}{\monTi}}
            }
            {MNet-$\stExtTag$-Bad}
        \quad%
        \Rule
            {
                \transition{\mnetN[1]}{\mnetLab}{\mnetNi[1]}
            }
            {\transition{\mnetN[1] \mnetPar \mnetN[2]\;}{\mnetLab}{\;\mnetNi[1] \mnetPar \mnetN[2]}}
            {MNet-Par}
        \\
        \Rule
            {\transition{\monT}{\monLabSilent{\dequeueTag{\roleQ}{\stLab}{\stS}}}{\monTi}}
            {\transition
                {\netPair{\roleP}{\monT}}
                {\mnetDequeueLab{\roleP}{\roleQ}{\stLab}{\stS}}
                {\netPair{\roleP}{\monTi}}
            }
            {MNet-Deq}\quad%
        \Rule
            {
                \transition{\mnetN[1]}{\mnetIOLabYes{\roleP}{\stSendTag{\roleQ}{\stLab}{\stS}}}{\mnetNi[1]}
                \qquad
                \transition{\mnetN[2]}{\mnetIOLabYes{\roleQ}{\stEnqueueTag{\roleP}{\stLab}{\stS}}}{\mnetNi[2]}
            }
            {\transition{\mnetN[1] \mnetPar \mnetN[2]\;}{\mnetCommLab{\roleP}{\roleQ}{\stLab}{\stS}}{\;\mnetNi[1] \mnetPar \mnetNi[2]}}
            {MNet-Comm}
    \end{array}
    \)}%
    \vspace{-3mm}%
\end{definition}

In \Cref{def:monitored-network-semantics}, each end host is modeled as a role
$\roleP$ with a session type (with an input queue) representing the end host
behavior, equipped with a monitor. Rules
$\textsc{MNet-$\stExtTag$-Good}$ and $\textsc{MNet-$\stExtTag$-Bad}$ annotate
accepted and rejected actions (with labels $\monLabYes{\stExtTag}$ and
$\monLabNo{\stExtTag}$ from \Cref{def:st-monitor-semantics}) with the end host
role $\roleP$ where such actions occurred.
The rule $\textsc{MNet-Comm}$ is different from \textsc{Net-Comm}
in~\Cref{def:network-semantics}, because it only allows communications between
two end hosts if their respective monitors
accept their outgoing and incoming message; %
\change{\#123C: clarify \textsc{MNet-Comm}}{%
    this is because the transitions in the premises of rule
    $\textsc{MNet-Comm}$ (denoting the acceptance of a message send and
    enqueuing) can only be fired via rule $\textsc{MNet-$\stExtTag$-Good}$.%
}

\subsection{Synthesizing Network Layer Monitors from Session Types}
\label{sec:formal-model:mon-synthesis}%

We now provide a concrete instantiation of
monitor $\monM$ from \Cref{def:st-monitor}: %
in \Cref{def:monitor-state-st} we formalize how to synthesize a
network-layer \emph{monitor state $\enc{\stT}$} from a session type $\stT$. %
In \Cref{sec:implementation-st-headers} (\Cref{fig:st-monitor-visual}) we
anticipated that the monitor state machine $\enc{\stT}$ differs from
that of $\stT$. This is because
$\enc{\stT}$ has to meet several non-trivial requirements:

\begin{enumerate}[label=\textbf{(R\arabic*)},leftmargin=10mm]
\item\label{req:stmon:stateful-accept} $\enc{\stT}$ must accept all messages that an end
    host implementing $\stT$ may send/receive to/from the network, depending on
    $\stT$'s state;
\item\label{req:stmon:stateful-reject} $\enc{\stT}$ should reject invalid messages that an end host
    implementing $\stT$ should \emph{not} send/receive, again depending on
    $\stT$'s state;
\item\label{req:stmon:out-of-order} $\enc{\stT}$ may receive messages from different senders in an order that
    does not match $\stT$'s expectations (due to the network semantics),\footnote{%
        In \Cref{sec:implementation-practical-protocols} we also addressed the
        issue of out-of-order delivery of messages \emph{from the same sender},
        which is orthogonal and handled by protocols like TCP.%
    } and yet, $\enc{\stT}$ must accept the valid (``good'') messages while still
    rejecting the invalid (``bad'') ones;
\item\label{req:stmon:immediate-decision} $\enc{\stT}$ must decide whether to accept or reject a message
    \emph{immediately}, \emph{without buffering}, to accommodate the limited
    memory and processing power of most network devices.
\end{enumerate}

\begin{definition}[Session-type-based network monitor]
    \label{def:monitor-state-st}%
    \label{def:monitor-state-st-queue}%
    We write $\enc{\stT}$ to represent the state of a monitor $\monM$ based on a
    session type $\stT$, with semantics given by the following rules:

    \change{\#123B: overflowing text area (too small?)}{%
    \smallskip%
    \centerline{\(%
    \fontsize{7}{7}\selectfont%
    \begin{array}{c}
        {\Rule
            {k \in I}
            {\transition
                {\enc{\stIntSum{i \in I}{\roleP[i]}{\stChoice{\stLab[i]}{\stS[i]} \stSeq \stT[i]}}}
                {\monLabYes{\stSendTag{\roleP[k]}{\stLab[k]}{\stS[k]}}}
                {\enc{\stT[k]}}}
            {STMon-IntC}
        \quad%
        \Rule
            {k \in I}
            {\transition
                {\enc{\stExtSum{i \in I}{\roleP}{\stChoice{\stLab[i]}{\stS[i]} \stSeq \stT[i]}}}
                {\monLabYes{\stEnqueueTag{\roleP}{\stLab[k]}{\stS[k]}}}
                {\enc{\stT[k]}}}
            {STMon-ExtC}}\\[10mm]
        \Rule
            {
                \exists \text{ maximal } \emptyset \neq K \subseteq I,\roleQ,\stLab,\stS:%
                \quad%
                \forall i \in I: \roleQ \neq \roleP[i]%
                \quad%
                \forall k \in K: \exists \stTi[k]:%
                {\transition
                    {\enc{\stT[k]}}
                    {\monLabYes{\stEnqueueTag{\roleQ}{\stLab}{\stS}}}
                    {\enc{\stTi[k]}}}
            }
            {\transition
                {\enc{\stIntSum{i \in I}{\roleP[i]}{\stChoice{\stLab[i]}{\stS[i]} \stSeq \stT[i]}}}
                {\monLabYes{\stEnqueueTag{\roleQ}{\stLab}{\stS}}}
                {\enc{\stIntSum{k \in K}{\roleP[k]}{\stChoice{\stLab[k]}{\stS[k]} \stSeq \stTi[k]}}}}
            {STMon-IntPfx}\\[10mm]
        \Rule
            {
                \exists \text{ maximal } \emptyset \neq K \subseteq I,\roleQ,\stLab,\stS:%
                \quad%
                \roleQ \neq \roleP%
                \quad%
                \forall k \in K: \exists \stTi[k]:%
                {\transition
                    {\enc{\stT[k]}}
                    {\monLabYes{\stEnqueueTag{\roleQ}{\stLab}{\stS}}}
                    {\enc{\stTi[k]}}}
            }
            {\transition
                {\enc{\stExtSum{i \in I}{\roleP}{\stChoice{\stLab[i]}{\stS[i]} \stSeq \stT[i]}}}
                {\monLabYes{\stEnqueueTag{\roleQ}{\stLab}{\stS}}}
                {\enc{\stExtSum{k \in K}{\roleP}{\stChoice{\stLab[k]}{\stS[k]} \stSeq \stTi[k]}}}}
            {STMon-ExtPfx}\\[10mm]
        \Rule
            {\transition{\enc{\subst{\stT}{\stRecVar}{\stRec{\stRecVar}{\stT}}}}{\monLabYes{\stExtTag}}{\enc{\stTi}}}
            {\transition{\enc{\stRec{\stRecVar}{\stT}}}{\monLabYes{\stExtTag}}{\enc{\stTi}}}
            {STMon-Rec}\qquad
       \Rule
           {\noTransition{\enc{\stT}}{\monLabYes{\stExtTag}}}
           {\transition{\enc{\stT}}{\monLabNo{\stExtTag}}{\enc{\stT}}}
           {STMon-Bad}
    \end{array}
    \small%
    \)}
    \smallskip}%

    \noindent%
    We write $\enc{\stqPair{\stT}{\stQ}}$ to represent the monitor state defined
    as follows:

\smallskip\centerline{\(%
        \enc{\stqPair{\stT}{\stQ}} = \enc{\stTi}
        \qquad\text{if and only if}\qquad
        \transition{\enc{\stT}}{\monLabYes{\enqueueTag{\stQ}}}{\enc{\stTi}}
    \)}%
    \smallskip

    \noindent where %
    $\transition{\enc{\stT}}{
      \enqueueTag\monLabYes{(\stQConsWide{\stQMsgWide{\roleP[i]}{\stLab[1]}{\stT[1]}}{\stQConsWide{\,\cdots\,}{\stQMsgWide{\roleP[n]}{\stLab[n]}{\stT[n]}}})}
    }{\enc{\stTi}}$ %
    iff %
    $\transition{\enc{\stT}}{{\monLabYes{\stEnqueueTag{\roleP[1]}{\stLab[1]}{\stS[1]}}}}{
        \transition{\cdots}{{\monLabYes{\stEnqueueTag{\roleP[n]}{\stLab[n]}{\stS[n]}}}}{
            \enc{\stTi}
        }
    }$.
\end{definition}

By rule~\textsc{STMon-IntC} in \Cref{def:monitor-state-st}, if $\stT$ is an
internal choice, then $\enc{\stT}$ accepts the corresponding send actions and
updates its state. %
Dually, by rule~\textsc{STMon-ExtC}, if $\stT$ is an
external choice, then $\enc{\stT}$ accepts the corresponding receive actions and
updates its state. This reflects requirements \ref{req:stmon:stateful-accept} and \ref{req:stmon:immediate-decision}.

Rules~\textsc{STMon-IntPfx} and \textsc{STMon-ExtPfx} allow $\enc{\stT}$ to
accept an incoming message $\stLab$ with payload type $\stS$ from role $\roleQ$,
even if the shape of $\stT$ does not expect a message from $\roleQ$ right now. %
This is necessary to satisfy requirements \ref{req:stmon:out-of-order} and \ref{req:stmon:immediate-decision}.
By the premises of these rules, acceptance is allowed only if %
role $\roleQ$ is \emph{not} an immediate recipient/sender in the internal/external choice $\stT$, %
and at least one monitor $\enc{\stT[k]}$ (where $\stT[k]$ is a continuation of $\stT$)
can indeed accept that message by firing a transition %
$\transition{\enc{\stT[k]}}{\monLabYes{\stEnqueueTag{\roleQ}{\stLab}{\stS}}}{\enc{\stTi[k]}}$.
If these conditions hold, then $\enc{\stT}$ performs the same accepting transition and becomes $\enc{\stTi}$, where $\stTi$ has the same shape as $\stT$, except:
\begin{enumerate}
\item $\stTi$ keeps only the branches of $\stT$ (indexed by the maximal set
    $K \subseteq I$) that could accept $\stEnqueueTag{\roleQ}{\stLab}{\stS}$
    in their future transitions; and
\item The continuation of each kept branch is reduced to $\stTi[k]$ (for $k \in K$).
\end{enumerate}%
In other words, after accepting $\stEnqueueTag{\roleQ}{\stLab}{\stS}$,
rules \textsc{STMon-IntPfx} and \textsc{STMon-ExtPfx}
``prune'' $\stT$ by removing all the branches that, if taken, could not
possibly accept $\stEnqueueTag{\roleQ}{\stLab}{\stS}$ in their future transitions. %
Note that, %
these rules can fire only if $K \neq \emptyset$:
there must therefore be at least one branch of $\stT$ that can accept $\stEnqueueTag{\roleQ}{\stLab}{\stS}$
in its future transitions.

Rule~\textsc{STMon-Rec} unfolds recursion. Rule~\textsc{STMon-Bad} rejects
any send/receive action that $\enc{\stT}$ does not explicitly accept,
per requirements \ref{req:stmon:stateful-reject} and \ref{req:stmon:immediate-decision}.
Finally, $\enc{\stqPair{\stT}{\stQ}}$ represents the monitor state obtained by
feeding all messages in $\stQ$ as inputs to $\enc{\stT}$, which must
accept all of them: i.e., $\enc{\stqPair{\stT}{\stQ}}$ is undefined if
$\enc{\stT}$ does not accept some message in $\stQ$.

\paragraph{Examples.} To illustrate how our session monitors work, we present
three examples:%
\begin{itemize}
    \item \Cref{ex:st-monitor-reductions} shows how a monitor $\enc{\stT}$ can
        accept messages that arrive in an order different from $\stT$'s expectations, and
        how doing so restricts the inputs and outputs it will accept next.
    \item \Cref{ex:complex-mon-st-progression} revisits
        \Cref{ex:session-type-queue-reduct} to track how a monitor evolves
        alongside the end host's session type and input queue.
    \item \Cref{ex:infinite-state-monitor} shows that some session types yield
        infinite-state monitors under \Cref{def:monitor-state-st}, which
        cannot be represented using a finite number of states in P4
        (\Cref{sec:design}).%
    \item A further monitor-reduction example appears in
        \Cref{ex:st-monitor-reductions-2} in the appendix.
\end{itemize}

\begin{example}
    \label{ex:st-monitor-reductions}%
    Consider the type $\stT[\roleQ]$ from \Cref{ex:bad-network-message}: %
    (for brevity, we omit the payload types)

    \smallskip%
    \centerline{\(%
    \stT[\roleQ] \;=\;
            \roleP \stExtC \left\{\begin{array}{@{}l@{}}
                        \stLabFmt{a} \stSeq \roleR \stExtC \stLabFmt{a'} \stSeq \stEnd \\
                        \stLabFmt{b} \stSeq \roleR \stExtC \stLabFmt{b'} \stSeq \stEnd
            \end{array}\right.
    \)}%
    \smallskip%

    By \Cref{def:monitor-state-st}, the corresponding monitor $\enc{\stT[\roleQ]}$ can
    immediately accept not only the two top-level messages from $\roleP$, but also
    the successive messages from $\roleR$ -- which appear later in $\stT[\roleQ]$, but may be
    delivered earlier by the surrounding network. %
    For the top-level inputs we have:

    \smallskip%
    \centerline{\(%
        \fontsize{9}{12}\selectfont%
            {\transition
                {\enc{\stT[\roleQ]}}
                {\monLabYes{\stInTag{\roleP}{\stLabFmt{a}}{}}}
                {\enc{\roleR \stExtC \stLabFmt{a'} \stSeq \stEnd}}}
        \quad\text{and}\quad
            {\transition
                {\enc{\stT[\roleQ]}}
                {\monLabYes{\stInTag{\roleP}{\stLabFmt{b}}{}}}
                {\enc{\roleR \stExtC \stLabFmt{b'} \stSeq \stEnd}}}
        \qquad%
        \text{\footnotesize(by \textsc{STMon-ExtC} in \Cref{def:monitor-state-st})}
    \)}%
    \smallskip%

    \noindent%
    Notice that the message sent by $\roleP$ restricts what the monitor $\enc{\stT[\roleQ]}$
    accepts from $\roleR$ afterwards. %
    If $\enc{\stT[\roleQ]}$ is deployed in the network of
    \Cref{ex:bad-network-message} to monitor end host $\roleQ$, then, if
    $\roleP$ sends $\stLabFmt{a}$, the monitor will accept $\stLabFmt{a'}$ from
    $\roleR$ (which is a ``good'' message in this state); instead, if $\roleP$
    sends $\stLabFmt{b}$, the monitor will reject $\stLabFmt{a'}$ from $\roleR$
    (which is a ``bad'' message in this state).

    Notably, the monitor $\enc{\stT[\roleQ]}$ can also immediately accept the
    messages from $\roleR$.
    The transitions $\monLabYes{\stEnqueueTag{\roleR}{\stLabFmt{a'}}{}}$
    and $\monLabYes{\stEnqueueTag{\roleR}{\stLabFmt{b'}}{}}$
    are fired by the following derivations:

    \change{\#123B: overflowing text area (too small?)}{
    \centerline{\(%
    \fontsize{6}{10}\selectfont%
        \Rule
        {
            \Rule
            {}
            {
                {\transition
                    {\enc{\roleR \stExtC \stLabFmt{a'} \stSeq \stEnd}}
                    {\monLabYes{\stEnqueueTag{\roleR}{\stLabFmt{a'}}{}}}
                    {\enc{\stEnd}}}
            }
            {STMon-ExtC}
        }
        {\transition
            {\enc{\stT[\roleQ]}}
            {\monLabYes{\stEnqueueTag{\roleR}{\stLabFmt{a'}}{}}}
            {\enc{\roleP \stExtC \stLabFmt{a} \stSeq \stEnd}}}
        {STMon-ExtPfx}
        \quad%
        \Rule
        {
            \Rule
            {}
            {
                {\transition
                    {\enc{\roleR \stExtC \stLabFmt{b'} \stSeq \stEnd}}
                    {\monLabYes{\stEnqueueTag{\roleR}{\stLabFmt{b'}}{}}}
                    {\enc{\stEnd}}}
            }
            {STMon-ExtC}
        }
        {\transition
            {\enc{\stT[\roleQ]}}
            {\monLabYes{\stEnqueueTag{\roleR}{\stLabFmt{b'}}{}}}
            {\enc{\roleP \stExtC \stLabFmt{b} \stSeq \stEnd}}}
        {STMon-ExtPfx}
    \)}%
    \smallskip}%

    \noindent%
    Observe that the message $\stLabFmt{a'}$ (resp.~$\stLabFmt{b'}$) from
    $\roleR$ causes rule \textsc{STMon-ExtPfx} to ``prune'' the session type
    $\stT[\roleQ]$ in the monitor state, only keeping the branch where
    $\stLabFmt{a}$ (resp~$\stLabFmt{b}$) from $\roleP$ can be received. %
    Therefore, if $\enc{\stT[\roleQ]}$ is deployed in the network of
    \Cref{ex:bad-network-message} to monitor end host $\roleQ$, it will accept
    $\stLabFmt{a'}$ from $\roleR$ even before $\roleP$ sends $\stLabFmt{a}$ or
    $\stLabFmt{b}$: this is because in this state it is still possible for the
    end host to consume $\stLabFmt{a'}$ without getting stuck. %
    Then, after accepting $\stLabFmt{a'}$ from $\roleR$:
    \begin{enumerate}
      \item the monitor will accept $\stLabFmt{a}$ from $\roleP$---which is a
        ``good'' message in this state, because $\stT[\roleQ]$ can consume
        $\stLabFmt{a}$ from $\roleP$ and then $\stLabFmt{a'}$ from $\roleR$ from
        the end host's input queue. However,
      \item the monitor will reject $\stLabFmt{b}$ from $\roleP$---which is a
        ``bad'' message in this state, because $\stT[\roleQ]$ cannot consume
        $\stLabFmt{b}$ from $\roleP$ and then $\stLabFmt{a'}$ from $\roleR$ (as
        shown at the end of \Cref{ex:bad-network-message}).
    \end{enumerate}

    \noindent%
    This strategy for accepting messages is necessary because, depending on the
    overall multiparty interaction, messages from $\roleR$ may be delivered
    before those from $\roleP$. This phenomenon is further illustrated in
    \Cref{ex:complex-mon-st-progression} below.
\end{example}

\begin{example}
    \label{ex:complex-mon-st-progression}%
    Consider again the session type
    $\stInfo$ from \Cref{ex:formal-info-st} (for the \lstinline|Info| role in
    \Cref{fig:our-approach}): its LTS is shown in
    \Cref{fig:st-info-visual-graph}. %
    Consider also the example execution of $\stInfo$ (with an input queue)
    in \Cref{ex:session-type-queue-reduct}. We
    now instrument $\stInfo$ and an empty input queue with a monitor $\enc{\stInfo}$, visualized in
    \Cref{fig:st-monitor-visual-graph}, obtaining $\mstPair{\stqPair{\stInfo}{\stQEmpty}}{\enc{\stInfo}}$ (by \Cref{def:st-monitor});
    we explain how their respective states change %
    as they send and receive messages, according to \Cref{def:st-monitor-semantics}.

     In \Cref{fig:st-monitor-visual}, the session type $\stInfo$ and its monitor $\enc{\stInfo}$ begin in state s0
     and m0 in  respectively. At this point,
     $\enc{\stInfo}$ will \emph{only} allow the message
     $\stEnqueueTag{\roleFmt{c}}{\stLabFmt{req}}{\tyInt}$ from the \lstinline|Client|
     (represented by $\roleFmt{c}$) to go through. Once the message arrives,
     $\enc{\stInfo}$ will accept it by progressing to state m1
     (by~\tsc{STMon-ExtC}). As part of accepting the message, the monitor
     forwards it to the end host input queue, and then the session type
     $\stInfo$ consumes it (by~\tsc{SQ-Recv} and~\tsc{SQ-Deq}), reaching state s1. The session
     type can then immediately progress to state s2, then s3, by sending
     $\stSendTag{\roleFmt{r}}{\labFmt{r\_req}}{\tyInt}$ and
     $\stSendTag{\roleFmt{d}}{\labFmt{d\_req}}{\tyInt}$ to \lstinline|Review| (role
     $\roleFmt{r}$) and \lstinline|Details| (role $\roleFmt{d}$) respectively.
     Let's assume that the monitor
     forwards both of these messages before it sees a response, progressing to state m2, then m4.

    Now, the session type (now in state s3) expects a response from $\roleFmt{r}$ and then from $\roleFmt{d}$,
    but there is no guarantee that the responses will be delivered in this exact order. %
    The monitor (now in state m2) accounts for this. %
    Suppose that the monitor receives $\stEnqueueTag{\roleFmt{d}}{\labFmt{d\_rsp}}{\tyString}$ as the
    first response. The monitor accepts the message, and progresses to state m6
    while forwarding the message. The session type, however, does not progress
    immediately, but remains in state s3 as the message from $\roleFmt{d}$ in its queue does not
    match any of the branches in its external choice. (See the execution in \Cref{ex:session-type-queue-reduct}.)

    Eventually, the monitor also receives the response message
    $\stEnqueueTag{\roleFmt{r}}{\labFmt{r\_rsp}}{\tyString}$ from $\roleFmt{r}$, and
    progresses to state m7 while forwarding the message to the session type's
    input queue---which can then finally dequeue both messages from
    $\roleFmt{r}$ and $\roleFmt{d}$, proceeding to state s4, then s5.

    Finally, the session type (now in state s5) sends the message
    $\stSendTag{\roleFmt{c}}{\labFmt{rsp}}{\tyString}$ to $\roleFmt{c}$ and progresses
    to state s6; the monitor (now in state m7) accepts the outgoing message and progresses to
    state m8, at which point the protocol has finished.
\end{example}

\begin{example}[On unmonitorable session types]
    \label{ex:infinite-state-monitor}
    Consider the session type %
    \(%
        \stRec{\stRecVar}{\stIn{\roleQ}{\stLabFmt{a}}{} \stSeq \stIn{\roleP}{\stLabFmt{b}}{} \stSeq \stRecVar}
    \) %
    (for brevity, we omit the payload types)
    By~\tsc{STMon-ExtC} and~\tsc{STMon-Rec}, a monitor with this session type in its state
    can transition by receiving from $\roleQ$. Moreover, the same monitor can
    transition by receiving from $\roleP$, with the following derivation:

    \centerline{\(%
    \small%
        \Rule
        {
            \Rule
                {
                    \Rule
                        {}
                        {
                            \transition
                            {\enc{\stIn{\roleP}{\stLabFmt{b}}{} \stSeq \stRec{\stRecVar}{\stIn{\roleQ}{\stLabFmt{a}}{} \stSeq \stIn{\roleP}{\stLabFmt{b}}{} \stSeq \stRecVar}}}
                            {\monLabYes{\stEnqueueTag{\roleP}{\stLabFmt{b}}{}}}
                            {\enc{\stRec{\stRecVar}{\stIn{\roleQ}{\stLabFmt{a}}{} \stSeq \stIn{\roleP}{\stLabFmt{b}}{} \stSeq \stRecVar}}}
                        }
                        {STMon-ExtC}
                }
                {\transition
                    {\enc{\stIn{\roleQ}{\stLabFmt{a}}{} \stSeq \stIn{\roleP}{\stLabFmt{b}}{} \stSeq \stRec{\stRecVar}{\stIn{\roleQ}{\stLabFmt{a}}{} \stSeq \stIn{\roleP}{\stLabFmt{b}}{} \stSeq \stRecVar}}}
                    {\monLabYes{\stEnqueueTag{\roleP}{\stLabFmt{b}}{}}}
                    {\enc{\stIn{\roleQ}{\stLabFmt{a}}{} \stSeq \stRec{\stRecVar}{\stIn{\roleQ}{\stLabFmt{a}}{} \stSeq \stIn{\roleP}{\stLabFmt{b}}{} \stSeq \stRecVar}}}}
                {STMon-ExtPfx}
        }
        {\transition
            {\enc{\stRec{\stRecVar}{\stIn{\roleQ}{\stLabFmt{a}}{} \stSeq \stIn{\roleP}{\stLabFmt{b}}{} \stSeq \stRecVar}}}
            {\monLabYes{\stEnqueueTag{\roleP}{\stLabFmt{b}}{}}}
            {\enc{\stIn{\roleQ}{\stLabFmt{a}}{} \stSeq \stRec{\stRecVar}{\stIn{\roleQ}{\stLabFmt{a}}{} \stSeq \stIn{\roleP}{\stLabFmt{b}}{} \stSeq \stRecVar}}}}
        {STMon-Rec}
    \)}%
    \smallskip%

    \noindent%
    The monitor could then accept an incoming message $\stLabFmt{a}$ from $\roleQ$, and return
    to its original state.
    However, %
    the monitor can also accept the next input from $\roleP$:

    \smallskip%
    \centerline{\(%
        \Rule
        {
            \Rule{
                \vdots \quad \text{\emph{(Same derivation above)}}
            }
            {
                \transition
                    {\enc{\stRec{\stRecVar}{\stIn{\roleQ}{\stLabFmt{a}}{} \stSeq \stIn{\roleP}{\stLabFmt{b}}{} \stSeq \stRecVar}}}
                    {\monLabYes{\stEnqueueTag{\roleP}{\stLabFmt{b}}{}}}
                    {\enc{\stIn{\roleQ}{\stLabFmt{a}}{} \stSeq \stRec{\stRecVar}{\stIn{\roleQ}{\stLabFmt{a}}{} \stSeq \stIn{\roleP}{\stLabFmt{b}}{} \stSeq \stRecVar}}}
            }
            {STMon-Rec}
        }
        {
            \transition
                {\enc{\stIn{\roleQ}{\stLabFmt{a}}{} \stSeq \stRec{\stRecVar}{\stIn{\roleQ}{\stLabFmt{a}}{} \stSeq \stIn{\roleP}{\stLabFmt{b}}{} \stSeq \stRecVar}}}
                {\monLabYes{\stEnqueueTag{\roleP}{\stLabFmt{b}}{}}}
                {\enc{\stIn{\roleQ}{\stLabFmt{a}}{} \stSeq \stIn{\roleQ}{\stLabFmt{a}}{} \stSeq \stRec{\stRecVar}{\stIn{\roleQ}{\stLabFmt{a}}{} \stSeq \stIn{\roleP}{\stLabFmt{b}}{} \stSeq \stRecVar}}}
        }
        {SMon-ExtPfx}
    \)}%
    \smallskip%

    \noindent%
    We can repeat this transition to accept more inputs from $\roleP$, each time reaching a new monitor state that expects more inputs from $\roleQ$:

    \smallskip%
    \centerline{\(%
        \transition
            {\ldots}
            {\monLabYes{\stEnqueueTag{\roleP}{\stLabFmt{b}}{}}}
            {\transition
                {\enc{\stIn{\roleQ}{\stLabFmt{a}}{} \stSeq \stIn{\roleQ}{\stLabFmt{a}}{} \stSeq \stIn{\roleQ}{\stLabFmt{a}}{} \stSeq \stIn{\roleQ}{\stLabFmt{a}}{} \stSeq\ \dotsc}}
                {\monLabYes{\stEnqueueTag{\roleP}{\stLabFmt{b}}{}}}
                {\ldots}
            }
    \)}%
    \smallskip%

    \noindent%
    Consequently, the LTS of this session-type monitor has infinitely many states.
\end{example}

Our monitor synthesis implementation (\Cref{sec:design}) rejects session
types such as \Cref{ex:infinite-state-monitor}, %
because the monitor state machine is constrained by the (often limited) amount
of storage available in network hardware. %
To avoid infinite-state monitors, the session types being monitored cannot
receive unbounded inputs from multiple roles. Many communication protocols
involve ``request-response'' patterns that keep our monitors
finite-state, including all the examples we evaluate in \Cref{sec:evaluation}.

\subsection{Soundness of Session-Types-Based Network Monitoring}
\label{sec:formal-model:correctness}

A non-negotiable feature of session-type-based monitors from
\Cref{def:monitor-state-st} is \emph{soundness}: this means that monitors must
not reject ``good'' messages---i.e., monitors must not produce false
positives and interfere with a well-behaved network. %
We formalize this intuition by considering a monitored network $\mnetN$ where all
monitors $\enc{\stT[\roleP]}$ (for all roles $\roleP$ in $\mnetN$) are
based on session types $\stT[\roleP]$ that are mutually compatible, and
each end host for role $\roleP$ behaves according to $\stT[\roleP]$. %
In \Cref{thm:net-mon-bisim} we show that the monitors in such $\mnetN$
are \emph{transparent}: they never disrupt communications between
well-behaved hosts.

We now develop the technical machinery for this result. In
\Cref{def:consistent-mon-instr} we define a \emph{consistent instrumentation}
where each end host is given a monitor matching the end host specification.

\begin{definition}[Consistent Instrumentation of a Network]
    \label{def:consistent-mon-instr}%
    Given a network $\netN$, we define its \emph{consistent monitor
    instrumentation} $\mnetInstr{\netN}$ as:

    \smallskip\centerline{\(%
        \mnetInstr{\netBigPar{i \in I}{\mnetPair{\roleP[i]}{\stqT[i]}}}
        \;=\;
        \mnetBigPar{i \in I}{\mnetPair{\roleP[i]}{\left(\mstPair{\stqT[i]}{\enc{\stqT[i]}}\right)}}
    \)}%
    \vspace{-6mm}%
\end{definition}

For an arbitrary $\netN$, the instrumented network $\mnetInstr{\netN}$ may reject
messages if the underlying session types are not ``compatible'' with each other,
e.g., some $\roleP$ may send to $\roleQ$ a message that $\roleQ$ does not
expect. For instance, if $\netN$ is the network in \Cref{ex:bad-network-message},
then $\mnetInstr{\netN}$ would reject messages as shown in \Cref{ex:st-monitor-reductions}.

To prove monitor soundness, we must ensure that monitored session types are compatible: we require \emph{output-liveness} as in \Cref{def:live-net}
below. %
Our output-liveness is a weaker variant of the \emph{typing context liveness} property
adopted in many session typing papers \cite{Scalas2019,BarwellSY022,DBLP:journals/tocl/GhilezanPPSY23,Prokic2025Federated,DBLP:series/lncs/YoshidaH24}:
like the standard liveness definition, we require that messages sent by a participant are eventually consumed by the intended recipient
(assuming fair scheduling)
\change{\#123B: clarify output-liveness definition}{%
  -- but unlike the standard definition, we do \emph{not} require that a participant awaiting a message will eventually receive one. %
}%
In other words, our \Cref{def:live-net} does \emph{not} allow a network to have
orphan messages that are sent and queued but never consumed -- but it allows a network to have participants
that wait forever to receive messages which are never sent.

\begin{definition}[Output-Live Session Type Networks, adapted from {\cite[Def.~4.7]{DBLP:journals/tocl/GhilezanPPSY23}}]
    \label{def:path}%
    \label{def:fair-path}%
    \label{def:live-path}%
    \label{def:live-net}%
    A \emph{network path} is a possibly infinite sequence of pairs of network
    configurations $(\netN[i])_{i \in I}$, %
    where $I = \{0, 1, \ldots\}$ is a set of consecutive natural numbers and,
    $\forall i \in I$, $\transition{\netN[i]}{\netTau[\gamma_i]}{\netNi[i]}$. %
    We say that \emph{a network path is fair} iff, $\forall i \in I$:
    \begin{enumerate}
      \item\label{item:fair:send} if $\transition{\netN[i]}{\netCommLab{\roleP}{\roleQ}{\stLab}{\stS}}{}$\!,
        then $\exists k \in I$ such that $k \ge i$ and $\transition{\netN[k]}{\netCommLab{\roleP}{\roleQi}{\stLabi}{\stSi}}{\netN[k+1]}$;
      \item\label{item:fair:recv} if $\transition{\netN[i]}{\netDequeueLab{\roleP}{\roleQ}{\stLab}{\stS}}{}$\!,
        then $\exists k \in I$ such that $k \ge i$ and $\transition{\netN[k]}{\netDequeueLab{\roleP}{\roleQ}{\stLab}{\stS}}{\netN[k+1]}$.
    \end{enumerate}

    We say that \emph{a network path is output-live} iff,
    taking any $i \in I$ and letting $\netN[i] = \netBigPar{j \in J}{\netPair{\roleP[j]}{(\stqPair{\stT[j]}{\stQ[j]})}}$,
    $\forall j \in J$ we have that
        if $\stQ[j] \equiv \stQCons{\stQMsg{\roleQ}{\stLab}{\stS}}{\stQi}$,
        then $\exists k \in I$ such that $k \ge i$ and $\transition{\netN[k]}{\netDequeueLab{\roleP[j]}{\roleQ}{\stLab}{\stS}}{\netN[k+1]}$.

    We say that \emph{$\netN$ is output-live} if every fair path beginning with
    $\netN$ is output-live.
\end{definition}

In \Cref{def:live-net}, a ``path'' represents a possible network execution.
A path is \textit{fair} if it eventually
\change{\#123C: fix typo and clarify}{%
    allows enabled communications between participants (item~\ref{item:fair:send}) and every enabled dequeuing action (item~\ref{item:fair:recv}); %
    note that, in item~\ref{item:fair:send}, the existence of $\roleQ, \stLab, \stS$
    establishes that $\roleP$ is ready to send some message
    (with an internal choice) that a recipient $\roleQ$ is ready to enqueue
    -- while $\roleQi, \stLabi, \stSi$ are the actual recipient and message
    selected by $\roleP$ in this execution path. %
}%
A path is \textit{output-live} if every
queued message is eventually consumed by its intended recipient.

\begin{example}[Output-live networks]
    \label{ex:live-net}%
    Consider the network $\netN$ in \Cref{ex:bad-network-message}: %
    $\netN$ is not output-live, because it has e.g.~a fair path where $\roleR$
    sends $\stLabFmt{a'}$ to $\roleQ$, $\roleP$ sends $\stLabFmt{b}$ to
    $\roleQ$, hence $\roleQ$ cannot consume the queued message $\stLabFmt{a'}$
    from $\roleR$. %
    In contrast, the network $\netNi$ obtained by replacing the session type of
    $\roleP$ with $\stTi[\roleP] = \stOut{\roleQ}{\stLabFmt{a}}{} \stSeq \stEnd$
    is output-live, because in every fair path of $\netNi$ every queued message
    is eventually consumed. %
    Also, all the examples we evaluate in \Cref{sec:evaluation} are output-live.
\end{example}

To ensure monitoring correctness,
we require a further \emph{half-duplex} assumption to control monitor state-space size. %
Intuitively, a network $\netN$ is \emph{half-duplex} if, for any two roles
$\roleP$ and $\roleQ$ in $\netN$, %
\change{\#123B: clarify half-duplex definition}{%
  data can flow only in one direction at a time, i.e., if $\roleP$ sent a
  message to $\roleQ$, then $\roleQ$ must consume that message before sending
  another message to $\roleP$ (and \emph{vice versa}). %
  In other words, $\roleP$ and $\roleQ$ can only communicate by ``taking turns''
  -- thus, the input queue of $\roleP$ can contain a message from $\roleQ$ only
  if the input queue of $\roleQ$ does \emph{not} contain any message from
  $\roleP$. %
}%
All the examples we evaluate in \Cref{sec:evaluation} are half-duplex. %
\emph{(For the formal definition of half-duplex and an example showing why we need it,
see \Cref{def:half-duplex-net} and \Cref{ex:half-duplex-motivation} in the appendix.)}

To state our monitoring soundness result, we use \Cref{def:internal-bisim} to
define when two networks have equivalent internal behavior, i.e., when they communicate and consume messages in the same way.

\begin{definition}[Internal Bisimulation]
    \label{def:internal-bisim}
    Let $\phi$ be an annotation to distinguish $\tau$-labels of the form $\internalLab{\phi}$. %
    We say that \emph{$\relR$ is an internal bisimulation relation} iff,
    whenever $(s, t) \in \mathop{\relR}$,
    \begin{enumerate}
    \item\label{item:internal-bisim-lr} if $\transition{s}{\internalLab{\phi}}{s'}$,
        then $\exists t'$ such that $\transition{t}{\internalLab{\phi}}{t'}$
        and $(s', t') \in \mathop{\relR}$;
    \item\label{item:internal-bisim-rl} if $\transition{t}{\internalLab{\phi}}{t'}$,
        then $\exists s'$ such that $\transition{s}{\internalLab{\phi}}{s'}$
        and $(s', t') \in \mathop{\relR}$.
    \end{enumerate}
    We say $s$ and $t$ are \emph{internally bisimilar}, written $s \tbisim t$, iff
    there is an internal bisimulation $\relR$ such that $(s, t) \in
    \mathop{\relR}$.
\end{definition}

We now have all the ingredients to state and prove that our monitors are
\emph{sound} by runtime verification standards \cite{DBLP:series/lncs/BartocciFFR18}, i.e., they have no false positives.\footnote{%
  Another desirable property for monitors is
  \emph{completeness}, i.e., having no false negatives. %
  Here we focus on soundness because it is non-negotiable, and completeness may
  not be achievable together with soundness: we discuss these issues in
  \Cref{sec:conclusion-futurework}.%
} %
In our setting, this means they never misclassify a good message as bad
and never interfere with well-behaved end hosts, provided the implemented protocol is output-live and half-duplex. %
\emph{(Proof in \Cref{appendix:SessionTypesAddionalDetails}.)}

\begin{restatable}[Monitor Soundness]{theorem}{lemMonitorCorrectness}
    \label{thm:net-mon-bisim}
    If $\netN$ is output-live and half-duplex, then $\netN \tbisim \mnetInstr{\netN}$.
\end{restatable}

\begin{remark}[On determining output-liveness and half-duplex properties]
    \label{remark:determining-liveness-half-duplex}%
    Output-liveness (\Cref{def:live-net}) is generally undecidable, since two
    session types with unbounded queues can encode a Turing
    machine~\cite[Theorem 2.5]{DBLP:journals/corr/abs-1211-2609}. It can,
    however, be guaranteed by decidable approximations such as bounded model
    checking or projection from a global type~\cite{DBLP:conf/cav/LangeY19,DBLP:conf/cav/LiSWZ23}. %
    Similar techniques can be used to
    ensure half-duplex execution (\Cref{def:half-duplex-net}, \cite{DBLP:journals/corr/abs-2209-10328}). %
    These checks are orthogonal to this work.
    The protocols we evaluate in \Cref{sec:evaluation} are output-live and half-duplex,
    with bounded queue sizes,
    so they have finite LTSs and are amenable to model checking.
\end{remark}

\begin{remark}[On the monitor rejection strategy]
    \label{remark:monitor-rejections}%
    The particular rejection strategy for session-type monitors does not affect
    the soundness \Cref{thm:net-mon-bisim}, which only concerns accepted
    behavior. %
    Concretely, \Cref{def:monitor-state-st} says that the monitor's verdict is
    not persistent: if a monitor rejects a message, then it can still accept a
    subsequent valid message. %
    \Cref{thm:net-mon-bisim} would still hold, for example, if rule
    \textsc{STMon-Bad} in \Cref{def:monitor-state-st} always moved $\enc{\stT}$
    to $\enc{\stEnd}$ after an invalid send/receive, making the rejection
    verdict persistent and blocking any subsequent message to/from the end host. %
    Indeed, \toolName's TCP-oriented monitors (described in
    \Cref{sec:implementation-practical-protocols}) use persistent verdicts: they
    block all end host communications (by dropping the whole TCP connection)
    when the end host sends a ``bad'' message (as we show in
    \Cref{sec:evaluation-qualitative}).
\end{remark}
\section{Empirical Evaluation}
\label{sec:evaluation}

In this section, we evaluate \toolName's correctness and effectiveness against
the following \change{\#123A: use ``research questions'' instead of
``objectives'' (also updating the text that follows)}{%
    research questions%
}:

\begin{enumerate}[label=\textbf{(Q\arabic*)},leftmargin=10mm]
	\item\label{item:eval-non-trivial} Can \toolName generate monitors for
	non-trivial multiparty protocols?
	\item\label{item:eval-correctness} Do \toolName-generated
          monitors accommodate correct communication without
          interference, while rejecting incorrect messages even in the
          presence of packet loss, duplication, and reordering (when
          using the TCP-oriented monitors described in
          \Cref{sec:implementation-practical-protocols})?
	\item\label{item:eval-concurrency} Does \toolName support monitoring
	multiple concurrent sessions?
\end{enumerate}

We first describe test cases and setup (\Cref{sec:evaluation:test-cases-setup}), then analyze a representative case (\Cref{sec:evaluation-qualitative}), and finally report aggregate monitoring statistics (\Cref{sec:overview:monitoring-statistics}).
\change{\#123A: why no performance evaluation}{%
    Our evaluation is based on simulated networks including nodes running \texttt{BMv2}~\cite{bmv2}, a P4-enabled software switch. %
    We adopt this setting because careful hardware experiments would require significant engineering efforts that are orthogonal to the main contributions of this paper. %
    The drawback of this choice is that we cannot perform meaningful performance evaluations, as BMv2 does not reflect the performance characteristic of real P4-enabled hardware. %
    However, previous work suggests that, if the \toolName-generated P4 monitors can be compiled to a hardware platform without exceeding the available resources, then they will run with little to no overhead, no matter how much traffic the device is processing, up to its limit. For example, Figures 9 and 10(c) of~\cite{NetCache2017} show constant throughput and latency, even when the P4 switch is fully loaded.%
}

\subsection{Test Cases and Evaluation Setup}
\label{sec:evaluation:test-cases-setup}

\ifecoop
\subparagraph{Methodology.}%
\else
\paragraph{Methodology}%
\fi
To address question~\ref{item:eval-non-trivial}, we selected non-trivial test cases based on real-world multiparty protocols, with varying numbers of participants and branching/looping structures (\Cref{tab:example-properties}).
We describe each test case in \Cref{appendix:ExampleDesc}, together with
the local session type of each participant.
In each test case, participant behavior is specified as a
session type, and the full system is output-live and half-duplex (\Cref{def:live-net,def:half-duplex-net}).
\change{\#123C: clarify sentence}{%
    This guarantees the formal preconditions for sound monitoring
    (\Cref{thm:net-mon-bisim}) and provides the basis for empirically evaluating
    question~\ref{item:eval-correctness} -- for which here we also consider
    networks with packet loss, duplication, and reordering (that are are not
    formally covered by \Cref{thm:net-mon-bisim}).%
}%

\begin{table}
	\footnotesize%
	\renewcommand{\arraystretch}{1.2}%
	\begin{tabular}{|c|c|c|c|c|}
		\hline\hline
		\rowcolor{lightgray} \textbf{Test case} & \textbf{Participants} & \textbf{Branching} & \textbf{Loops} & \textbf{Description} \\
		\hline\hline
		BookInfo \cite{IstioBookinfo} & 5 &  & & match a review to a book name \\
		\hline
		Store management & 7 & \textcolor{ForestGreen}{\checkmark} & \textcolor{ForestGreen}{\checkmark} & online ordering service \\
		\hline
		VPN & 4 & \textcolor{ForestGreen}{\checkmark} &  \textcolor{ForestGreen}{\checkmark}\ & authenticate communication \\
		\hline
		Stateful firewall & 2 & \textcolor{ForestGreen}{\checkmark} & \textcolor{ForestGreen}{\checkmark} & traffic filtering\\
		\hline
		DNS \cite{ross2021computer} & 5 &  &  & DNS resolver server \\
		\hline
		Auction protocol
		& 3 & \textcolor{ForestGreen}{\checkmark} & \textcolor{ForestGreen}{\checkmark} & two-buyer auction protocol \\
		\hline
		CDN \cite{ross2021computer} & 4 &  &  & content distribution network\\
		\hline
		SIP \cite{rfc3261} & 3 & \textcolor{ForestGreen}{\checkmark} &  &  session initiation protocol over proxy\\
		\hline
		POP3 \cite{rfc1734}\cite{rfc1939}& 2 & \textcolor{ForestGreen}{\checkmark} & \textcolor{ForestGreen}{\checkmark} & client sends multiple queries to a server\\
		\hline
		Multiplayer game & 4 & \textcolor{ForestGreen}{\checkmark} & \textcolor{ForestGreen}{\checkmark} & turn-based game \\
		\hline
	\end{tabular}
		\caption{The 10 test case protocols on which we evaluated \toolName. %
		For each one of them we report the number of participants, and whether
		the session type representation includes branching and/or loops.}%
\label{tab:example-properties}
\end{table}

\ifecoop
\subparagraph{Evaluation Setup.}%
\else
\paragraph{Evaluation Setup.}%
\fi
For each test case in \Cref{tab:example-properties}, we set up a simulated
network in which end hosts (i.e., multiparty-protocol participants) communicate via P4-enabled border switches, as in
\Cref{fig:network-microservice-example}. %
We use Mininet~\cite{mininet}, which allows us to simulate networks with
different topologies, end hosts, and switch configurations. %
The simulated network includes nodes running \texttt{BMv2}~\cite{bmv2}, a P4-enabled software switch.
These switches perform regular forwarding when monitoring is disabled, and deploy/run our
\toolName-generated monitors to evaluate question~\ref{item:eval-non-trivial}.

To evaluate question~\ref{item:eval-correctness}, for each test case
in \Cref{tab:example-properties} we provide correct and faulty
participant implementations and assess whether
\toolName-generated monitors accept or reject packets as expected.
Each variant is implemented in Python and executed on Mininet using the \toolName-generated API (\Cref{sec:implementation-st-api}).

To stress-test whether question \ref{item:eval-correctness} can be answered positively under different transports, we implement
each scenario in \Cref{tab:example-properties} with both UDP and
TCP communication. For TCP, we evaluate the
behavior of our TCP-oriented monitors
(\Cref{sec:implementation-practical-protocols}) both on a perfectly
reliable network, and on an unreliable network with packet loss,
duplication, and delay: %
specifically, we configure the Mininet end hosts to drop 1\% of all incoming
packets, duplicate 1\% of all outgoing packets, and delay the sending of
outgoing packets by a variable amount (up to 100\,ms).
The delay perturbs packet ordering both across senders (where arbitrary interleavings are allowed) and for packets from the same sender (where out-of-order packets are rejected and retransmitted).

To evaluate question~\ref{item:eval-concurrency}, we run multiple concurrent
sessions per test case (typically 10 to 50). All experiments were conducted on a machine with an 8-core, 3\,GHz CPU and 32\,GB of RAM,
running Ubuntu~22.04.

\subsection{Assessing the Correctness of \toolName-Generated Session Monitors}
\label{sec:evaluation-qualitative}
\label{sec:evaluation-qualitative-reject}
\label{sec:evaluation-qualitative-overhead}

To illustrate our evaluation of question \ref{item:eval-correctness}
(i.e., whether \toolName-generated monitors accept/reject messages
correctly), we focus on one of the 10 test
cases in \Cref{tab:example-properties}: the BookInfo
protocol~\cite{IstioBookinfo}, our running example from
\Cref{sec:introduction}. We applied the same assessment to every
test case in \Cref{tab:example-properties} and observed similar
results, so the analysis below is representative. We also briefly report results for the VPN
test case, which is structurally richer than BookInfo and covers all
session-type features (branching and nested loops). Full
details for the other test cases appear in \Cref{appendix:ExampleDesc}.

\begin{figure}
\centering
\begin{minipage}{0.65\textwidth}
\begin{lstlisting}[language=py, basicstyle=\scriptsize\ttfamily]
def protocol(s):
    isbn = s.recvMsg(Client,  Request).text
    s.sendMsg(Review,  RqsReview(isbn))
    s.sendMsg(Review,  RqsDetail(isbn)) # BAD
    s.sendMsg(Details, RqsDetail(isbn))
    s.sendMsg(Details, RqsReview(isbn)) # BAD
    review = s.recvMsg(Review,  RspReview).text
    details = s.recvMsg(Details, RspDetail).text
    s.sendMsg(Client,  Response(review + details))
\end{lstlisting}%
\caption{Faulty implementation of the session type in \Cref{code:st-info-impl}.}
\label{code:faulty-info-impl}
\end{minipage}
\end{figure}

For this representative case, \Cref{fig:result-barplots} reports cumulative packets
received across all hosts under different configurations (UDP or TCP,
reliable or unreliable TCP networks, faulty or correct hosts, and with or
without runtime monitors).
\change{\#123C: anomalous packet counts}{The bars show the median packet counts
over five runs, with the maximum and minimum counts shown with the black dot/line
(most noticeable in \Cref{fig:bookreview-barplot}).}

\begin{figure}
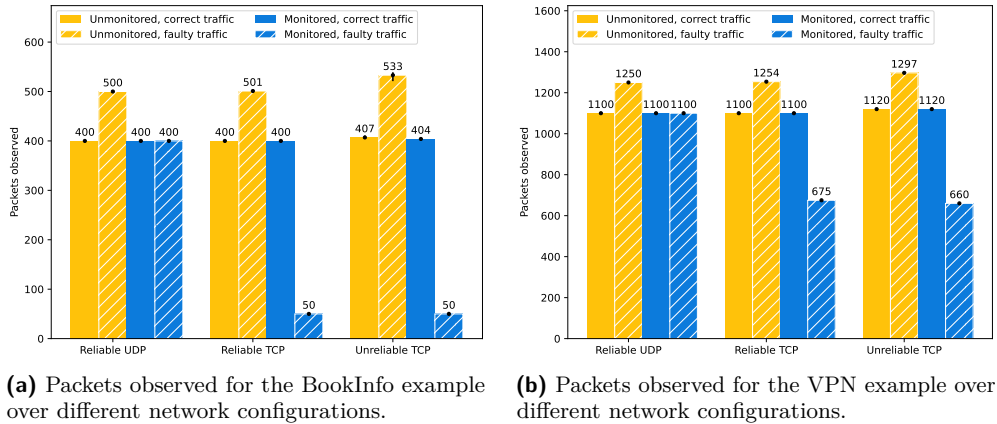

    \centering
    \begin{subfigure}{.45\textwidth}
        \includegraphics[width=\textwidth]{plots/bookreview\_barplot.pdf}
        \caption{Packets observed for the BookInfo example over different network configurations.}
        \label{fig:bookreview-barplot}
    \end{subfigure}
    \hspace{5pt}
    \begin{subfigure}{.45\textwidth}
        \includegraphics[width=\textwidth]{plots/vpn\_barplot.pdf}
        \caption{Packets observed for the VPN example over different network configurations.}
        \label{fig:vpn-barplot}
    \end{subfigure}
    \caption{Incoming packets observed across all hosts, in different examples and configurations.%
    }
    \label{fig:result-barplots}
\end{figure}

\begin{figure}[t]
    \centering
    \begin{subfigure}{.45\textwidth}
        \includegraphics[width=\textwidth]{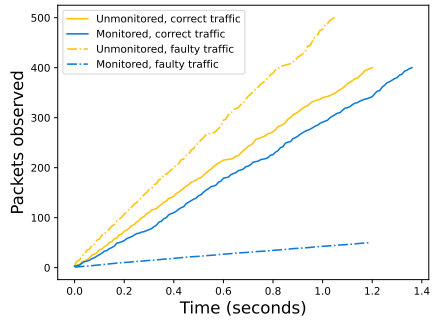}
        \caption{BookInfo, reliable network with TCP}
        \label{fig:bookinfo-tcp-network}
    \end{subfigure}\hspace{5pt}
    \begin{subfigure}{.45\textwidth}
        \includegraphics[width=\textwidth]{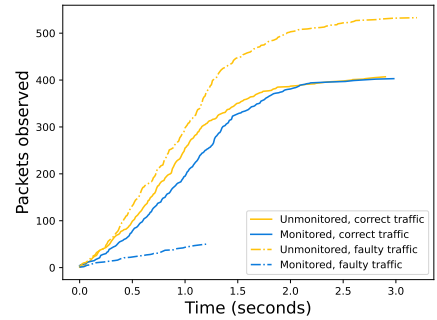}
        \caption{BookInfo, unreliable network with TCP}
        \label{fig:bookinfo-tcploss-network}
    \end{subfigure}

    \vspace{4pt}
    \begin{subfigure}{.45\textwidth}
        \includegraphics[width=\textwidth]{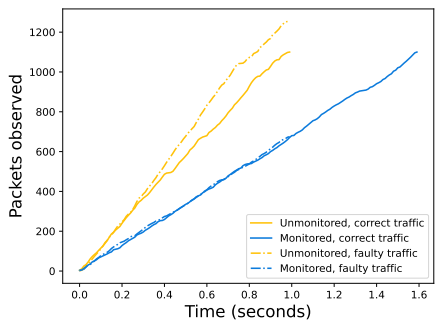}
        \caption{VPN, reliable network with TCP}
        \label{fig:vpn-tcp-network}
    \end{subfigure}\hspace{5pt}
    \begin{subfigure}{.45\textwidth}
        \includegraphics[width=\textwidth]{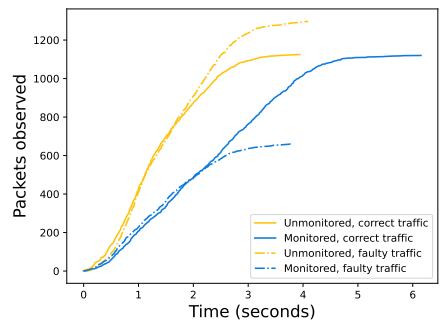}
        \caption{VPN, unreliable network with TCP}
        \label{fig:vpn-tcploss-network}
    \end{subfigure}
    \caption{Incoming packets observed across all hosts over time for BookInfo (top row) and VPN (bottom row).}
    \label{fig:tcp-timeseries-results}
\end{figure}

The ``correct traffic'' bars in \Cref{fig:bookreview-barplot}
(solid red and blue columns) represent BookInfo runs where each end host correctly implements the session protocol.
\toolName-generated monitors do not reject any packets: the same number of
packets is observed in both monitored (blue) and unmonitored (red) networks. %
This is consistent with our soundness \Cref{thm:net-mon-bisim}. %

In BookInfo configurations with ``faulty traffic'', the end host for
participant \lstinline|Info| runs a program that does \emph{not} conform
to the expected protocol and sends incorrect messages to other hosts. %
Specifically, it sends additional messages to
\lstinline|Review| and \lstinline|Details| (\Cref{code:faulty-info-impl}).
Key observations:
\begin{itemize}
\item The total number of packets with faulty traffic on
  unmonitored networks (red hatched columns) is higher than that for
  correct traffic (red solid columns), because all faulty
  packets reach their destination end host. %
\item By contrast, there is no increase in packets observed under
  monitoring (blue hatched columns), because each faulty packet
  (which is UDP in this case) is dropped by the session monitor for
  \lstinline|Info| and does not enter the network; hence, the faulty packet
  is not observed by other end hosts.%
\item With TCP transport and TCP-oriented monitors, incorrect packets cause a significant
  drop in the observed packet counts: this is due to the session-closing mechanism in
  \Cref{sec:implementation-practical-protocols}, which blocks a session as soon as a
  faulty packet is observed. In \Cref{fig:tcp-timeseries-results}
  it is possible to see that the \toolName TCP monitors keep the observed packet count
  persistently low in faulty runs, as no further packets are observed after a faulty
  one causes its session to be closed.
\item \change{\#123C: anomalous packet counts}{On unreliable networks with TCP,
  there are slight differences for observed packet counts between monitored and
  unmonitored examples. This is because (1) monitors may drop out-of-order packets,
  which may slightly reduce packet count, or slightly increase it due to
  retransmissions; and (2) we perform random packet drops, reordering, and
  delays to simulate unreliable networks. Consequently, some out-of-order and
  retransmitted packets may be dropped or received by the end host multiple
  times, causing variations in packet counts. Crucially, the plots
  show that our TCP monitors do not block correct traffic -- otherwise they
  would disrupt TCP connections and significantly drop the ``monitored, correct
  traffic'' packet count like the ``monitored, faulty traffic'' TCP count.}
\end{itemize}

\noindent
The same trend appears for VPN in the bottom row of \Cref{fig:tcp-timeseries-results}.
There, faulty packets are sent later than in BookInfo, so monitored faulty runs observe more packets before session closure. Corresponding plots for the remaining examples appear in \Cref{appendix:BarPlots}.

These observations show that our monitors do not disrupt
well-behaved programs and correctly reject non-conformant messages, for non-trivial cases: therefore,
questions \ref{item:eval-non-trivial} and \ref{item:eval-correctness} can be answered positively. %
Question \ref{item:eval-concurrency} can be answered positively as well, as these observations hold when running multiple concurrent sessions for each test case.%

\subsection{Monitoring Statistics}
\label{sec:overview:monitoring-statistics}

\begin{table}
	\tiny%
  \centering
	\renewcommand{\arraystretch}{1.2}%
  \begin{tabular}{|c||c|c|c||c|c|c|}
		\hline
		    \rowcolor{lightgray} \cellcolor{white} & \multicolumn{3}{|c|}{\textbf{UDP}}     & \multicolumn{3}{|c|}{\textbf{UDP}}    \\
        \rowcolor{lightgray} \cellcolor{white} & \multicolumn{3}{|c|}{\textbf{Correct}} & \multicolumn{3}{|c|}{\textbf{Faulty}} \\ \hline
        \rowcolor{lightgray} \cellcolor{white} & A    & R    & T                        & A    & R    & T                       \\ \hline \hline
        \textbf{VPN}                           & 1950 & 0    & 0                        & 1950 & 150  & 0                       \\ \hline
        \textbf{Book}                          & 800  & 0    & 0                        & 800  & 100  & 0                       \\ \hline
        \textbf{Store}                         & 1800 & 0    & 0                        & 1800 & 100  & 0                       \\ \hline
        \textbf{Firewall}                      & 2100 & 0    & 0                        & 2100 & 500  & 0                       \\ \hline
        \textbf{DNS}                           & 800  & 0    & 0                        & 800  & 100  & 0                       \\ \hline
        \textbf{Auction}                       & 2800 & 0    & 0                        & 2800 & 200  & 0                       \\ \hline
        \textbf{CDN}                           & 500  & 0    & 0                        & 500  & 100  & 0                       \\ \hline
        \textbf{SIP}                           & 350  & 0    & 0                        & 350  & 150  & 0                       \\ \hline
        \textbf{POP3}                          & 1000 & 0    & 0                        & 1000 & 500  & 0                       \\ \hline
        \textbf{Game}                          & 3000 & 0    & 0                        & 3000 & 250  & 0                       \\ \hline
	\end{tabular}
  \hspace{3mm}
  \begin{tabular}{|c||c|c|c||c|c|c|}
		\hline
		    \rowcolor{lightgray} \cellcolor{white} & \multicolumn{3}{|c|}{\textbf{TCP + reliable net}}     & \multicolumn{3}{|c|}{\textbf{TCP + reliable net}} \\
        \rowcolor{lightgray} \cellcolor{white} & \multicolumn{3}{|c|}{\textbf{Correct}}                & \multicolumn{3}{|c|}{\textbf{Faulty}}             \\ \hline
        \rowcolor{lightgray} \cellcolor{white} & A    & R    & T                                       & A    & R    & T                                   \\ \hline \hline
        \textbf{VPN}                           & 1950 & 0    & 0                                       & 1125 & 175  & 0                                   \\ \hline
        \textbf{Book}                          & 800  & 0    & 0                                       & 100  & 935  & 0                                   \\ \hline
        \textbf{Store}                         & 1800 & 0    & 0                                       & 400  & 450  & 0                                 \\ \hline
        \textbf{Firewall}                      & 2100 & 0    & 0                                       & 50   & 531  & 0                                \\ \hline
        \textbf{DNS}                           & 800  & 0    & 0                                       & 400  & 614  & 0                                   \\ \hline
        \textbf{Auction}                       & 2800 & 0    & 10                                      & 360  & 995  & 0                                \\ \hline
        \textbf{CDN}                           & 500  & 0    & 0                                       & 200  & 850  & 0                                 \\ \hline
        \textbf{SIP}                           & 350  & 0    & 0                                       & 200  & 1102 & 0                                  \\ \hline
        \textbf{POP3}                          & 1000 & 0    & 0                                       & 500  & 500  & 0                                 \\ \hline
        \textbf{Game}                          & 3000 & 0    & 0                                       & 800  & 447  & 0                                 \\ \hline
	\end{tabular}\\

	  \begin{tabular}{|c||c|c|c||c|c|c||c|c|c||c|c|c||c|c|c||c|c|c|}
		\hline
		\rowcolor{lightgray} \cellcolor{white}     & \multicolumn{3}{|c|}{\textbf{TCP + unreliable net}} & \multicolumn{3}{|c|}{\textbf{TCP + unreliable net}} \\
        \rowcolor{lightgray} \cellcolor{white} & \multicolumn{3}{|c|}{\textbf{Correct}}              & \multicolumn{3}{|c|}{\textbf{Faulty}}               \\ \hline
        \rowcolor{lightgray} \cellcolor{white} & A    & R    & T                                     & A    & R    & T                                     \\ \hline \hline
        \textbf{VPN}                           & 1950 & 166  & 28                                    & 1059 & 510  & 7                                     \\ \hline
        \textbf{Book}                          & 800  & 3    & 20                                    & 100  & 712  & 0                                     \\ \hline
        \textbf{Store}                         & 1800 & 84   & 40                                    & 400  & 393  & 4                                   \\ \hline
        \textbf{Firewall}                      & 2100 & 0    & 21                                    & 50   & 494  & 1                                  \\ \hline
        \textbf{DNS}                           & 800  & 0    & 14                                    & 400  & 568  & 2                                   \\ \hline
        \textbf{Auction}                       & 2800 & 315  & 52                                    & 416  & 891  & 2                                  \\ \hline
        \textbf{CDN}                           & 500  & 0    & 6                                     & 200  & 745  & 7                                   \\ \hline
        \textbf{SIP}                           & 350  & 48   & 7                                     & 221  & 827  & 3                                   \\ \hline
        \textbf{POP3}                          & 1000 & 0    & 11                                    & 500  & 450  & 6                                   \\ \hline
        \textbf{Game}                          & 3000 & 106  & 36                                    & 800  & 509  & 16                                   \\ \hline
	\end{tabular}
		\caption{Statistics collected by our \toolName-generated session monitors for the test
		cases in \Cref{tab:example-properties}. For each test case, we consider both
		correct and faulty implementations, using either UDP or TCP as the transport
		protocol, with TCP evaluated under both reliable and unreliable network
		conditions. Column~\textbf{A} reports the total number of packets
		\textbf{A}ccepted by the monitor, \textbf{R} the total number of
		\textbf{R}ejected packets, and \textbf{T} the number of packets identified as
		re\textbf{T}ransmissions.%
    }
	\label{tab:example-results}
\end{table}

\change{\#123B: expand discussion}{%
  \noindent%
  \Cref{tab:example-results} summarises various packet statistics across the test cases in \Cref{tab:example-properties}:
}
\begin{itemize}
    \item For correct implementations (first, third, and fifth columns), monitors reject no packets.
    \item For faulty implementations (second, fourth, and sixth columns), monitors reject packets.
    \item With correct programs over TCP on unreliable networks (fifth column), monitors record retransmissions without introducing false rejects.
\end{itemize}
\change{\#123B: expand discussion}{%
  The statistics support the effectiveness of \toolName against question
  \ref{item:eval-correctness} (correct monitoring) and
  \ref{item:eval-concurrency} (monitoring of parallel sessions). %
}%
\change{\#123C: anomalous packet counts}{%
  \Cref{tab:example-results} also shows that when TCP is used, packet
  retransmissions can happen even if the intended protocol is correctly
  implemented -- both on reliable and unreliable networks. %
  Retransmission may occur when packets are lost (in unreliable netorks), or
  depending on their delivery speed. E.g., the sender's TCP stack may retransmit
  a packet if an ACK does not arrive quickly enough, or the recipient's TCP
  buffer may become full and drop some packets, causing their retransmission. %
  \toolName correctly handles these situations.%
}%
\section{Related Work}
\label{sec:related-work}

Session types have been extensively developed as standalone programming
languages~\cite{Ho93,HoYoCa08,GaVa10,Wa12,JaBaKr22}, and as libraries or tools
for existing languages such as Rust~\cite{JeMuLa15,Ko19,CuYo21,CuYoVa22,ChBaTo22},
Java~\cite{HuKoPeYoHo10,HuYoHo08}, Scala~\cite{ScYo16}, OCaml~\cite{Pa17,ImYoYu19},
Haskell~\cite{PuTo08,ImYuAg10,LiGa16,OrYo16}, Go~\cite{CaHuJoNgYo19,NgYo16}, and
others~\cite{Yoshida2024}.

Techniques for enforcing session types with run-time monitors have also been
studied in prior work, e.g.~\cite{BuFrSc21,BuFrScTrTu21,BoChDeHoYo13,BoChDeHoYo17,DBLP:journals/fac/NeykovaBY17,DBLP:journals/fmsd/DemangeonHHNY15,DBLP:conf/rv/NeykovaYH13}.
These approaches focus on \emph{application-layer} monitors abstracted from
the underlying network.
Burl{\`{o}} et al.~\cite{BuFrSc21,BuFrScTrTu21} study binary session types where at least one party is a closed-box process (i.e., not statically verified). They synthesize Scala monitors, prove correctness guarantees, and establish the impossibility of sound and complete black-box monitoring.
Bocchi et al.~\cite{BoChDeHoYo13,BoChDeHoYo17} developed a monitored-network framework based on $\pi$-calculus processes and multiparty session types. Their ``networks'' are at a different layer from ours: they model a global routing application (akin to a message broker), implemented with AMQP~\cite{AMQP,DBLP:journals/fmsd/DemangeonHHNY15}. That model allows unbounded buffering, so their semantics are close to our session types with queues (\Cref{def:local-semantics}) and do not address the synthesis requirements in \Cref{sec:formal-model:mon-synthesis} that motivate \Cref{def:monitor-state-st}.

There is also growing work on run-time enforcement of network properties without session types. For example,
\texttt{Hydra}~\cite{ReRuKiVeCaMoChMcFo23} deploys ``checkers'' on P4 switches
that enforce network-wide properties. These properties are expressed in terms
of packet trajectories through the network and observations of intermediate
state at each hop. \texttt{FLM}~\cite{JoBeChMaWa24} is a language and compiler
for enforcing line-rate network monitoring using programmable switches. Our
work is complementary: both \texttt{Hydra} and \texttt{FLM} could serve as
implementation platforms for the run-time monitors we propose.
At the microservice level, Grewal, Godfrey, and Hsu use run-time monitors to enforce policies~\cite{GrGoHs23}. Their goals are similar, but the technical setting differs: they rely on Istio Envoy proxies~\cite{IstioArchitecture} on end hosts, whereas we target lower-level P4 devices. They also use declarative tree policies, while we use multiparty session types; studying whether a tree-policy-like formalism could model and monitor session protocols is an interesting direction for future work.

\change{\#123A: cite Giallorenzo et al.}{%
  Giallorenzo et al.~\cite{DBLP:conf/icsoc/GiallorenzoMMMPP24} propose
  choreographic programming for specifying the global coordination between
  Cloud-native Network Functions (CNF) architecture components. The
  choreographic program is projected (i.e., compiled) into executable Java
  programs that perform Virtual Network Functions (VNF) such as intrusion
  detection and traffic filtering. They present a case study where a P4-enabled
  virtual switch (based on \texttt{BMv2}~\cite{bmv2}, also adopted in our
  evaluation) directs network traffic to the projected VNFs for network traffic
  monitoring. The work \cite{DBLP:conf/icsoc/GiallorenzoMMMPP24} is orthogonal
  to ours: they introduce a high-level software-defined network programming
  architecture and do not address the problem of tracking session protocols;
  moreover, their VNFs can implement and run arbitrary code without the
  constraints of P4-enabled devices (which are a major factor in our work). In
  principle, the P4 monitors generated by NEST could be deployed in the CNF
  architecture of \cite{DBLP:conf/icsoc/GiallorenzoMMMPP24} to perform session
  monitoring -- and their VNFs could deploy and control NEST monitors via
  P4Runtime~\cite{p4rtspec}. A question that links our work to theirs is: is it
  possible to synthesise NEST-style MAT tables from a choreographic program that
  describes a network-level monitoring policy? This would allow moving the
  monitoring and filtering logic from their (Java-based) VNFs to P4 devices.
  This is intriguing and non-trivial work that would require bridging the wide
  expressiveness gap between choreographic programming languages and P4. %
}%
\section{Conclusion and Future Work}
\label{sec:conclusion}

\subparagraph{Conclusion.}
In this work we addressed the challenge of enforcing session types directly in the network.
We developed a formal model of session-type-based monitors, synthesized
network-level monitors, and proved correctness under suitable network assumptions.
We then designed and implemented \toolName, which generates (1) session-type monitors for P4-enabled switches and (2) APIs for writing communicating programs tracked by these monitors.
Across diverse settings and protocols, our evaluation shows accurate blocking of incorrect messages while allowing correct ones, with low network overhead.

To our knowledge, this is the first work to leverage session types to
implement network-level monitors for application-level properties. %
\change{\#123B: better conclusion}{%
\toolName demonstrates that it is possible to automatically synthesize network-level monitors from application-level protocols and deploy these on programmable network switches. Our results have limitations, in part due to the restrictions of P4 – e.g., the fact that the input protocols must be finite-state and half-duplex; still, we demonstrate that even with these restrictions, network-level session monitoring can support complex multiparty protocols.
}

\subparagraph{Future Work.}%
\label{sec:conclusion-futurework}
Although our work is a first step toward network monitoring based on session types, several theoretical and practical challenges remain.

\smallskip\emph{Towards monitoring completeness.}
A natural next step is the dual of soundness (\Cref{thm:net-mon-bisim}): \emph{completeness}, i.e., rejection of all bad messages.
Proving completeness requires a precise characterization of ``bad'' messages (see \Cref{ex:bad-network-message,ex:st-monitor-reductions}), and~\cite[Theorem~21]{BuFrSc21} suggests that sound and complete monitoring may be unattainable in our setting.
Instead, we conjecture that our monitors are maximally strict: for any $\stT \neq \stEnd$, if any accepting transition of $\enc{\stT}$ is turned into \texttt{reject}, then there exists a network $\netN$ that satisfies the hypotheses of \Cref{thm:net-mon-bisim} but not its thesis.

\change{\#123C: formalising TCP monitors}{%
  \smallskip\emph{Formalising TCP-oriented monitors}. %
  Our formal model (\Cref{sec:formal-model}) focuses on the ``core logic'' of
  network-level session monitoring in an idealised network with perfect message
  delivery. This abstraction allows us to highlight the differences between our
  network-level monitors and previous work on application-level session
  monitoring
  \cite{BoChDeHoYo13,BoChDeHoYo17,DBLP:journals/fmsd/DemangeonHHNY15}; extending
  our formal model to cover TCP-oriented monitors under message duplication,
  loss, and reordering is valuable and challenging future work. It would require
  developing a (partial) formalisation of TCP, which is a significant
  undertaking worth a separate paper, as evidenced by previous work in this area
  (e.g.,
  \cite{DBLP:journals/scp/LockefeerWF16,DBLP:journals/jacm/BishopFMNRSSW19}).
  Therefore, we chose to focus our formalisation on the core monitor logic and
  empirically validate the TCP-oriented extension outlined in
  \Cref{sec:implementation-practical-protocols}.%
}

\change{\#123B,\#123C: clarify compatibility of input session types}{%
  \smallskip\emph{Ensuring properties of \toolName's input session types.} %
  As mentioned in \Cref{footnote:global-types-proj} and
  \Cref{remark:determining-liveness-half-duplex}, the current version of
  \toolName assumes that the local session types given as input are part of a
  multiparty protocol that is output-live (\Cref{def:live-net}) and half-duplex.
  \toolName can be extended to check and guarantee these properties, e.g., via
  bounded model checking, or by interfacing to existing tools (such as Scribble,
  $\nu$Scr, mpstk) to project local session types out of global types. %
  This extension would make \toolName more user-friendly %
  without impacting its core functionality (i.e., monitor synthesis)
  and the contributions of this work.%
}%

\smallskip\emph{Encryption.} %
End-to-end encryption below the session header is compatible with \toolName.
However, the current version assumes headers down to the session header are unencrypted, which may leak information.
Supporting encryption of packet and session headers is future work, potentially building on P4 encrypted-protocol techniques~\cite{HauserIPsecP42020,HauserMACsecP42020} and homomorphic encryption~\cite{Ge09}.

\smallskip\emph{Packet fragmentation.}
Another direction concerns packet fragmentation.
The current version of \toolName assumes a one-to-one correspondence
between session-type messages and network packets. This is often acceptable in modern IP networks with consistent MTUs, but application-level messages can still span multiple packets.
\toolName and our session API (\Cref{sec:implementation-st-api}) could
be extended to support messages spanning multiple packets, while still
avoiding network-level fragmentation, by adding a sequence number or
flag to the session header (\Cref{sec:implementation-st-headers}) to
indicate whether the current message continues in the next packet.

\smallskip\emph{Supporting other protocols and consistency models.}
Finally, we plan to broaden our support for diverse transport and consistency requirements. 
Our TCP-oriented monitor implementation should adapt to other reliable ordered transports, such as \texttt{QUIC}~\cite{rfc9000} and \texttt{SCTP}~\cite{rfc2960}. 
Beyond transport, we aim to extend \toolName to enforce network-level consistency models. Recent work has fruitfully connected session types with consistency guarantees such as causal consistency~\cite{MePe17} and linearizability~\cite{SoKr24}. In parallel, other work has shown how to formally enforce a range of consistency models for programmable network~\cite{YaSoLi18, Zh21, AmBaAmJu26}.
 
\bibliographystyle{plainurl}
\bibliography{ecoop26}

\newpage
\appendix
\renewcommand{\thesection}{\Alph{section}}

\iftodo
\fi

\section*{Appendix}

\section{Formal Model: Additional Details and Proofs}
\label{app:formal-model}

\begin{example}
    \label{ex:st-monitor-reductions-2}%
    Consider the session type: (for brevity, we omit the payload types)

    \smallskip%
    \centerline{\(%
    \stT \;=\; \left\{\begin{array}{@{}l@{}}
            \roleP \stIntC \stLabFmt{a} \stSeq \roleQ \stExtC \left\{\begin{array}{@{}l@{}}
                                                                \stLabFmt{c} \stSeq \stEnd \\
                                                                \stLabFmt{d} \stSeq \stEnd
                                                            \end{array}\right.
            \\%
            \roleP \stIntC \stLabFmt{b} \stSeq \roleQ \stExtC \stLabFmt{d} \stSeq \stEnd
        \end{array}\right.%
    \)}%
    \smallskip%

    By \Cref{def:monitor-state-st}, the corresponding monitor $\enc{\stT}$ can
    immediately accept not only the two top-level outputs to $\roleP$, but also
    the inputs from $\roleQ$ --- which appear later in $\stT$, but may be
    delivered earlier by the surrounding network. %
    For the top-level outputs, we have (by rule \textsc{STMon-IntC} in \Cref{def:monitor-state-st}):

    \smallskip%
    \centerline{\(%
            {\transition
                {\enc{\stT}}
                {\monLabYes{\stSendTag{\roleP}{\stLabFmt{a}}{}}}
                {\enc{\roleQ \stExtC \left\{\begin{array}{@{}l@{}}
                                         \stLabFmt{c} \stSeq \stEnd\\
                                         \stLabFmt{d} \stSeq \stEnd
                                     \end{array}\right.}}}
        \qquad\text{and}\qquad
            {\transition
                {\enc{\stT}}
                {\monLabYes{\stSendTag{\roleP}{\stLabFmt{b}}{}}}
                {\enc{\roleQ \stExtC \stLabFmt{d}} \stSeq \stEnd}}
    \)}%
    \smallskip%

    \noindent%
    Notice that the choice of message sent to $\roleP$ restricts what the
    monitor will accept from $\roleQ$ afterwards.
    The monitor $\enc{\stT}$ can also immediately accept the inputs from $\roleQ$,
    by rule~\textsc{STMon-IntPfx} in \Cref{def:monitor-state-st}. %
    The transition $\monLabYes{\stEnqueueTag{\roleQ}{\stLabFmt{c}}{\stS}}$
    is fired by the following derivation:

    \smallskip%
    \centerline{\(%
    \small%
        \Rule
        {
            \Rule
            {}
            {
                {\transition
                    {\enc{\roleQ \stExtC \left\{\begin{array}{@{}l@{}}
                                                    \stLabFmt{c} \stSeq \stEnd \\
                                                    \stLabFmt{d} \stSeq \stEnd
                                                \end{array}\right.}}
                    {\monLabYes{\stEnqueueTag{\roleQ}{\stLabFmt{c}}{}}}
                    {\enc{\stEnd}}}
            }
            {STMon-ExtC}
        }
        {\transition
            {\enc{\stT}}
            {\monLabYes{\stEnqueueTag{\roleQ}{\stLabFmt{c}}{}}}
            {\enc{\roleP \stIntC \stLabFmt{a} \stSeq \stEnd}}}
        {STMon-IntPfx}
    \)}%
    \smallskip%

    Instead, the transition $\monLabYes{\stEnqueueTag{\roleQ}{\stLabFmt{d}}{}}$
    is fired by the following derivation:

    \smallskip%
    \centerline{\(%
    \small%
        \Rule
        {
            \Rule
            {}
            {
                {\transition
                    {\enc{\roleQ \stExtC \left\{\begin{array}{@{}l@{}}
                                                    \stLabFmt{c} \stSeq \stEnd \\
                                                    \stLabFmt{d} \stSeq \stEnd
                                                \end{array}\right.}}
                    {\monLabYes{\stEnqueueTag{\roleQ}{\stLabFmt{d}}{}}}
                    {\enc{\stEnd}}}
            }
            {STMon-ExtC}
            \quad%
            \Rule
            {}
            {
                {\transition
                    {\enc{\roleQ \stExtC \stLabFmt{d} \stSeq \stEnd}}
                    {\monLabYes{\stEnqueueTag{\roleQ}{\stLabFmt{d}}{}}}
                    {\enc{\stEnd}}}
            }
            {STMon-ExtC}
        }
        {\transition
            {\enc{\stT}}
            {\monLabYes{\stEnqueueTag{\roleQ}{\stLabFmt{d}}{}}}
            {\enc{\roleP \stIntC \left\{\begin{array}{@{}l@{}}
                                             \stLabFmt{a} \stSeq \stEnd
                                             \\%
                                             \stLabFmt{b} \stSeq \stEnd
                                        \end{array}\right.}}}
        {STMon-IntPfx}
    \)}%
    \smallskip%

    \noindent%
    Observe that in the first case, the message $\stLabFmt{c}$ from $\roleQ$
    causes rule \textsc{STMon-IntPfx} to ``prune'' the session type
    $\stT$ in the monitor state, because in its internal choice there is only one
    branch in which the message $\stLabFmt{c}$ from $\roleQ$ can be received.
    After that, the monitor only accepts sending message $\stLabFmt{a}$ to
    $\roleP$, as it is the only output compatible with $\roleQ$.
    Instead, in the second case, the message $\stLabFmt{d}$ from $\roleQ$ causes
    rule \textsc{STMon-IntPfx} to keep both branches of $\stT$ in the monitor
    state: this is because in both branches the message $\stLabFmt{d}$ from
    $\roleQ$ can be received. After that, the monitor still allows sending
    either $\stLabFmt{a}$ or $\stLabFmt{b}$ to $\roleP$, as both options are
    compatible with $\roleQ$.

    This strategy for accepting messages is necessary
    because, depending on the overall multiparty interaction, messages from
    $\roleQ$ may be delivered before those to $\roleP$ are sent. This phenomenon is further
    illustrated in \Cref{ex:complex-mon-st-progression} below.
\end{example}

\begin{definition}[Half-Duplex Network]
    \label{def:half-duplex-net}%
    Write $\stQMsg{\roleP}{\stLab}{\stS} \in \stQ$ if queue $\stQ$ contains message $\stQMsg{\roleP}{\stLab}{\stS}$.
    We say a network $\netN$ is \emph{half-duplex} if, whenever
    $\transitionS{\netN}{\netTau}{\netNi} = \netBigPar{i \in I}{\netPair{\roleP[i]}{(\stqPair{\stT[i]}{\stQ[i]})}}$,
    then for all $i,j \in I$, $\stQMsg{\roleP[j]}{\stLab}{\stS} \in \stQ[i]$
    (for some $\stLab, \stS$)
    implies $\stQMsg{\roleP[i]}{\stLabi}{\stSi} \not\in \stQ[j]$
    (for any $\stLabi, \stSi$).
\end{definition}

\begin{example}[On the half-duplex restriction]
    \label{ex:half-duplex-motivation}
    Consider the following session types:
    \[
    \stT[\roleP] \;=\; \stRec{\stRecVar}{\stOut{\roleQ}{\stLabFmt{a}}{\stS} \stSeq \stIn{\roleQ}{\stLabFmt{b}}{\stS} \stSeq \stRecVar}
    \qquad{and}\qquad%
    \stT[\roleQ] \;=\; \stRec{\stRecVari}{\stOut{\roleP}{\stLabFmt{b}}{\stS} \stSeq \stIn{\roleP}{\stLabFmt{a}}{\stS} \stSeq \stRecVari}
    \]
    Consider the following simple unmonitored network $\netN$ with the two
    session types above plus input queues, for the end host roles $\roleP$ and
    $\roleQ$:
    \[
        \netN \;\;=\;\;%
        \netPair{\roleP}{\stqPair{\stT[\roleP]}{\stQEmpty}}%
        \;\netPar\;%
        \netPair{\roleQ}{\stqPair{\stT[\roleQ]}{\stQEmpty}}%
    \]
    This network executes with $\roleP$ and $\roleQ$ sending each other the
    messages $\stLabFmt{a}$ and $\stLabFmt{b}$, respectively. The messages are
    delivered to the input queues of $\roleP$ and $\roleQ$, reaching the network
    configuration $\netNi$ below:
    \[
        \netNi \;\;=\;\;%
        \netPair{\roleP}{\stqPair{\stIn{\roleQ}{\stLabFmt{b}}{\stS} \stSeq \stT[\roleP]}{\stQCons{\stQMsg{\roleQ}{\stLabFmt{b}}{\stS}}{\stQEmpty}}}%
        \;\netPar\;%
        \netPair{\roleQ}{\stqPair{\stIn{\roleP}{\stLabFmt{a}}{\stS} \stSeq \stT[\roleQ]}{\stQCons{\stQMsg{\roleP}{\stLabFmt{b}}{\stS}}{\stQEmpty}}}%
    \]
    Then, the session types of $\roleP$ and $\roleQ$ consume the respective input messages and empty the respective queues,
    looping back to the network configuration $\netN$.

    Observe that $\netN$ is output-live (by \Cref{def:live-net}) but is \emph{not}
    half-duplex (by \Cref{def:half-duplex-net}), because it can reach the
    configuration $\netNi$ above where the input queue of $\roleP$ contains a
    message from $\roleQ$, and \emph{vice versa}.

    Therefore, $\netN$ above does \emph{not} satisfy the hypotheses of
    \Cref{thm:net-mon-bisim}. And indeed, consider the instrumented version of
    $\netN$, by \Cref{def:consistent-mon-instr}:
    \[
        \mnetInstr{\netN} \;\;=\;\;%
        \mnetPair{\roleP}{\mstPair{\stqPair{\stT[\roleP]}{\stQEmpty}}{\enc{\stT[\roleP]}}}%
        \;\netPar\;%
        \mnetPair{\roleQ}{\mstPair{\stqPair{\stT[\roleQ]}{\stQEmpty}}{\enc{\stT[\roleQ]}}}%
    \]
    By \Cref{def:monitor-state-st}, the monitor $\enc{\stT[\roleP]}$ only
    accepts an outgoing message to $\roleQ$, and rejects any incoming message from $\roleQ$; and similarly,
    $\enc{\stT[\roleQ]}$ only accepts an outgoing message to $\roleP$, and rejects any incoming message from $\roleP$. %
    Therefore, the monitored network $\mnetInstr{\netN}$ cannot
    advance to a configuration corresponding to $\netNi$ above.

    In principle, one may attempt lifting the half-duplex restriction of \Cref{thm:net-mon-bisim} (and thus, covering the example above)
    by allowing our monitors to accept inputs coming from a role $\roleP$ even if such inputs are expected after an output to $\roleP$. %
    To this end, \Cref{def:monitor-state-st} may be relaxed by extending
    rules \tsc{STMon-IntPfx} and \tsc{STMon-ExtPfx}
    to ``skip'' internal choices from $\roleP$. %

    However, this extension to \Cref{def:monitor-state-st}, would cause a
    problem: the monitor $\enc{\stT[\roleP]}$ would become infinite-state, by
    accepting unbounded sequences of inputs from $\roleQ$, similarly to
    \Cref{ex:infinite-state-monitor}. The same issue would affect the monitor
    $\enc{\stT[\roleQ]}$. In general, this relaxed version of
    \Cref{def:monitor-state-st} would cause most non-trivial session types with
    recursion (including many examples we evaluate in \Cref{sec:evaluation})
    to generate infinite-state monitors that would not be representable in
    finite P4 tables by \toolName (\Cref{sec:design}).
\end{example}

\subsection{Proofs}
\label{appendix:SessionTypesAddionalDetails}

\begin{definition}[Auxiliary notation]
    \label{def:st-unfolding}%
    Given a session type $\stT$, %
    we define its \emph{unfolding} $\unfold{\stT} =
    \unfold{\stFmt{\subst{\stTi}{\stRecVar}{\stT}}}$ if $\stT =
    \stRec{\stRecVar}{\stTi}$, and $\unfold{\stT} = \stT$ otherwise.
\end{definition}

\begin{proposition}
    \label{lem:live-persistent}%
    If $\netN$ is output-live and $\transitionS{\netN}{\netTau}{\netNi}$, then $\netNi$
    is output-live.
\end{proposition}
\begin{proof}
    Similar to the proof of \cite[Prop.~4.9]{DBLP:journals/tocl/GhilezanPPSY23},
    but using our output-liveness (\Cref{def:live-net}) instead of their
    liveness \cite[Def.~4.7]{DBLP:journals/tocl/GhilezanPPSY23}.
\end{proof}

\begin{lemma}
    \label{lemma:mon-trans-congruent}%
    If $\stQ \equiv \stQ'$ and %
    $\transition{\enc{\stT}}{\monLabYes{\inTag\stQ}}{\enc{\stTi}}$,
    then %
    $\transition{\enc{\stT}}{\monLabYes{\inTag\stQi}}{\enc{\stTi}}$,
\end{lemma}
\begin{proof}
    We first prove that if two queues $\stQii$ and $\stQiii$ are equal except
    for the order of two consecutive messages having different senders, then
    $\transition{\enc{\stT}}{\monLabYes{\inTag\stQii}}{\enc{\stTi}}$ implies
    $\transition{\enc{\stT}}{\monLabYes{\inTag\stQiii}}{\enc{\stTi}}$.

    We then prove the main statement by observing that, by hypothesis and
    \Cref{def:queue-congruence}, there is a sequence of $n$ applications of
    $\equiv$ such that $\stQ = \stQ[0] \equiv \stQ[1] \equiv \ldots \equiv
    \stQ[n] = \stQi$, and for all $i \in 1..n-1$, the queue $\stQ[i]$ is
    equal to $\stQ[i+1]$ except for the order of two consecutive messages having
    different senders. %
    Therefore, for $i \in 1..n-1$, $\transition{\enc{\stT}}{\monLabYes{\inTag\stQ[i]}}{\enc{\stTi}}$ implies
    $\transition{\enc{\stT}}{\monLabYes{\inTag\stQ[i+1]}}{\enc{\stTi}}$ (by the result above),
    and we obtain the thesis.
\end{proof}

\begin{corollary}
    \label{lemma:mon-congruent}%
    If $\stQ \equiv \stQ'$ and $\enc{\stqPair{\stT}{\stQ}}$ is defined, then
    $\enc{\stqPair{\stT}{\stQ}} = \enc{\stqPair{\stT}{\stQi}}$.
\end{corollary}
\begin{proof}
    Consequence of \Cref{lemma:mon-trans-congruent} and \Cref{def:monitor-state-st-queue}.
\end{proof}

\begin{lemma}
    \label{lem:type-mon-dequeuing}%
    Assume that $\mnetInstr{\mnetPair{\roleP}{\stqT}}$ is defined. Then, we have that %
    $\transition{\netPair{\roleP}{\stqT}}
                {\netDequeueLab{\roleP}{\roleQ}{\stLab}{\stS}}
                {\netPair{\roleP}{\stqTi}}$
    implies
    $\transition{\mnetInstr{\netPair{\roleP}{\stqT}}}
                {\mnetDequeueLab{\roleP}{\roleQ}{\stLab}{\stS}}
                {\mnetInstr{\netPair{\roleP}{\stqTi}}}$.
\end{lemma}
\begin{proof}
    By inversion of rule \textsc{Net-Dequeue} in \Cref{def:network-semantics}
    we have $\transition{\stqT}{\stIntTag{\dequeueTag{\roleQ}{\stLab}{\stS}}}{\stqTi}$, and thus:
    \begin{enumerate}
        \item\label{item:proof:type-mon-dequeuing:sttii-stqi}
            By inversion of rule \textsc{SQ-Dequeue} in \Cref{def:local-semantics-queues},
            $\stqT = \stqPair{\stT}{\stQ}$ such that, for some $\stTii$ and $\stQi$, %
            $\transition{\stT}{\stInTag{\roleP}{\stLab}{\stS}}{\stTii}$ and $\stQ \equiv
            \stQCons{\stQMsg{\roleQ}{\stLab}{\stS}}{\stQi}$.
        \item\label{item:proof:type-mon-dequeuing:mon-cont-shape}
            Therefore, by rule \textsc{SQ-Dequeue} in
            \Cref{def:local-semantics-queues} we have $\stqTi = \stqPair{\stTii}{\stQi}$.
        \item\label{item:proof:type-mon-dequeuing:mon-yes} By
            \Cref{item:proof:type-mon-dequeuing:sttii-stqi} and
            \Cref{def:local-semantics}, $\unfold{\stT}$ %
            must be an external choice from $\roleQ$ including a message with
            label $\stLab$ and payload type $\stS$. Thus, by
            \Cref{def:monitor-state-st}, %
            we have
            $\transition{\enc{\stT}}{\monLabYes{\stEnqueueTag{\roleQ}{\stLab}{\stS}}}{\enc{\stTii}}$
            (using rule \textsc{STMon-Rec} to unfold $\stT$ if needed, and then
            \textsc{STMon-ExtC}).
        \item\label{item:proof:type-mon-dequeuing:mon-cont-equiv} By
            \Cref{item:proof:type-mon-dequeuing:sttii-stqi},
            \Cref{item:proof:type-mon-dequeuing:mon-cont-shape},
            \Cref{item:proof:type-mon-dequeuing:mon-yes},
            \Cref{def:monitor-state-st-queue}, and
            \Cref{lemma:mon-trans-congruent}, %
            $\enc{\stqT} = \enc{\stqPair{\stT}{\stQ}} = \enc{\stqPair{\stTii}{\stQi}} = \enc{\stqTi}$.
        \item\label{item:proof:type-mon-deueuing:mon-trans}
            Hence, by rule \textsc{M-Dequeue} in \Cref{def:st-monitor-semantics}
            and \Cref{item:proof:type-mon-dequeuing:mon-cont-equiv},
            we have %
            $\transition
                {\mstPair{\stqT}{\enc{\stqT}}}
                {\monLabSilent{\dequeueTag{\roleQ}{\stLab}{\stS}}}
                {\mstPair{\stqTi}{\enc{\stqT}}} =
                {\mstPair{\stqTi}{\enc{\stqTi}}}$.
        \item\label{item:proof:type-mon-deueuing:mon-net-trans} Thus, by
            \Cref{item:proof:type-mon-deueuing:mon-trans} and rule
            \textsc{MNet-Dequeue} in \Cref{def:monitored-network-semantics} we have \thesisNewline
            $\transition{\mnetPair{\roleP}{\mstPair{\stqT[1]}{\enc{\stqT}}}}
                        {\mnetDequeueLab{\roleP[i]}{\roleQ}{\stLab}{\stS}}
                        {\mnetPair{\roleP}{{\mstPair{\stqTi}{\enc{\stqTi}}}}}$.
    \end{enumerate}
    Therefore, by \Cref{item:proof:type-mon-deueuing:mon-net-trans}
    and \Cref{def:consistent-mon-instr}
    we obtain the thesis.
\end{proof}

\begin{lemma}
    \label{lem:mon-type-dequeuing}%
    If %
    $\transition{\mnetInstr{\netPair{\roleP}{\stqT}}}
                {\mnetDequeueLab{\roleP}{\roleQ}{\stLab}{\stS}}
                {\mnetNi}$, then
    $\exists \stqTi$ such that %
    $\transition{\netPair{\roleP}{\stqT}}
                {\netDequeueLab{\roleP}{\roleQ}{\stLab}{\stS}}
                {\netPair{\roleP}{\stqTi}}$
    and %
    $\mnetNi = \mnetInstr{\netPair{\roleP}{\stqTi}}$.
\end{lemma}
\begin{proof}
    By inversion of rule \textsc{MNet-Dequeue} in \Cref{def:monitored-network-semantics}, %
    rule \textsc{M-Dequeue} in \Cref{def:st-monitor-semantics}, %
    and rule \textsc{SQ-Dequeue} in \Cref{def:local-semantics-queues} %
    we infer the same shape of $\stqT$ and $\stqTi$ obtained in the proof of
    \Cref{lem:type-mon-dequeuing} above, %
    and thus, the same transition
    $\transition{\netPair{\roleP}{\stqT}}
                {\netDequeueLab{\roleP}{\roleQ}{\stLab}{\stS}}
                {\netPair{\roleP}{\stqTi}}$. %
    Then, we conclude $\mnetNi = \mnetInstr{\netPair{\roleP}{\stqTi}}$
    by \Cref{def:consistent-mon-instr}.%
\end{proof}

\begin{proposition}
    \label{lem:type-send-mon-send}
    If \(\transition{\stqPair{\stT}{\stQ}}
                    {\stSendTag{\roleP}{\stLab}{\stS}}
                    {\stqPair{\stTi}{\stQ}}\),
    then
    \(\transition{\mstPair{\stqPair{\stT}{\stQ}}{\enc{\stqPair{\stT}{\stQ}}}}
                    {{\monLabYes{\stSendTag{\roleP}{\stLab}{\stS}}}}
                    {{\mstPair{\stqPair{\stTi}{\stQ}}{\enc{\stqPair{\stTi}{\stQ}}}}}\).
\end{proposition}
\begin{proof}
    By inversion of rule \textsc{SQ-Send} in \Cref{def:local-semantics-queues},
    $\unfold{\stT}$ %
    must be an internal choice that sends a message to $\roleP$ with label
    $\stLab$ and payload type $\stS$. Then, the result is direct consequence of
    rules \textsc{STMon-IntC} (and possibly \textsc{STMon-Rec})
    in \Cref{def:monitor-state-st-queue},
    and \Cref{def:st-monitor-semantics}.
\end{proof}

\newcommand{\ioSeq}{\tilde{\stFmt{\rho}}}%
\newcommand{\ioSeqi}{\tilde{\stFmt{\rho'}}}%
\newcommand{\ioSeqii}{\tilde{\stFmt{\rho''}}}%
\newcommand{\ioSeqiii}{\tilde{\stFmt{\rho'''}}}%
\newcommand{\ioSeqiiii}{\tilde{\stFmt{\rho''''}}}%

\begin{lemma}
    \label{lem:st-mon-action-seq}%
    For any sequence of input/output labels $\ioSeq$, %
    $\transition{\stT}{\ioSeq}{\stTi}$ implies
    $\transition{\enc{\stT}}{\monLabYes{\ioSeq}}{\enc{\stTi}}$.
\end{lemma}
\begin{proof}
    From \Cref{def:monitor-state-st} we have that for all $\stExtTag$,
    $\transition{\stT}{\stExtTag}{\stTi}$
    implies $\transition{\enc{\stT}}{\monLabYes{\stExtTag}}{\enc{\stTi}}$.
    We then prove the thesis by induction on $\ioSeq$.
\end{proof}

\begin{lemma}
    \label{lem:mon-st-send}%
    If $\transition{\enc{\stT}}{\monLabYes{\stSendTag{\roleP}{\stLab}{\stS}}}{\enc{\stTi}}$,
    then $\transition{\stT}{\stSendTag{\roleP}{\stLab}{\stS}}{\stTi}$.
\end{lemma}
\begin{proof}
    Immediate by \Cref{def:monitor-state-st}.
\end{proof}

\newcommand{\futureSet}{\mathcal{F}}%
\newcommand{\futureSeti}{\mathcal{F}'}%
\newcommand{\futureSetii}{\mathcal{F}''}%

\begin{lemma}
    \label{lem:mon-st-enqueue-approx}%
    If $\transition{\enc{\stT}}{\monLabYes{\stEnqueueTag{\roleP}{\stLab}{\stS}}}{\enc{\stTi}}$,
    then there is a non-empty set $\futureSet$ containing all pairs $(\ioSeq, \stTii)$
    where $\ioSeq$ is a (possibly empty) sequence of input/output labels not involving $\roleP$
    such that $\transition{\transition{\stT}{\ioSeq}{\!\!}}{\stInTag{\roleP}{\stLab}{\stS}}{\stTii}$. %
    Moreover, for all pairs $(\ioSeq, \stTii) \in \futureSet$, %
    we have:
    \begin{enumerate}
    \item\label{item:mon-st-enqueue-approx:reachable-ti} $\transition{\stTi}{\ioSeq}{\stTii}$, and
    \item\label{item:mon-st-enqueue-approx:reachable-mon-t} $\transition{\transition{\enc{\stT}}{\monLabYes{\ioSeq}}{\!\!}}{\monLabYes{\stEnqueueTag{\roleP}{\stLab}{\stS}}}{\enc{\stTii}}$.
    \end{enumerate}
\end{lemma}
\begin{proof}
    To prove the first part of the statement, we construct $\futureSet$ as:
    \[
        \futureSet \;=\;
        \left\{ (\ioSeq, \stTii) \;\middle|\;%
          \ioSeq\text{ does not involve $\roleP$ and }
          \transition{\transition{\stT}{\ioSeq}{\!\!}}{\stInTag{\roleP}{\stLab}{\stS}}{\stTii}
        \right\}
    \]
    \newcommand{\someDeriv}[1]{\mathcal{D}_{#1}}%
    Observe that $\futureSet$ above is not empty: %
    this is because, by hypothesis and \Cref{def:monitor-state-st}, %
    $\stT$ must contain some external choices containing message $\stLab$ and payload $\stS$,
    which can be fired by the axiom \tsc{STMon-ExtC} in a derivation based on rules \tsc{STMon-Rec}, \tsc{STMon-IntPfx}, \tsc{STMon-ExtPfx}.
    More precisely, this means that there exists and indexing set $K \neq \emptyset$
    such that each of the aforementioned external choices in $\stT$ has the form %
    $\stExtSum{i \in I_k}{\roleP}{\stChoice{\stLab[i]}{\stS[i]} \stSeq \stT[i]}$ for $k \in K$,
    such that $\forall k \in K: \exists i \in I_k: \stLab[i] = \stLab$ and $\stS[i] = \stS$. %
    We construct each $\ioSeq$ in the definition of $\futureSet$ above by induction on the derivation that proves the hypothesis $\transition{\enc{\stT}}{\monLabYes{\stEnqueueTag{\roleP}{\stLab}{\stS}}}{\enc{\stTi}}$,
    obtaining one transition in $\ioSeq$ for each application of rules \tsc{STMon-IntPfx} and \tsc{STMon-ExtPfx} (which, by \Cref{def:monitor-state-st}, does not involve $\roleP$);
    let us call that derivation $\someDeriv{\ioSeq}$.

    For the \emph{``moreover\ldots''} part of the statement, given $\ioSeq$ we obtain item~\ref{item:mon-st-enqueue-approx:reachable-ti}
    by induction the derivation $\someDeriv{\ioSeq}$ above, turning each application of rules \tsc{STMon-IntPfx} and \tsc{STMon-ExtPfx} %
    into an application of \tsc{S-IntC} and \tsc{S-ExtC}, %
    from \Cref{def:local-semantics-no-queues}, respectively.%

    Finally, item~\ref{item:mon-st-enqueue-approx:reachable-mon-t} is direct consequence of \Cref{lem:st-mon-action-seq}.
\end{proof}

\begin{lemma}
    \label{lem:mon-st-enqueue-approx-seq}
    If $\transition{\enc{\stT}}{\enqueueTag\monLabYes{\stQ}}{\enc{\stTi}}$, then
    there is a non-empty set $\futureSet$ containing all pairs $(\ioSeq, \stTii)$
    where $\ioSeq$ is a sequence of input/output labels including inputs for all
    messages in $\stQ$ maintaining sender ordering and ending with one of such
    inputs,
    such that $\transition{\stT}{\ioSeq}{\stTii}$. %
    Moreover, for all pairs $(\ioSeq, \stTii) \in \futureSet$, %
    we have:
    \begin{enumerate}
    \item\label{item:mon-st-enqueue-approx-seq:reachable-ti} $\transition{\stTi}{\ioSeq \setminus \enqueueTag\stQ}{\stTii}$, and
    \item\label{item:mon-st-enqueue-approx-seq:reachable-mon-t} $\transition{\enc{\stT}}{\monLabYes{\ioSeq}}{\enc{\stTii}}$.
    \end{enumerate}
\end{lemma}
\begin{proof}
    We construct $\futureSet$ by induction on $\stQ$. The base case (i.e., when $\stQ$ is empty) is immediate.
    In the inductive case $\stQ = \stQCons{\stQi}{\stQMsg{\roleP}{\stLab}{\stS}}$, %
    we have $\transition{\enc{\stT}}{\enqueueTag\monLabYes{\stQi}}{\enc{\stTiii}}$
    for some $\stTiii$ such that $\transition{\enc{\stTiii}}{\monLabYes{\stEnqueueTag{\roleP}{\stLab}{\stS}}}{\enc{\stTi}}$.
    By the i.h.~we have a set $\futureSeti$ containing all pairs $(\ioSeqi, \stTiiii)$
    where $\ioSeqi$ is a sequence of input/output labels including inputs for all
    messages in $\stQi$
    such that $\transition{\stT}{\ioSeqi}{\stTiiii}$; %
    moreover, $\transition{\stTiii}{\ioSeqi \setminus \enqueueTag\stQi}{\stTiiii}$ and
    $\transition{\enc{\stT}}{\monLabYes{\ioSeqi}}{\enc{\stTiiii}}$. %
    We now construct $\futureSet$ from $\futureSeti$, by inspecting each pair
    $(\ioSeqi, \stTiiii) \in \futureSeti$ and determining whether we can select
    and extend the pair to include new actions ensuring that
    $\stEnqueueTag{\roleP}{\stLab}{\stS}$ is fired and a desired $\stTii$ is reached.%

    Recalling that by hypothesis we have $\transition{\enc{\stTiii}}{\monLabYes{\stEnqueueTag{\roleP}{\stLab}{\stS}}}{\enc{\stTi}}$,
    we apply \Cref{lem:mon-st-enqueue-approx} to obtain %
    a non-empty set $\futureSetii$ containing all pairs $(\ioSeqii, \stTiiiii)$
    where $\ioSeqii$ is a (possibly empty) sequence of input/output labels not involving $\roleP$
    such that $\transition{\transition{\stTiii}{\ioSeqii}{\!\!}}{\stInTag{\roleP}{\stLab}{\stS}}{\stTiiiii}$. %
    Moreover, for all pairs $(\ioSeqii, \stTiiiii) \in \futureSetii$, %
    we have $\transition{\stTi}{\ioSeqii}{\stTiiiii}$ and
    $\transition{\transition{\enc{\stTiii}}{\monLabYes{\ioSeqii}}{\!\!}}{\monLabYes{\stEnqueueTag{\roleP}{\stLab}{\stS}}}{\enc{\stTiiiii}}$.
    We construct the desired set $\futureSet$ as follows:
    \[
    \futureSet \;=\; \left\{ (\ioSeq, \stTii) \;\middle|\;
        \begin{array}{@{}l@{}}
            \exists (\ioSeqi, \_) \in \futureSeti:%
            \exists (\ioSeqii, \stTii) \in \futureSetii:%
            \\
            (\ioSeqi \setminus \enqueueTag\stQi) \;\text{ is a prefix of }\; \ioSeqii%
            \\%
            \ioSeq = \ioSeqi \cdot (\ioSeqii \setminus \ioSeqi)
        \end{array}
    \right\}
    \]
    i.e., from all the input/output sequences $\ioSeqi$ in $\futureSeti$, we
    only select those $\ioSeqi$ that from $\stTiii$ can lead to performing the action
    $\stInTag{\roleP}{\stLab}{\stS}$ %
    (since $\ioSeqi \setminus \enqueueTag\stQi$ is a prefix of some $\ioSeqii$
    in $\futureSetii$, which contains all sequences from $\stTiii$ that
    eventually perform $\stInTag{\roleP}{\stLab}{\stS}$ as first input from
    $\roleP$); %
    then, we concatenate $\ioSeqi$ and $\ioSeqii$ (avoiding the duplication of
    the common input/output actions), %
    thus obtaining a $\ioSeq$ that performs $\stInTag{\roleP}{\stLab}{\stS}$ and
    such that $\transition{\stT}{\ioSeq}{\stTii}$.

    For the \emph{``moreover\ldots''} part of the statement, we obtain item~\ref{item:mon-st-enqueue-approx-seq:reachable-ti}
    from the construction of $\futureSet$ above, %
    while item~\ref{item:mon-st-enqueue-approx-seq:reachable-mon-t} is direct consequence of \Cref{lem:st-mon-action-seq}.
\end{proof}

\begin{lemma}
    \label{lem:live-hd-net-enqueue-mon-enqueue}
    Assume $\netN$ is an output-live, half-duplex network such that
    \(
    \transition{\netN}{\netCommLab{\roleP}{\roleQ}{\stLab}{\stS}}{\netNi}
    \), containing $\netPair{\roleQ}{\stqPair{\stT}{\stQ}}$.
    Assume $\mnetInstr{\netN}$ is defined. %
    Then,\thesisNewline
    \(
      \transition{\mstPair{\stqPair{\stT}{\stQ}}{\enc{\stqPair{\stT}{\stQ}}}}
                 {\monLabYes{\stEnqueueTag{\roleP}{\stLab}{\stS}}}
                 {\mstPair{\stqPair{\stT}{\stQ}}{\enc{\stqPair{\stT}{\stQCons{\stQ}{\stQMsg{\roleP}{\stLab}{\stS}}}}}}
    \).
\end{lemma}
\begin{proof}
    By hypothesis, we have:
    \begin{align}
        &\label{eq:proof:live-hd-net-enqueue-mon-enqueue:msg-enqueued}%
        \netNi \;\text{ contains }\; \netPair{\roleQ}{\stqPair{\stT}{\stQCons{\stQ}{\stQMsg{\roleP}{\stLab}{\stS}}}}
        &\text{(by hypothesis and inversion of $\transition{\netN}{\netCommLab{\roleP}{\roleQ}{\stLab}{\stS}}{\netNi}$)}
    \end{align}
    Also by hypothesis, $\enc{\stqPair{\stT}{\stQ}}$ is defined, i.e.,
    \begin{align}
        &\label{eq:proof:live-hd-net-enqueue-mon-enqueue:mon-exists}%
        \exists \stTi = \enc{\stqPair{\stT}{\stQ}}: \transition{\enc{\stT}}{\enqueueTag\monLabYes{\stQ}}{\enc{\stTi}}
        &\text{(by hypothesis and \Cref{def:monitor-state-st-queue})}
    \end{align}
    Notice that, by \Cref{def:local-semantics} (rule \textsc{SQ-Recv}),
    we have \(\transition{\stqPair{\stT}{\stQ}}{\stEnqueueTag{\roleP}{\stLab}{\stS}}{\stqPair{\stT}{\stQCons{\stQ}{\stQMsg{\roleP}{\stLab}{\stS}}}}\).
    Therefore, by rule \textsc{MNet-$\stExtTag$-Good} of \Cref{def:monitored-network},
    we can prove the statement by proving:%
    \begin{align}
    &\label{eq:proof:live-hd-net-enqueue-mon-enqueue:mon-enq}%
        \transition{\enc{\stTi}}
                   {\monLabYes{\stEnqueueTag{\roleP}{\stLab}{\stS}}}
                   {\enc{\stqPair{\stT}{\stQCons{\stQ}{\stQMsg{\roleP}{\stLab}{\stS}}}}}
    \end{align}
    By contradiction, assume that \eqref{eq:proof:live-hd-net-enqueue-mon-enqueue:mon-enq} does \emph{not} hold,
    and thus, $\forall \monMi: \transition{\enc{\stTi}}{\monLabYes{\stEnqueueTag{\roleP}{\stLab}{\stS}}}{\monMi}$
    implies $\monMi \neq \enc{\stqPair{\stT}{\stQCons{\stQ}{\stQMsg{\roleP}{\stLab}{\stS}}}}$.
    By \Cref{def:monitor-state-st-queue} the quantification ``$\forall \monMi\ldots$'' above must be vacuous,
    i.e.,
    \begin{align}
        &\label{eq:proof:live-hd-net-enqueue-mon-enqueue:contra:mon-not-enq}%
            \not\exists \monMi:%
            \transition{\enc{\stTi}}
                       {\monLabYes{\stEnqueueTag{\roleP}{\stLab}{\stS}}}
                       {\monMi}
        &\text{(by \Cref{def:monitor-state-st-queue})}
    \end{align}

    By \Cref{def:monitor-state-st}, \cref{eq:proof:live-hd-net-enqueue-mon-enqueue:contra:mon-not-enq}
    has two possible implications:
    \begin{enumerate}
    \item%
        All the external choices from $\roleP$ that are reached first in $\stTi$
        cannot syntactically consume the message $\stQMsg{\roleP}{\stLab}{\stS}$
        from the input queue $\stQ$. Observe that, by
        \eqref{eq:proof:live-hd-net-enqueue-mon-enqueue:mon-exists} and
        \Cref{lem:mon-st-enqueue-approx-seq}, $\stT$ consumes its input messages
        $\stQ$ by reducing into a session type $\stTii$ reachable from $\stTi$
        --- and therefore, such $\stTii$ also does not syntactically contain any
        external choice from $\roleP$ that can consume the message
        $\stQMsg{\roleP}{\stLab}{\stS}$ which is in $\roleQ$'s input queue in
        $\netNi$, by
        \eqref{eq:proof:live-hd-net-enqueue-mon-enqueue:msg-enqueued}. This
        means that $\netNi$ is \emph{not} output-live (by \Cref{def:live-net}), and
        therefore $\netN$ is not output-live either (by the contrapositive of
        \Cref{lem:live-persistent}) --- contradiction.
    \item%
        There are external choices in $\stTi$ that could consume the message
        $\stQMsg{\roleP}{\stLab}{\stS}$, but they are prefixed by an internal
        choice towards $\roleP$. Observe that, by
        \eqref{eq:proof:live-hd-net-enqueue-mon-enqueue:mon-exists} and
        \Cref{lem:mon-st-enqueue-approx-seq}, $\stT$ consumes its input messages
        $\stQ$ by reducing into a session type $\stTii$ reachable from $\stTi$
        by firing the same inputs/outputs of $\stTi$ (plus the inputs from
        $\stQ$) in the same order. Therefore, either before or after reducing to
        $\stTii$, $\stT$ must emit an output to $\roleP$ before it can consume
        the message $\stQMsg{\roleP}{\stLab}{\stS}$ in $\roleQ$'s input queue.
        But then, by rule~\textsc{Net-Comm} in
        \Cref{def:network-semantics}, a message from $\roleQ$ to $\roleP$ will
        land in the input queue of $\roleP$, while the message
        $\stQMsg{\roleP}{\stLab}{\stS}$ is still in $\roleQ$'s input queue.
        Therefore, by \Cref{def:half-duplex-net} we conclude that $\netN$ is not
        half-duplex --- contradiction.
    \end{enumerate}

    We have thus proven that if we negate
    \eqref{eq:proof:live-hd-net-enqueue-mon-enqueue:mon-enq} we contradict at
    least one hypothesis in the statement. Therefore,
    \eqref{eq:proof:live-hd-net-enqueue-mon-enqueue:mon-enq} must hold, and this
    (by rule \textsc{MNet-$\stExtTag$-Good} of \Cref{def:monitored-network})
    leads to the thesis.
\end{proof}

\lemMonitorCorrectness*
\begin{proof}
    Take any output-live and half-duplex network $\netN$, and consider the following relation:
    \begin{equation}
        \label{eq:proof-net-mon-bisim:rel}
        {\relR} \;=\; %
        \left\{ \left(\netNi, \mnetInstr{\netNi}\right) \;\middle|\; \transitionS{\netN}{\netTau}{\netNi} \right\}
    \end{equation}
    We now prove that $\relR$ is an internal bisimulation. To this end, we
    inspect each pair $(\netNi, \mnetNii) \in {\relR}$, and we show that the pair
    satisfies clauses~\ref{item:internal-bisim-lr}
    and~\ref{item:internal-bisim-rl} of \Cref{def:internal-bisim}. %
    By \Cref{eq:proof-net-mon-bisim:rel}, the pair has the form
    $(\netNi, \mnetInstr{\netNi})$, and we have the following cases.
    \begin{itemize}
        \item $\exists \roleP,\roleQ,\stLab,\stS,\netNii:%
            \transition{\netNi}{\netDequeueLab{\roleP}{\roleQ}{\stLab}{\stS}}{\netNii}$. %
            We must prove that clause~\ref{item:internal-bisim-lr} of \Cref{def:internal-bisim} is satisfied.\\
            By inversion of rule \textsc{Net-Par} in \Cref{def:network-semantics}
            we have %
            $\transition{\netPair{\roleP}{\stqT}}
                {\netDequeueLab{\roleP}{\roleQ}{\stLab}{\stS}}
                {\netPair{\roleP}{\stqTi}}$;
            moreover, by \textsc{Net-Par} in \Cref{def:network-semantics},
            $\netPair{\roleP}{\stqTi} \in \netNii$.
            Therefore, by \Cref{lem:type-mon-dequeuing} we have %
            $\transition{\mnetInstr{\netPair{\roleP}{\stqT}}}
                {\mnetDequeueLab{\roleP}{\roleQ}{\stLab}{\stS}}
                {\mnetInstr{\netPair{\roleP}{\stqTi}}}$.
            Hence, by rule \textsc{MNet-Par} in \Cref{def:monitored-network-semantics}
            and \Cref{def:consistent-mon-instr},
            we have $\transition{\mnetInstr{\netNi}}{\mnetDequeueLab{\roleP}{\roleQ}{\stLab}{\stS}}{\mnetInstr{\netNii}}$.
            Since by \eqref{eq:proof-net-mon-bisim:rel} we have $(\netNii, \mnetInstr{\netNii}) \in \mathop{\relR}$,
            we have satisfied clause~\ref{item:internal-bisim-lr} of \Cref{def:internal-bisim}.
        \item $\exists \roleP,\roleQ,\stLab,\stS,\mnetNiii:%
            \transition{\mnetInstr{\netNi}}{\mnetDequeueLab{\roleP}{\roleQ}{\stLab}{\stS}}{\mnetNiii}$. %
            We must prove that clause~\ref{item:internal-bisim-rl} of \Cref{def:internal-bisim} is satisfied.\\
            By inversion of rule \textsc{MNet-Par} in \Cref{def:monitored-network-semantics}
            and \Cref{lem:mon-type-dequeuing} %
            we have \thesisNewline
            $\transition{\mnetInstr{\netPair{\roleP}{\stqT}}}
                {\mnetDequeueLab{\roleP}{\roleQ}{\stLab}{\stS}}
                {\mnetInstr{\netPair{\roleP}{\stqTi}}}$
            such that $\transition{\netPair{\roleP}{\stqT}}{\netDequeueLab{\roleP}{\roleQ}{\stLab}{\stS}}{\netPair{\roleP}{\stqTi}}$.
            Therefore, by rule \textsc{MNet-Par} in \Cref{def:monitored-network-semantics},
            rule \textsc{Net-Par} in \Cref{def:network-semantics},
            and \Cref{def:consistent-mon-instr}, $\exists \netNii$ such that
            $\transition{\netNi}{\netDequeueLab{\roleP}{\roleQ}{\stLab}{\stS}}{\netNii} $
            and $\mnetInstr{\netPair{\roleP}{\stqTi}} \in \mnetNiii = \mnetInstr{\netNii}$.
            Since by \eqref{eq:proof-net-mon-bisim:rel} we have $(\netNii, \mnetInstr{\netNii}) \in \mathop{\relR}$,
            we have satisfied clause~\ref{item:internal-bisim-rl} of \Cref{def:internal-bisim}.
        \item $\exists \roleP,\roleQ,\stLab,\stS,\mnetNii:%
            \transition{\netNi}{\netCommLab{\roleP}{\roleQ}{\stLab}{\stS}}{\netNii}$. %
            We must prove that clause~\ref{item:internal-bisim-lr} of \Cref{def:internal-bisim} is satisfied.\\
            By inversion of the transition, we know that $\netNi$ contains
            $\netPair{\roleP}{\stqPair{\stT[\roleP]}{\stQ[\roleP]}}$
            such that $\transition{\stqPair{\stT[\roleP]}{\stQ[\roleP]}}{\stSendTag{\roleQ}{\stLab}{\stS}}{\stqPair{\stTi[\roleP]}{\stQ[\roleP]}}$
            (for some $\stTi[\roleP]$),
            with $\netPair{\roleP}{\stqPair{\stTi[\roleP]}{\stQ[\roleP]}}$ contained in $\netNii$.
            Correspondingly, by \Cref{def:monitored-network}, $\mnetInstr{\netNi}$ contains
            $\mnetPair{\roleP}{\mstPair{\stqPair{\stT[\roleP]}{\stQ[\roleP]}}{\enc{\stqPair{\stT[\roleP]}{\stQ[\roleP]}}}}$
            such that:
            \begin{align}
              &\label{eq:proof-net-mon-bisim:comm-send-tqm}%
              \transition{\mstPair{\stqPair{\stT[\roleP]}{\stQ[\roleP]}}{\enc{\stqPair{\stT[\roleP]}{\stQ[\roleP]}}}}
                         {\monLabYes{\stSendTag{\roleQ}{\stLab}{\stS}}}
                         {\mstPair{\stqPair{\stTi[\roleP]}{\stQ[\roleP]}}{\enc{\stqPair{\stTi[\roleP]}{\stQ[\roleP]}}}}
              &\text{(by \Cref{lem:type-send-mon-send})}
              \\
              &\label{eq:proof-net-mon-bisim:comm-send-tqm-net}%
              \transition{\mnetPair{\roleP}{\mstPair{\stqPair{\stT[\roleP]}{\stQ[\roleP]}}{\enc{\stqPair{\stT[\roleP]}{\stQ[\roleP]}}}}}
                         {\mnetIOLabYes{\roleP}{\stSendTag{\roleQ}{\stLab}{\stS}}}
                         {\mnetPair{\roleP}{\mstPair{\stqPair{\stTi[\roleP]}{\stQ[\roleP]}}{\enc{\stqPair{\stTi[\roleP]}{\stQ[\roleP]}}}}}
              &\text{(by \eqref{eq:proof-net-mon-bisim:comm-send-tqm}, \textsc{MNet-$\stExtTag$-Good} in \Cref{def:monitored-network})}
            \end{align}
            Moreover, again by inversion of the transition $\transition{\netNi}{\netCommLab{\roleP}{\roleQ}{\stLab}{\stS}}{\netNii}$,
            we know that $\netNi$ contains
            $\netPair{\roleQ}{\stqPair{\stT[\roleQ]}{\stQ[\roleQ]}}$
            such that $\transition{\stqPair{\stT[\roleQ]}{\stQ[\roleQ]}}
                                  {\stEnqueueTag{\roleP}{\stLab}{\stS}}
                                  {\stqPair{\stT[\roleQ]}{\stQCons{\stQ[\roleQ]}{\stQMsg{\roleP}{\stLab}{\stS}}}}$,
            with $\netPair{\roleP}{\stqPair{\stT[\roleQ]}{\stQCons{\stQ[\roleQ]}{\stQMsg{\roleP}{\stLab}{\stS}}}}$ contained in $\netNii$.
            Correspondingly, by \Cref{def:monitored-network}, $\mnetInstr{\netNi}$ contains
            $\mnetPair{\roleQ}{\mstPair{\stqPair{\stT[\roleQ]}{\stQ[\roleQ]}}{\enc{\stqPair{\stT[\roleQ]}{\stQ[\roleQ]}}}}$
            such that:
            \begin{align}
              &\label{eq:proof-net-mon-bisim:comm-recv-tqm-net}%
              \transition{\mnetPair{\roleQ}{\mstPair{\stqPair{\stT[\roleQ]}{\stQ[\roleQ]}}{\enc{\stqPair{\stT[\roleQ]}{\stQ[\roleQ]}}}}}
                         {\mnetIOLabYes{\roleQ}{\stEnqueueTag{\roleP}{\stLab}{\stS}}}
                         {\mstPair{\stqPair{\stT[\roleQ]}{\stQCons{\stQ[\roleQ]}{\stQMsg{\roleP}{\stLab}{\stS}}}}{\enc{\stqPair{\stT[\roleQ]}{\stQCons{\stQ[\roleQ]}{\stQMsg{\roleP}{\stLab}{\stS}}}}}}
              &\text{(by \Cref{lem:live-hd-net-enqueue-mon-enqueue})}
            \end{align}
            Hence, we obtain:
            \begin{align}
                &\label{eq:proof-net-mon-bisim:comm-mon-ni-nii}%
                \transition{\mnetInstr{\netNi}}{\mnetCommLab{\roleP}{\roleQ}{\stLab}{\stS}}{\mnetInstr{\netNii}}
                &\text{(by \eqref{eq:proof-net-mon-bisim:comm-send-tqm-net}, \eqref{eq:proof-net-mon-bisim:comm-recv-tqm-net},
                        and \textsc{MNet-Comm} in \Cref{def:monitored-network-semantics})}
                \\
                &\nonumber%
                (\netNii, \mnetInstr{\netNii}) \in \mathop{\relR}
                &\text{(by \eqref{eq:proof-net-mon-bisim:comm-mon-ni-nii} and \eqref{eq:proof-net-mon-bisim:rel})}
            \end{align}
            as required by clause~\ref{item:internal-bisim-lr} of \Cref{def:internal-bisim}.

        \item $\exists \roleP,\roleQ,\stLab,\stS,\mnetNi'i:%
            \transition{\mnetInstr{\netNi}}{\mnetCommLab{\roleP}{\roleQ}{\stLab}{\stS}}{\mnetNii}$. %
            We must prove that clause~\ref{item:internal-bisim-rl} of \Cref{def:internal-bisim} is satisfied.\\
            We infer the shape of $\netNi$ (which involves a communication of a
            message with label $\stLab$ and payload type $\stS$ from an internal
            choice in $\roleP$ towards $\roleQ$, similarly to the previous case)
            and then show that there is an unmonitored network %
            $\netNii$ such that $\transition{\netNi}{\netCommLab{\roleP}{\roleQ}{\stLab}{\stS}}{\netNii}$
            (i.e., the unmonitored $\netNi$ can perform the same communication allowed by $\mnetInstr{\netNi}$)
            and $\mnetNiii = \mnetInstr{\netNii}$. %
            Therefore, since by \eqref{eq:proof-net-mon-bisim:rel} we have $(\netNii, \mnetInstr{\netNii}) \in \mathop{\relR}$,
            we have satisfied clause~\ref{item:internal-bisim-rl} of \Cref{def:internal-bisim}.
    \end{itemize}

    We thus proven that, for any output-live network $\netN$, there exists a relation
    $\relR$ (\Cref{eq:proof-net-mon-bisim:rel}) which is an internal
    bisimulation (by \Cref{def:internal-bisim}); moreover, we have $(\netN,
    \mnetInstr{\netN}) \in {\relR}$ (by \Cref{eq:proof-net-mon-bisim:rel}). %
    Therefore, we conclude that, for any output-live network $\netN$, %
    we have $\netN \tbisim \mnetInstr{\netN}$ (by \Cref{def:internal-bisim}).
    \qed
\end{proof}
\section{Evaluation: Description of Examples}
\label{appendix:ExampleDesc}

\subsection{Book review}
\textit{Participants: Client, Info, Review, Request, Details}\\
Models a network composed of several microservices, based on the example
from Istio~\cite{IstioBookinfo}. A client queries info about a book through
the \textit{Info} service, which then again queries the other services to
obtain the info.
\begin{lstlisting}[language=scala]
val lst_BookClient =
    Rec("Start",
        IntChoice(Map((a_BookInfo, l_Request) ->
        ExtChoice(a_BookInfo, Map(l_Response ->
        Var("Start"))))))

val lst_BookInfo =
    Rec("Start",
        ExtChoice(a_BookClient, Map(l_Request ->
        IntChoice(Map((a_BookReview, l_ReviewRequest) ->
        IntChoice(Map((a_BookDetails, l_DetailRequest) ->
        ExtChoice(a_BookReview, Map(l_ReviewResponse ->
        ExtChoice(a_BookDetails, Map(l_DetailResponse ->
        IntChoice(Map((a_BookClient, l_Response) ->
        Var("Start"))))))))))))))

val lst_BookReview =
    Rec("Start",
        ExtChoice(a_BookInfo, Map(l_ReviewRequest ->
        IntChoice(Map((a_BookRatings, l_RatingsRequest) ->
        ExtChoice(a_BookRatings, Map(l_RatingsResponse ->
        IntChoice(Map((a_BookInfo, l_ReviewResponse) ->
        Var("Start"))))))))))

val lst_BookDetails =
    Rec("Start",
        ExtChoice(a_BookInfo, Map(l_DetailRequest ->
        IntChoice(Map((a_BookInfo, l_DetailResponse) ->
        Var("Start"))))))

val lst_BookRatings =
    Rec("Start",
        ExtChoice(a_BookReview, Map(l_RatingsRequest ->
        IntChoice(Map((a_BookReview, l_RatingsResponse) ->
        Var("Start"))))))
\end{lstlisting}

\subsection{Store}
\textit{Participants: Client, HTTP Gateway, Inventory, Payment,
		Shipping, External bank, External mailserver}\\
A network of several microservices that implement a store, as well as
external participants. The client generates, confirms or cancels an order
by sending a request to the \textit{HTTP Gateway}, which in turn manages
the order by sending requests to the additional services.
\begin{lstlisting}[language=scala]
val lst_Client =
    IntChoice(Map((a_HTTP, l_PlaceOrder) ->
    ExtChoice(a_HTTP, Map(l_SKUDetails ->
    Rec("decideOnOrder",
        IntChoice(Map(
            (a_HTTP, l_ConfirmOrder) ->
                ExtChoice(a_HTTP, Map(
                    l_ValidCC ->
                        End(),
                    l_InvalidCC ->
                        Var("decideOnOrder"))),
            (a_HTTP, l_CancelOrder) ->
                ExtChoice(a_HTTP, Map(l_Cancel ->
                End())))))))))

val lst_HTTP =
    ExtChoice(a_Client, Map(l_PlaceOrder ->
    IntChoice(Map((a_Inventory, l_GetSKUDetails) ->
    ExtChoice(a_Inventory, Map(l_SKUDetails ->
    IntChoice(Map((a_Client, l_SKUDetails) ->
    Rec("decideOnOrder",
        ExtChoice(a_Client, Map(
            l_ConfirmOrder ->
                IntChoice(Map((a_Payment, l_VerifyCC) ->
                ExtChoice(a_Payment, Map(
                    l_ValidCC ->
                        IntChoice(Map((a_Payment, l_MakeCharge) ->
                        IntChoice(Map((a_Inventory, l_MakeOrder) ->
                        IntChoice(Map((a_Client, l_ValidCC) ->
                        End())))))),
                    l_InvalidCC ->
                        IntChoice(Map((a_Client, l_InvalidCC) ->
                        Var("decideOnOrder"))))))),
            l_CancelOrder ->
                IntChoice(Map((a_Client, l_Cancel) ->
                End())))))))))))))

val lst_Inventory =
    ExtChoice(a_HTTP, Map(l_GetSKUDetails ->
    IntChoice(Map((a_HTTP, l_SKUDetails) ->
    ExtChoice(a_HTTP, Map(l_MakeOrder ->
    IntChoice(Map((a_Shipping, l_ShipOrder) ->
    End()))))))))

val lst_Payment =
    Rec("checkCC",
        ExtChoice(a_HTTP, Map(l_VerifyCC ->
        IntChoice(Map(
            (a_HTTP, l_ValidCC) ->
                ExtChoice(a_HTTP, Map(l_MakeCharge ->
                IntChoice(Map((a_Bank, l_MakeCharge) ->
                End())))),
            (a_HTTP, l_InvalidCC) ->
                Var("checkCC"))))))

val lst_Shipping =
    ExtChoice(a_Inventory, Map(l_ShipOrder ->
    IntChoice(Map((a_Email, l_SendEmail) ->
    ExtChoice(a_Email, Map(l_Success ->
    End()))))))

val lst_Bank =
    ExtChoice(a_Payment, Map(l_MakeCharge ->
    End()))

val lst_Email =
    ExtChoice(a_Shipping, Map(l_SendEmail ->
    IntChoice(Map((a_Shipping, l_Success) ->
    End()))))
\end{lstlisting}

\subsection{VPN}
\textit{Participants: Authenticator, Client A, Client B,
		Client C}\\
A VPN network that allows communication between clients. All clients
must successfully authenticate with the \textit{Authenticator} server before
being allowed to communicate with each other.
\begin{lstlisting}[language=scala]
val lst_Auth =
    Rec("acceptOrDeny",
        ExtChoice(a_ClientA, Map(l_Auth ->
        ExtChoice(a_ClientB, Map(l_Auth ->
        ExtChoice(a_ClientC, Map(l_Auth ->
        IntChoice(Map(
            (a_ClientA, l_Accept) ->
                IntChoice(Map((a_ClientB, l_Accept) ->
                IntChoice(Map((a_ClientC, l_Accept) ->
                ExtChoice(a_ClientA, Map(l_Terminate ->
                ExtChoice(a_ClientB, Map(l_Terminate ->
                ExtChoice(a_ClientC, Map(l_Terminate ->
                End())))))))))),
            (a_ClientA, l_Deny) ->
                IntChoice(Map((a_ClientB, l_Deny) ->
                IntChoice(Map((a_ClientC, l_Deny) ->
                Var("acceptOrDeny"))))))))))))))

val lst_ClientA =
    Rec("tryAuth",
        IntChoice(Map((a_Auth, l_Auth) ->
        ExtChoice(a_Auth, Map(
            l_Deny ->
                Var("tryAuth"),
            l_Accept ->
                Rec("requestB",
                    IntChoice(Map(
                        (a_ClientB, l_ClientRequest) ->
                            ExtChoice(a_ClientB, Map(l_ClientResponse ->
                            Var("requestB"))),
                        (a_ClientB, l_Terminate) ->
                            Rec("requestC",
                                IntChoice(Map(
                                    (a_ClientC, l_ClientRequest) ->
                                        ExtChoice(a_ClientC, Map(l_ClientResponse ->
                                        Var("requestC"))),
                                    (a_ClientC, l_Terminate) ->
                                        IntChoice(Map((a_Auth, l_Terminate) ->
                                        End())))))))))))))

val lst_ClientB =
    Rec("tryAuth",
        IntChoice(Map((a_Auth, l_Auth) ->
        ExtChoice(a_Auth, Map(
            l_Deny ->
                Var("tryAuth"),
            l_Accept ->
                Rec("responseA",
                    ExtChoice(a_ClientA, Map(
                        l_ClientRequest ->
                            IntChoice(Map((a_ClientA, l_ClientResponse) ->
                            Var("responseA"))),
                        l_Terminate ->
                            IntChoice(Map((a_Auth, l_Terminate) ->
                            End()))))))))))

val lst_ClientC =
    Rec("tryAuth",
        IntChoice(Map((a_Auth, l_Auth) ->
        ExtChoice(a_Auth, Map(
            l_Deny ->
                Var("tryAuth"),
            l_Accept ->
                Rec("responseA",
                    ExtChoice(a_ClientA, Map(
                        l_ClientRequest ->
                            IntChoice(Map((a_ClientA, l_ClientResponse) ->
                            Var("responseA"))),
                        l_Terminate ->
                            IntChoice(Map((a_Auth, l_Terminate) ->
                            End()))))))))))
\end{lstlisting}

\subsection{Stateful Firewall}
	\textit{Participants: Internal Client, External Client}\\
A simple stateful firewall. It first accepts only outgoing
communication from the \textit{Internal Client} to the
\textit{External Client}. Then, once the outgoing transmissions are forwarded,
the firewall allows the external client to transmit back to the internal client.
\begin{lstlisting}[language=scala]
val lst_ClientInt =
    Rec("Outgoing",
        IntChoice(Map(
            (a_ClientExt, l_Request) ->
                ExtChoice(a_ClientExt, Map(l_Response ->
                Var("Outgoing"))),
            (a_ClientExt, l_Terminate) ->
                Rec("Incoming",
                    ExtChoice(a_ClientExt, Map(
                        l_Request ->
                            IntChoice(Map((a_ClientExt, l_Response) ->
                            Var("Incoming"))),
                        l_Terminate ->
                            End()))))))

val lst_ClientExt =
    Rec("Outgoing",
        ExtChoice(a_ClientInt, Map(
            l_Request ->
                IntChoice(Map((a_ClientInt, l_Response) ->
                Var("Outgoing"))),
            l_Terminate ->
                Rec("Incoming",
                    IntChoice(Map(
                        (a_ClientInt, l_Request) ->
                            ExtChoice(a_ClientInt, Map(l_Response ->
                            Var("Incoming"))),
                        (a_ClientInt, l_Terminate) ->
                            End()))))))
\end{lstlisting}

\subsection{DNS}
\textit{Participants: Client, Local DNS, Root DNS, TLD DNS,
		Authoritative DNS}\\
A DNS resolver, where the \textit{Local DNS} performs an iterative
lookup after the \textit{Client} requests an address.
\begin{lstlisting}[language=scala]
val lst_DNSClient =
    IntChoice(Map((a_LocalDNS, l_RequestIP) ->
    ExtChoice(a_LocalDNS, Map(l_ResponseIP ->
    End()))))

val lst_LocalDNS =
    ExtChoice(a_DNSClient, Map(l_RequestIP ->
    IntChoice(Map((a_RootDNS, l_RequestRoot) ->
    ExtChoice(a_RootDNS, Map(l_ResponseRoot ->
    IntChoice(Map((a_TLDDNS, l_RequestTLD) ->
    ExtChoice(a_TLDDNS, Map(l_ResponseTLD ->
    IntChoice(Map((a_AuthDNS, l_RequestAuth) ->
    ExtChoice(a_AuthDNS, Map(l_ResponseAuth ->
    IntChoice(Map((a_DNSClient, l_ResponseIP) ->
    End()))))))))))))))))

val lst_RootDNS =
    ExtChoice(a_LocalDNS, Map(l_RequestRoot ->
    IntChoice(Map((a_LocalDNS, l_ResponseRoot) ->
    End()))))

val lst_TLDDNS =
    ExtChoice(a_LocalDNS, Map(l_RequestTLD ->
    IntChoice(Map((a_LocalDNS, l_ResponseTLD) ->
    End()))))

val lst_AuthDNS =
    ExtChoice(a_LocalDNS, Map(l_RequestAuth ->
    IntChoice(Map((a_LocalDNS, l_ResponseAuth) ->
    End()))))
\end{lstlisting}

\subsection{Two-buyer auction}
	\textit{Participants: Auction, Buyer A, Buyer B}\\
A two-buyer auction protocol where two buyers, \textit{Buyer A} and
\textit{Buyer B}, send repeated bids to the \textit{Auction} server, which
eventually decides on a winner. The winner and the auction then
exchange the payment and the item.
\begin{lstlisting}[language=scala]
val lst_Auction =
    Rec("GetBids",
        ExtChoice(a_BuyerA, Map(l_Bid ->
        ExtChoice(a_BuyerB, Map(l_Bid ->
        IntChoice(Map(
            (a_BuyerA, l_Resend) ->
                IntChoice(Map((a_BuyerB, l_Resend) ->
                Var("GetBids"))),
            (a_BuyerA, l_Winner) ->
                ExtChoice(a_BuyerA, Map(l_Pay ->
                IntChoice(Map((a_BuyerA, l_SendItem) ->
                End())))),
            (a_BuyerB, l_Winner) ->
                ExtChoice(a_BuyerB, Map(l_Pay ->
                IntChoice(Map((a_BuyerB, l_SendItem) ->
                End())))))))))))

val lst_BuyerA =
    Rec("SendBids",
        IntChoice(Map((a_Auction, l_Bid) ->
        ExtChoice(a_Auction, Map(
            l_Resend ->
                Var("SendBids"),
            l_Winner ->
                IntChoice(Map((a_Auction, l_Pay) ->
                ExtChoice(a_Auction, Map(l_SendItem ->
                End())))))))))

val lst_BuyerB =
    Rec("SendBids",
        IntChoice(Map((a_Auction, l_Bid) ->
        ExtChoice(a_Auction, Map(
            l_Resend ->
                Var("SendBids"),
            l_Winner ->
                IntChoice(Map((a_Auction, l_Pay) ->
                ExtChoice(a_Auction, Map(l_SendItem ->
                End())))))))))
\end{lstlisting}

\subsection{CDN}
\textit{Participants User, Local DNS, Internal Server,
    External Server}\\
A content distribution network where a \textit{Local DNS} communicates
with an \textit{Internal Server} and \textit{External Server} to obtain the
address which the \textit{User} is requesting.
\begin{lstlisting}[language=scala]
val lst_User =
    IntChoice(Map((a_LocalDNS, l_DNSQuery) ->
    ExtChoice(a_LocalDNS, Map(l_IPAddr ->
    End()))))

val lst_LocalDNS =
    ExtChoice(a_User, Map(l_DNSQuery ->
    IntChoice(Map((a_IntServer, l_AuthQuery) ->
    ExtChoice(a_IntServer, Map(l_Hostname ->
    IntChoice(Map((a_ExtServer, l_AuthQuery) ->
    ExtChoice(a_ExtServer, Map(l_IPAddr ->
    IntChoice(Map((a_User, l_IPAddr) ->
    End()))))))))))))

val lst_IntAuthDNS =
    ExtChoice(a_LocalDNS, Map(l_AuthQuery ->
    IntChoice(Map((a_LocalDNS, l_Hostname) ->
    End()))))

val lst_ExtAuthDNS =
    ExtChoice(a_LocalDNS, Map(l_AuthQuery ->
    IntChoice(Map((a_LocalDNS, l_IPAddr) ->
    End()))))
\end{lstlisting}

\subsection{SIP}
\textit{Participants: Client A, Client B, Proxy}\\
An implementation of the Session Initiation Protocol where two clients,
\textit{Client A} and \textit{Client B}, communicate over a
\textit{Proxy} server.
\begin{lstlisting}[language=scala]
val lst_SIPClientA =
    IntChoice(Map((a_SIPProxy, l_INVITE) ->
    ExtChoice(a_SIPProxy, Map(
        l_Trying ->
            ExtChoice(a_SIPProxy, Map(
                l_Ringing ->
                    ExtChoice(a_SIPProxy, Map(
                        l_OK ->
                            End(),
                        l_Error ->
                            End())),
                l_Error ->
                    End())),
        l_Error ->
            End()))))

val lst_SIPClientB =
    ExtChoice(a_SIPProxy, Map(
        l_INVITE -> IntChoice(Map(
            (a_SIPProxy, l_Ringing) ->
                IntChoice(Map(
                    (a_SIPProxy, l_OK) ->
                        End(),
                    (a_SIPProxy, l_Error) ->
                        End())),
            (a_SIPProxy, l_Error) ->
                End()))))

val lst_SIPProxy =
    ExtChoice(a_SIPClientA, Map(l_INVITE ->
    IntChoice(Map(
        (a_SIPClientB, l_INVITE) ->
            IntChoice(Map((a_SIPClientA, l_Trying) ->
            ExtChoice(a_SIPClientB, Map(
                l_Ringing ->
                    IntChoice(Map((a_SIPClientA, l_Ringing) ->
                    ExtChoice(a_SIPClientB, Map(
                        l_OK ->
                            IntChoice(Map((a_SIPClientA, l_OK) ->
                            End())),
                        l_Error ->
                            IntChoice(Map((a_SIPClientA, l_Error) ->
                            End())))))),
                l_Error ->
                    IntChoice(Map((a_SIPClientA, l_Error) ->
                    End())))))),
        (a_SIPClientA, l_Error) ->
            End()))))
\end{lstlisting}

\subsection{POP3}
\textit{Participants: Client, Server}\\
An implementation of a POP3 server, where the \textit{Client} logs in and sends a number of different queries to the \textit{Server} before logging out.
\begin{lstlisting}[language=scala]
val lst_POPClient =
    Rec("TryLogin",
        IntChoice(Map((a_POPServer, l_Username) ->
        ExtChoice(a_POPServer, Map(
            l_ERR ->
                Var("TryLogin"),
            l_OK  ->
                IntChoice(Map((a_POPServer, l_Password) ->
                ExtChoice(a_POPServer, Map(
                    l_ERR ->
                        Var("TryLogin"),
                    l_OK ->
                        Rec("ClientQuery",
                            IntChoice(Map(
                                (a_POPServer, l_ListMsgs) ->
                                    ExtChoice(a_POPServer, Map(l_Messages ->
                                    Var("ClientQuery"))),
                                (a_POPServer, l_Retransmit) ->
                                    ExtChoice(a_POPServer, Map(
                                        l_OK ->
                                            Var("ClientQuery"),
                                        l_ERR ->
                                            Var("ClientQuery"))),
                                (a_POPServer, l_Delete) ->
                                    ExtChoice(a_POPServer, Map(
                                        l_OK ->
                                            Var("ClientQuery"),
                                        l_ERR ->
                                            Var("ClientQuery"))),
                                (a_POPServer, l_NoOp) ->
                                    ExtChoice(a_POPServer, Map(l_OK ->
                                    Var("ClientQuery"))),
                                (a_POPServer, l_Quit) ->
                                    End()))))))))))))

val lst_POPServer =
    Rec("TryLogin",
        ExtChoice(a_POPClient, Map(l_Username ->
        IntChoice(Map(
            (a_POPClient, l_ERR) ->
                Var("TryLogin"),
            (a_POPClient, l_OK) ->
                ExtChoice(a_POPClient, Map(l_Password ->
                IntChoice(Map(
                    (a_POPClient, l_ERR) ->
                        Var("TryLogin"),
                    (a_POPClient, l_OK)  ->
                        Rec("ClientQuery",
                            ExtChoice(a_POPClient, Map(
                                l_ListMsgs ->
                                    IntChoice(Map((a_POPClient, l_Messages) ->
                                    Var("ClientQuery"))),
                                l_Retransmit ->
                                    IntChoice(Map(
                                        (a_POPClient, l_OK) ->
                                            Var("ClientQuery"),
                                        (a_POPClient, l_ERR) ->
                                            Var("ClientQuery"))),
                                l_Delete ->
                                    IntChoice(Map(
                                        (a_POPClient, l_OK) ->
                                            Var("ClientQuery"),
                                        (a_POPClient, l_ERR) ->
                                            Var("ClientQuery"))),
                                l_NoOp ->
                                    IntChoice(Map((a_POPClient, l_OK) ->
                                    Var("ClientQuery"))),
                                l_Quit ->
                                    End()))))))))))))
\end{lstlisting}

\subsection{Turn-Based Game}
\textit{Participants: Coordinator, Game Server, Player A,
    Player B}\\
An example of a simple turn-based game. Two players search for a game
by subscribing to a \textit{Coordinator}, and are then joined in a game
managed by the \textit{Game Server}. Each player takes their turn until
the \textit{Game Server} decides on a winner.
\begin{lstlisting}[language=scala]
val lst_Coordinator =
    ExtChoice(a_Player1, Map(l_LookForGame ->
    ExtChoice(a_Player2, Map(l_LookForGame ->
    IntChoice(Map((a_Player1, l_GameFound) ->
    IntChoice(Map((a_Player2, l_GameFound) ->
    IntChoice(Map((a_GameServer, l_GameFound) ->
    End()))))))))))

val lst_GameServer =
    ExtChoice(a_Coordinator, Map(l_GameFound ->
    Rec("GameRound",
        IntChoice(Map(
            (a_Player1, l_TurnStart) ->
                ExtChoice(a_Player1, Map(l_TakeAction ->
                IntChoice(Map((a_Player2, l_TurnStart) ->
                ExtChoice(a_Player2, Map(l_TakeAction ->
                Var("GameRound"))))))),
            (a_Player1, l_Victory) ->
                IntChoice(Map((a_Player2, l_Defeat) ->
                End())),
            (a_Player1, l_Defeat) ->
                IntChoice(Map((a_Player2, l_Victory) ->
                End())))))))

val lst_Player =
    IntChoice(Map((a_Coordinator, l_LookForGame) ->
    ExtChoice(a_Coordinator, Map(l_GameFound ->
    Rec("GameRound",
        ExtChoice(a_GameServer, Map(
            l_TurnStart ->
                IntChoice(Map((a_GameServer, l_TakeAction) ->
                Var("GameRound"))),
            l_Victory ->
                End(),
            l_Defeat ->
                End())))))))
\end{lstlisting} %
\section{Evaluation: Bar plots}
\label{appendix:BarPlots}

\noindent
\begin{figure}[H]
    \centering
    \begin{subfigure}{.45\textwidth}
        \includegraphics[width=\textwidth]{plots/bookreview\_barplot.pdf}
        \caption{Packets observed for the BookInfo example over different network configurations.}
        \label{fig:appendix-bookreview-barplot}
    \end{subfigure}%
    \hspace{5pt}%
    \begin{subfigure}{.45\textwidth}
        \includegraphics[width=\textwidth]{plots/vpn\_barplot.pdf}
        \caption{Packets observed for the VPN example over different network configurations.}
        \label{fig:appendix-vpn-barplot}
    \end{subfigure}
\end{figure}%
\begin{figure}[H]
    \begin{subfigure}{.45\textwidth}
        \includegraphics[width=\textwidth]{plots/store\_barplot.pdf}
        \caption{Packets observed for the store example over different network configurations.}
        \label{fig:store-barplot}
    \end{subfigure}%
    \hspace{5pt}%
    \begin{subfigure}{.45\textwidth}
        \includegraphics[width=\textwidth]{plots/firewall\_barplot.pdf}
        \caption{Packets observed for the firewall example over different network configurations.}
        \label{fig:firewall-barplot}
    \end{subfigure}
\end{figure}%
\begin{figure}[H]
    \begin{subfigure}{.45\textwidth}
        \includegraphics[width=\textwidth]{plots/dns\_barplot.pdf}
        \caption{Packets observed for the DNS example over different network configurations.}
        \label{fig:dns-barplot}
    \end{subfigure}%
    \hspace{5pt}
    \begin{subfigure}{.45\textwidth}
        \includegraphics[width=\textwidth]{plots/auction\_barplot.pdf}
        \caption{Packets observed for the auction example over different network configurations.}
        \label{fig:auction-barplot}
    \end{subfigure}
\end{figure}%
\begin{figure}[H]
    \begin{subfigure}{.45\textwidth}
        \includegraphics[width=\textwidth]{plots/cdn\_barplot.pdf}
        \caption{Packets observed for the CDN example over different network configurations.}
        \label{fig:cdn-barplot}
    \end{subfigure}%
    \hspace{5pt}%
    \begin{subfigure}{.45\textwidth}
        \includegraphics[width=\textwidth]{plots/sip\_barplot.pdf}
        \caption{Packets observed for the SIP example over different network configurations.}
        \label{fig:sip-barplot}
    \end{subfigure}
\end{figure}%
\begin{figure}[H]
    \begin{subfigure}{.45\textwidth}
        \includegraphics[width=\textwidth]{plots/pop3\_barplot.pdf}
        \caption{Packets observed for the POP3 example over different network configurations.}
        \label{fig:pop3-barplot}
    \end{subfigure}%
    \hspace{5pt}
    \begin{subfigure}{.45\textwidth}
        \includegraphics[width=\textwidth]{plots/game\_barplot.pdf}
        \caption{Packets observed for the multiplayer game example over different network configurations.}
        \label{fig:game-barplot}
    \end{subfigure}
\end{figure}
 
\end{document}